\documentclass{article}
\usepackage{graphicx} 
\usepackage{fullpage}
\usepackage{amssymb}
\usepackage{amsmath}
\usepackage{amsthm}
\usepackage{thmtools}
\usepackage[hidelinks, colorlinks = true]{hyperref}
\usepackage{xcolor}
\usepackage{enumitem}
\usepackage{cleveref}
\usepackage{quantikz}
\usepackage{physics}
\usepackage{booktabs}
\usepackage{array}
\usepackage{makecell}
\usepackage{setspace}

\newcommand{\NC}{\textsf{NC}}
\newcommand{\AC}{\textsf{AC}}
\newcommand{\TC}{\textsf{TC}}

\newcommand{\UQAC}{\textsf{Q$_\textsf{U}$AC}}
\newcommand{\UQNC}{\textsf{Q$_\textsf{U}$NC}}
\newcommand{\UQTC}{\textsf{Q$_\textsf{U}$TC}}

\newcommand{\HQAC}{\textsf{Q$_\textsf{R}$AC}}
\newcommand{\HQNC}{\textsf{Q$_\textsf{R}$NC}}
\newcommand{\HQTC}{\textsf{Q$_\textsf{R}$TC}}

\newcommand{\QAC}{\textsf{QAC}}
\newcommand{\QNC}{\textsf{QNC}}
\newcommand{\QTC}{\textsf{QTC}}

\newcommand{\QACZ}{\textsf{QAC}$^0$}

\newcommand{\poly}{\textnormal{poly}}
\newcommand{\eps}{\varepsilon}

\newcommand{\Rz}{R_z(\theta)}

\providecommand{\Parity}{\operatorname{Parity}}

\providecommand{\poly}{\operatorname{poly}}
\providecommand{\depth}{\operatorname{depth}}

\newcommand{\BQP}{\textsf{BQP}}

\hypersetup{
    colorlinks=true,
    linkcolor=magenta,
    filecolor=magenta,
    urlcolor=cyan,
    citecolor=cyan
}

\newcommand{\thresh}{f_M^m}
\newcommand{\x}{\vec{x}}
\newcommand{\M}{\vec{M}}
\newcommand{\val}{\operatorname{int}}

\newcommand{\FANOUT}{\textnormal{Fan-Out}}
\newcommand{\GTof}{\operatorname{GToffoli}}
\newcommand{\cnot}{\textnormal{CNOT}}
\newcommand{\toff}{\textnormal{Toffoli}}
\newcommand{\be}{\mathcal{B}}

\newtheorem{theorem}{Theorem}
\newtheorem{lemma}{Lemma}
\newtheorem{fact}{Fact}
\newtheorem{definition}{Definition}
\newtheorem{corollary}{Corollary}
\newtheorem{proposition}{Proposition}
\newtheorem{open}{Open Question}

\title{Depth-Optimal Quantum Compilation}
\author{Francisca Vasconcelos \vspace{0.05in}\\  UC Berkeley \vspace{0.05in}\\ \small \url{francisca@berkeley.edu}}
\date{}

\begin{document}

\maketitle

\begin{abstract}
We achieve the first constant-depth circuit for arbitrary single-qubit gate synthesis. Unlike prior approaches, the construction is fully unitary and requires no pre-supplied catalyst. For any constant $\delta>0$, it $\varepsilon$-approximates an arbitrary single-qubit gate using $O(\log^{1+\delta}(1/\varepsilon))$ clean ancillae, Hadamard and $T$ single-qubit gates, $O(\log(1/\varepsilon))$-width generalized Toffoli gates, and sublogarithmic-width Fan-Out gates. We further eliminate Fan-Out entirely, showing that Hadamard, $T$, and generalized Toffoli gates alone suffice for constant-depth synthesis. When restricted to the standard bounded-width gate model, our construction has depth $O(\log\log(1/\varepsilon))$, and we prove a matching $\Omega(\log\log(1/\varepsilon))$-depth lower bound. Overall, we establish that $\Theta(\log\log(1/\varepsilon))$-depth is unavoidable with only bounded-width gates, yet allowing even logarithmic-width multi-qubit gates suffices to achieve constant-depth synthesis.

These results also reveal new structure in shallow quantum circuit complexity. We give a depth-preserving real simulation of bounded-error decision computation, showing that every depth-$d$ $\QAC$ circuit can be simulated in depth $O(d)$ using only Hadamard, $X$, and generalized Toffoli gates. Thus arbitrary single-qubit rotations and complex amplitudes do not increase the bounded-error decision power of $\QAC$, even at constant depth. In particular, this reduces the long-standing conjecture $\Parity\notin\QAC^0$ to proving a Parity lower bound against circuits consisting only of Hadamard, $X$, and generalized Toffoli gates. More generally, this real normal form exposes a direct correspondence between the standard shallow-depth quantum circuit hierarchy and a hierarchy of Forrelation circuits with restricted oracle families.
\end{abstract}
\thispagestyle{empty}
\newpage
\begin{spacing}{0.95}
\small
\tableofcontents
\end{spacing}
\thispagestyle{empty}
\clearpage
\pagenumbering{arabic}

\section{Introduction}
\label{sec:introduction}

Quantum computation is remarkably robust to gate-set restrictions. The celebrated Solovay--Kitaev theorem established that arbitrary quantum gates can be efficiently approximated to any desired precision over a fixed, finite universal gate set~\cite{kitaev1997,dawson2006the}. Moreover, Bernstein and Vazirani established that quantum computation can be simulated using only real amplitudes~\cite{berstein1993quantum}. Building on this observation, Shi proved that the real gate set consisting solely of Hadamard and Toffoli gates is universal, and Aharonov later gave an explicit gate-by-gate simulation construction~\cite{shi2003both,aharonov2003simple}. Thus, at the level of polynomial-time quantum computation, neither arbitrary single-qubit gates nor complex amplitudes are essential resources.

These fundamental robustness results, however, do not preserve circuit depth. For standard gate-by-gate compilation, preserving the behavior of a polynomial-size bounded-error circuit requires approximating each gate to inverse-polynomial precision, $\varepsilon=1/\poly(n)$. Solovay--Kitaev compilation, therefore, incurs a $\poly\log n$ depth overhead per layer of the original circuit. Meanwhile, Aharonov's Hadamard--Toffoli real simulation encodes the real and imaginary parts of the state using a single shared ancilla. Consequently, parallel non-real gates in the original circuit must interact sequentially with this ancilla in the gate-by-gate simulation, potentially serializing a depth-one layer into linear depth.

These depth overheads matter whenever circuit depth is itself a resource. In practical quantum implementations, circuit depth equates to sequential runtime, which often proves to be a major bottleneck and motivates the use of additional ancillae to increase parallelism~\cite{beverland2022}. Such space--time tradeoffs are especially natural for reconfigurable neutral-atom arrays, which combine large qubit counts with highly parallel entangling operations~\cite{evered2023high}. The same overheads also matter for circuit complexity. Although they are negligible at the coarse-grained level of $\BQP$, they can transform a constant-depth circuit into one of polylogarithmic- or even linear-depth, thereby failing to preserve shallow-depth complexity classes. Moreover, depth-preserving real simulation can expose a substantially simpler underlying circuit structure, potentially making questions of computational power and circuit lower-bounds more tractable.  Together, these practical and complexity-theoretic considerations motivate the central question of this work:
\begin{quote}
\centering
\textit{Can quantum computation be restricted to a fixed gate set, without increasing circuit depth?}
\end{quote}
\noindent We show that the answer depends on the available multi-qubit primitives. With $O(\log(1/\varepsilon))$-width gates, we give an explicit constant-depth procedure for arbitrary single-qubit gate synthesis, yielding fully depth-preserving universal gate-set compilation. We further show that real simulation can likewise be made depth-preserving using larger-width multi-qubit gates. However, we show that constant depth is impossible to attain in the standard bounded-width gate model. There, we achieve $\Theta(\log\log(1/\varepsilon))$-depth synthesis and prove this dependence to be optimal. Overall, some depth-overhead is unavoidable without superconstant-width multi-qubit gates, while logarithmic-width gates suffice for constant-depth single-qubit gate synthesis.

\paragraph{Shallow-Depth Single-Qubit Gate Synthesis.} We first make this width--depth tradeoff precise for gate-set compilation. For every constant $\delta\in(0,1)$, we give a constant-depth circuit for $\varepsilon$-approximate single-qubit gate synthesis using $O_\delta(\log^{1+\delta}(1/\varepsilon))$ ancillae. Specifically, we demonstrate that any single-qubit unitary can be synthesized in constant depth using only $H$ and $T$ gates together with $O(\log(1/\varepsilon))$-width generalized Toffoli gates and sublogarithmic-width Fan-Out. Unlike recent catalytic approaches~\cite{kim2026catalytic,kim2025clifford}, our construction is fully end-to-end and requires no pre-supplied catalyst.\footnote{In particular, Kim--Laakkonen obtain constant $T$-depth online synthesis using an $O(\log(1/\varepsilon))$-qubit catalyst. However, they do not give an end-to-end constant-depth unitary circuit for synthesizing this catalyst state.} To our knowledge, our construction therefore gives the first \emph{end-to-end constant-depth procedure for arbitrary single-qubit gate synthesis}.

The same construction also reveals a sharp width--depth tradeoff. When the multi-qubit operations are compiled into the standard bounded-width gate set $\{H,T,\mathrm{CNOT}\}$, the depth becomes $O(\log\log(1/\varepsilon))$. We prove a matching $\Omega(\log\log(1/\varepsilon))$-depth lower bound, establishing optimality of this depth-dependence in the bounded-width gate model. Kitaev, Shen, and Vyalyi gave an earlier phase-kickback construction with end-to-end polylogarithmic depth~\cite{kitaev2002classical}, which Kim~\cite{kim2026catalytic} recently observed can be implemented end-to-end in $O(\log\log(1/\varepsilon))$ depth. Our construction achieves the same asymptotic depth through a substantially different and arguably simpler route, based on block encoding and oblivious amplitude amplification, while making the width--depth tradeoff explicit.

Finally, using techniques from shallow-depth quantum circuit complexity, we eliminate the need for Fan-Out altogether, at the cost of allowing wider generalized Toffoli gates. This yields constant-depth synthesis with only $H$, $T$, and generalized Toffoli gates (a restriction of the \QAC$^0$ circuit family).

\paragraph{Shallow-Depth Real Simulation.}
We next consider whether real simulation can similarly be achieved without sacrificing depth. While Shi and Aharonov showed that Hadamard and Toffoli gates suffice for real simulation of quantum computation~\cite{shi2003both,aharonov2003simple}, we show that Hadamard and \emph{generalized Toffoli} gates suffice for \emph{depth-preserving} real simulation of quantum computation.

The key observation is that, for decision computation, it suffices to preserve the final acceptance probability rather than the full complex state throughout the simulation. This allows the phase information associated with different gates to be handled in parallel and, together with our fixed-gate compilation, avoids the serialization inherent in the standard real-simulation construction. Prior work of McKague, Mosca, and Gisin showed that real simulation can preserve locality by distributing the real encoding across the system~\cite{mckague2009simulating}, but their construction requires access to general real local gates. Our result instead compiles all the way to Hadamard, $X$, and generalized Toffoli gates, while preserving depth.

\newcommand{\vpad}[1]{%
  \raisebox{0pt}[\dimexpr\height+4pt\relax][\dimexpr\depth+4pt\relax]{#1}%
}

\begin{table}[t]
\centering\small
\begin{tabular}{
    >{\centering\arraybackslash}m{0.12\linewidth}|
    >{\centering\arraybackslash}m{0.20\linewidth}|
    >{\centering\arraybackslash}m{0.17\linewidth}
    >{\centering\arraybackslash}m{0.20\linewidth}
    >{\centering\arraybackslash}m{0.17\linewidth}
}
\vpad{\textbf{Hierarchy}}
&
\vpad{\textbf{Single-Qubit Gates}}
&
\vpad{\shortstack{\textbf{Bounded Toffoli}\\ \textbf{(with Fan-Out)}}}
&
\shortstack{\textbf{Unbounded Toffoli}\\ \textbf{(with Fan-Out)}}
&
\shortstack{\textbf{Threshold}\\ \textbf{(with Fan-Out)}} \\
\hline
\vpad{\textsf{Standard}}
&
\vpad{Arbitrary}
&
\vpad{\shortstack{$\QNC$\\($\QNC_f$)}}
&
\vpad{\shortstack{$\QAC$\\($\QAC_f$)}}
&
\vpad{\shortstack{$\QTC$\\($\QTC_f$)}} \\
\hline
\vpad{\textsf{Universal}}
&
\vpad{$H,T,X$}
&
\vpad{\shortstack{$\UQNC$\\($\UQNC_f$)}}
&
\vpad{\shortstack{$\UQAC$\\($\UQAC_f$)}}
&
\vpad{\shortstack{$\UQTC$\\($\UQTC_f$)}} \\
\hline
\vpad{\textsf{Real}}
&
\vpad{$H,X$}
&
\vpad{\shortstack{$\HQNC$\\($\HQNC_f$)}}
&
\vpad{\shortstack{$\HQAC$\\($\HQAC_f$)}}
&
\vpad{\shortstack{$\HQTC$\\($\HQTC_f$)}}
\end{tabular}
\caption{The standard, universal, and real circuit hierarchies.
Within each column, the multiqubit primitives are unchanged and only
the single-qubit gate set is restricted.}
\label{tab:shallow-gate-hierarchies}
\end{table}

\paragraph{Implications for Circuit Complexity.}
Finally, we use our depth-preserving synthesis and real-simulation results to revisit the standard quantum circuit hierarchy. For a circuit class $\mathcal C$ and depth bound $D(n)$, we write $\mathcal C[D]$ for its restriction to depth $O(D(n))$, and use $\mathcal C^0$ for constant depth. The standard hierarchy $\QNC\subseteq\QAC\subseteq\QTC$ allows arbitrary single-qubit gates while varying the available multiqubit primitives, from bounded-width Toffoli gates in $\QNC$, to generalized Toffoli gates in $\QAC$, to Threshold gates in $\QTC$. The corresponding Fan-Out hierarchy additionally permits unbounded Fan-Out. We introduce a \emph{universal hierarchy} ($\UQNC\subseteq\UQAC\subseteq\UQTC$), in which single-qubit gates are restricted to $\{H,T,X\}$, and a \emph{real hierarchy} ($\HQNC\subseteq\HQAC\subseteq\HQTC$),  in which they are further restricted to $\{H,X\}$. The resulting classes are summarized in Table~\ref{tab:shallow-gate-hierarchies}.

These restricted hierarchies expose a feature hidden by the standard model. Fan-Out has long appeared to be an unusually powerful primitive for shallow quantum computation~\cite{hoyerspalek}. Takahashi and Tani sharpened this picture by proving $\QNC_f^0=\QAC_f^0=\QTC_f^0$,
when arbitrary single-qubit gates are available~\cite{takahashi2016collapse}. Thus, in the standard hierarchy, Fan-Out appears powerful enough to subsume generalized Toffoli and even Threshold gates at constant depth. Their reductions, however, rely on continuously parameterized single-qubit rotations. This dependence is particularly notable in light of our synthesis lower bound, which shows that approximating arbitrary single-qubit rotations to inverse-polynomial accuracy over any fixed finite bounded-arity gate set requires $\Omega(\log\log n)$-depth in the worst case. Although this does not rule out a different constant-depth simulation in the presence of Fan-Out, it raises the possibility that some of the apparent power of Fan-Out comes from its interaction with freely available local rotations.

Our results determine how much of this hierarchy survives after arbitrary local rotations, and subsequently complex amplitudes, are removed. For every polynomially bounded depth function $D(n)\geq1$,
\[
    \QAC[D]=\UQAC[D]=\HQAC[D],
\]
where the final equality concerns bounded-error decision computation. With Fan-Out or Threshold gates, we obtain the broader collapse
\[
\begin{aligned}
    \QAC_f[D]
    &=\UQAC_f[D]
     =\HQAC_f[D]
     =\QTC[D]
     =\UQTC[D]
     =\HQTC[D] \\
    &=\QTC_f[D]
     =\UQTC_f[D]
     =\HQTC_f[D].
\end{aligned}
\]
For bounded-arity circuits with Fan-Out, our catalytic construction incurs only an additive $O(\log\log n)$ depth cost, so whenever $D(n)=\Omega(\log\log n)$,
\[
    \QNC_f[D]=\UQNC_f[D]=\HQNC_f[D].
\]
An important unresolved case, which we leave as an interesting direction for future work, is whether $\QNC_f^0=\UQNC_f^0=\HQNC_f^0$.
A positive answer would show that the standard Fan-Out collapse survives even after arbitrary local rotations and complex amplitudes are removed. A negative answer would reveal a genuine constant-depth separation arising solely from the allowed single-qubit gate sets.

The same tension between Fan-Out and generalized Toffoli underlies the longstanding conjecture $\Parity\notin\QAC^0$~\cite{moore1999qac0}. Recent work has made substantial progress toward this conjecture in restricted regimes~\cite{rosenthal2021bounds,nadimpalli2024pauli,anshu2025computational,joshi2026improved,gretta2026fourier}, but the general case of arbitrary constant-depth, polynomial-size $\QAC^0$ circuits remains open. Resolving the conjecture would substantially sharpen our understanding of the computational power of $\QAC^0$, with implications for recent results on the learnability of these circuits, pseudorandom unitary generation, quantum state synthesis, and even computing constant-fold Forrelation~\cite{vasconcelos2025learning,foxman2026random,joshi2026constant,vasconcelos2026forrelation}. Since coherent Parity and Fan-Out are equivalent under Hadamard conjugation, the conjecture asks whether generalized-Toffoli circuits can recover the power supplied directly by Fan-Out. Our established equality
\[
    \QAC^0=\UQAC^0=\HQAC^0
\]
substantially simplifies this question by showing that arbitrary single-qubit rotations and complex amplitudes can be removed entirely. Thus, it suffices to prove a Parity lower bound against constant-depth circuits containing only Hadamard, Pauli-$X$, and generalized Toffoli gates.

The real hierarchy also exposes a simple normal form for shallow quantum circuits. Because Hadamard is the only basis-changing gate, every real circuit can be written as an alternation of Hadamard layers and diagonal Boolean phase layers. The phase functions are highly structured since their allowed nonlinear terms are determined directly by the available multiqubit primitives, such as bounded-width AND, unbounded AND, Threshold, or controlled Parity. This is precisely the architecture underlying Forrelation, which alternates Boolean phase operators with global Hadamard transforms~\cite{aaronson2015forrelation}. We therefore obtain a \emph{structured Forrelation} normal form for the real hierarchy and show that the acceptance bias of any single-output circuit can be expressed as a scalar structured Forrelation quantity with only a constant-factor increase in the number of layers. Consequently, proving that these restricted Forrelation expressions cannot realize sufficiently high-order correlations would immediately yield lower bounds for the corresponding shallow quantum circuit classes. This suggests that techniques from Forrelation and Fourier-growth lower bounds~\cite{bansal2021forrelation,girish2021lower} may provide a new route towards circuit lower bounds, such as $\Parity\notin\QAC^0$.

\subsection{The Standard, Universal, and Real Circuit Hierarchies}
\label{sec:hierarchy-gate-sets}

We briefly formalize the circuit models used throughout the paper, previously summarized in Table~\ref{tab:shallow-gate-hierarchies}. All circuit families are nonuniform and polynomial-size. They do not allow for intermediate measurements and are not given access to pre-supplied quantum advice or catalytic resource states.

For a gate model $\mathcal C$ and depth bound $D(n)\geq1$, let $\mathcal C[D]$ denote the set of languages $L\subseteq\{0,1\}^*$ for which there exists a family of depth-$O(D(n))$ circuits $\{C_n\}_{n\geq1}$ over the gate set of $\mathcal C$. On input $x\in\{0,1\}^n$, the circuit $C_n$ acts on $\ket{x}\ket*{0^{a(n)}}$, where $a(n)=\poly(n)$. Let $\Pi_{\mathrm{out}}:=\ket1\!\bra1_{\mathrm{out}}\otimes I$ denote the projector onto outcome $1$ of the designated output qubit. We define the acceptance probability of $C_n$ on input $x$ as
\begin{align}
    p_{C_n}(x)
    &:=
    \left\|
        \Pi_{\mathrm{out}}
        C_n
        \ket{x}\ket*{0^{a(n)}}
    \right\|_2^2.
    \label{eq:intro-acceptance-probability}
\end{align}
The family $\{C_n\}$ therefore decides $L$ with bounded error if, for every $x\in\{0,1\}^n$,
\begin{align}
    x\in L
    &\quad\Longleftrightarrow\quad
    p_{C_n}(x)\geq\frac23,
    \notag\\
    x\notin L
    &\quad\Longleftrightarrow\quad
    p_{C_n}(x)\leq\frac13.
    \label{eq:intro-bounded-error-condition}
\end{align}
Equivalently, $L\in\mathcal C[D]$ if such a circuit family exists. Finally, for fixed integer $k\geq0$, we write $\mathcal C^k:=\mathcal C[\log^k n]$, so that $\mathcal C^0$ denotes constant depth.

We first define the multiqubit primitives that distinguish the circuit hierarchies. A generalized Toffoli gate with control set $S$ and target qubit $t$ acts as
\begin{align}
    \ket{z_S,z_t}
    \longmapsto
    \ket{z_S}\ket{z_t\oplus\prod_{i\in S}z_i},
    \label{eq:intro-generalized-toffoli}
\end{align}
so the target qubit is flipped if and only if every control qubit is in state $\ket1$. A Threshold gate with control set $S$, target qubit $t$, and hardwired threshold $u$ acts as
\begin{align}
    \ket{z_S,z_t}
    \longmapsto
    \ket{z_S}\ket{z_t\oplus\left[\sum_{i\in S}z_i\geq u\right]},
    \label{eq:intro-threshold}
\end{align}
so the target qubit is flipped if and only if at least $u$ of the control qubits are in state $\ket1$. Finally, Fan-Out with control qubit $c$ and target set $S$ acts as
\begin{align}
    \ket{z_c}\ket{z_S}
    \longmapsto
    \ket{z_c}\ket{z_S\oplus z_c\mathbf 1},
    \label{eq:intro-fanout}
\end{align}
so the value of the control qubit is coherently XORed into every target qubit.

We now define the standard, universal, and real quantum circuit hierarchies. These are quantum analogues of the classical $\NC\subseteq\AC\subseteq\TC$ circuit hierarchy, whose three levels are distinguished by their available multiqubit gates: bounded-arity Toffoli, generalized Toffoli of unbounded arity, and Threshold gates, respectively. The standard, universal, and real hierarchies then differ only in their allowed single-qubit gates. Although the classical classes natively include Fan-Out, in the quantum setting we use the subscript $f$ to denote additional access to the quantum Fan-Out gate.

\begin{definition}[Standard hierarchy]
The \emph{standard hierarchy} consists of the classes
$\QNC\subseteq\QAC\subseteq\QTC$, together with their Fan-Out variants
$\QNC_f$, $\QAC_f$, and $\QTC_f$, and allows arbitrary single-qubit gates.
\end{definition}

\begin{definition}[Universal hierarchy]
The \emph{universal hierarchy} is obtained from the corresponding standard classes by restricting the single-qubit gate set to the fixed universal set $\{H,T,X\}$. We denote the resulting classes by
$\UQNC\subseteq\UQAC\subseteq\UQTC$, with Fan-Out variants
$\UQNC_f$, $\UQAC_f$, and $\UQTC_f$.
\end{definition}

\begin{definition}[Real hierarchy]
The \emph{real hierarchy} is obtained from the corresponding universal classes by further restricting the single-qubit gate set to $\{H,X\}$. We denote the resulting classes by
$\HQNC\subseteq\HQAC\subseteq\HQTC$, with Fan-Out variants
$\HQNC_f$, $\HQAC_f$, and $\HQTC_f$.
\end{definition}
\noindent As previously mentioned, the circuit classes above are defined as bounded-error \emph{decision} classes, and therefore capture the complexity of computing classical decision functions. We will also consider the corresponding \emph{unitary-synthesis setting}, where the goal is instead to approximate the full action of a target unitary on arbitrary quantum inputs. For a circuit model $\mathcal C$ and depth bound $D(n)$, we write $\textsf{Unitary}(\mathcal C[D])$ for the families of unitaries that admit polynomial-size depth-$O(D(n))$ implementations over $\mathcal C$ to every inverse-polynomial accuracy, using $\ket{0}$-initialized ancillae that are approximately restored at the end, up to an input-independent global phase.

This distinction will be important throughout. Our standard-to-universal compilation results generally hold at the stronger level of unitary synthesis, immediately implying the corresponding bounded-error decision-class equalities. By contrast, our real simulations preserve the acceptance probabilities of the original computation through an encoded real representation, establishing equalities of decision classes rather than equivalences of the underlying complex unitaries.

\subsection{Main Results and Technical Overview}
\label{sec:main-results}
\label{sec:proof-overview}

We now give a technical overview of our main results and the ideas behind their proofs. We begin with single-qubit synthesis, where we obtain an optimal depth bound in the bounded-arity setting and show that wider gates remove this obstruction entirely (\Cref{sec:intro-synthesis}). We then develop two complementary tools for extending these single-gate results to whole circuits. Reusable catalyst states allow a one-time preparation cost to be shared across many gates, while shallow implementations of the limited Fan-Out operations introduced by our synthesis procedure remove the need to treat Fan-Out as a primitive (\Cref{sec:intro-catalyst,sec:intro-small-fanout}). We next turn to real simulation, where the challenge is to remove complex amplitudes without destroying the parallel structure of the original circuit (\Cref{sec:intro-real-simulation}). Finally, we show that the resulting real circuits admit a structured Fourier description and a scalar Forrelation representation of their acceptance probabilities (\Cref{sec:intro-fourier-hierarchy}).

\subsubsection{Depth-Optimal Single-Qubit Synthesis}
\label{sec:intro-synthesis}

Our first result characterizes the depth required for arbitrary single-qubit gate synthesis based on the available multi-qubit gates. In the bounded-arity setting, we obtain an $O(\log\log(1/\varepsilon))$-depth construction and prove this dependence to be optimal. Allowing wider multi-qubit gates reduces the synthesis depth to a constant. The formal upper bounds appear in \Cref{thm:rz-synthesis-qnc} and \Cref{cor:rz-synthesis-qacf}, together with \Cref{lem:single-qubit-to-rz} for the reduction from arbitrary single-qubit unitaries to $R_z$ rotations.

\begin{theorem}[Single-qubit synthesis, informal]
\label{thm:intro-single-qubit-synthesis}
For every single-qubit unitary $U$ and precision $\varepsilon\in(0,1)$, there exists a depth-$O(\log\log(1/\varepsilon))$ circuit over $\{H,T,\operatorname{CNOT}\}$ that cleanly $\varepsilon$-approximates $U$ using $O(\log(1/\varepsilon))$ gates and clean ancillae. Moreover, for every fixed $\gamma\in(0,1)$, there exists a constant-depth circuit with additional access to width-$O(\log(1/\varepsilon))$ generalized Toffoli gates and width-$O_\gamma(\log^{1/2+\gamma}(1/\varepsilon))$ Fan-Out gates that cleanly $\varepsilon$-approximates $U$ using $O(\log(1/\varepsilon))$ gates and $O_\gamma(\log^{1+\gamma}(1/\varepsilon))$ clean ancillae.
\end{theorem}

By the standard Euler-angle decomposition, it suffices to synthesize arbitrary $R_z(\theta)$ rotations, since any single-qubit unitary can be written, up to global phase, as a constant-length sequence of $R_z$ rotations and Hadamard gates~\cite[Theorem~4.1]{nielsen2010quantum}. At a high level, our $R_z(\theta)$ synthesis procedure has two main steps. We first construct an approximate block encoding of $\frac12R_z(\theta)$ and then apply a single round of oblivious amplitude amplification to remove the factor of $1/2$, yielding an approximation of $R_z(\theta)$. The circuit depth is therefore governed by the reversible operations required to construct the block encoding and implement the amplification reflections.

For the block-encoding step, we begin with a dyadic approximation of the complex number $\frac12e^{i\theta/2}$. For $m=\Theta(\log(1/\varepsilon))$, we choose nonnegative integers $n_0,n_1,n_2,n_3$ summing to $2^m$ such that
\[
    \frac{(n_0-n_1)+i(n_2-n_3)}{2^m}
    \approx
    \frac12e^{i\theta/2}
\]
to error $O(2^{-m})$ (\Cref{thm:int_approx}). We realize these coefficients coherently by preparing a uniform superposition over $m$ address qubits, partitioning its $2^m$ basis states into four intervals of lengths $n_0,n_1,n_2,n_3$, and conditionally applying $I$, $-I$, $-iZ$, or $iZ$. After uncomputing the interval label, the all-zero ancilla block is
\begin{align}
    K
    &=
    \frac{(n_0-n_1)I-i(n_2-n_3)Z}{2^m}
    \approx
    \frac12R_z(\theta).
    \label{eq:intro-dyadic-block}
\end{align}
The only nontrivial circuit task is therefore to determine, coherently, which interval contains a given $m$-bit address. This is where the available multiqubit gates determine the depth. With bounded-arity Toffoli gates, the required comparisons are implemented by balanced trees in $O(\log m)$-depth, whereas wider Toffoli and Fan-Out gates evaluate them in constant depth (\Cref{lem:comparison-no-fanout,lem:comparison-fanout}).

The second step removes the factor of $1/2$ using a single round of oblivious amplitude amplification. Aside from calls to the block-encoding circuit and its inverse, the only additional operation is a reflection about the all-zero $m$-qubit ancilla state. This is the second point at which the available multiqubit gates determine the depth. A generalized Toffoli implements the reflection in constant depth, whereas bounded-arity Toffoli gates require a balanced tree of depth $O(\log m)$. \Cref{lem:robust-oaa-rz} shows that the same procedure works for the approximate block encoding above, yielding an approximation of $R_z(\theta)$ on every input while leaving only negligible amplitude on nonzero ancilla states. Since $m=\Theta(\log(1/\varepsilon))$, the complete bounded-arity construction has depth $O(\log\log(1/\varepsilon))$. We next show that this dependence is optimal. The formal statement appears in \Cref{thm:depth_lower_bound}.

\begin{theorem}[Synthesis depth lower bound, informal]
\label{thm:intro-synthesis-lower-bound}
For every fixed finite bounded-arity gate set, approximating every
single-qubit $R_z$ rotation to error $\varepsilon$ requires worst-case depth
$\Omega(\log\log(1/\varepsilon))$, even with arbitrarily many initialized ancillae and without requiring those ancillae to be restored at the end.
\end{theorem}

\noindent The proof combines a light-cone argument with counting. If every gate has arity at most $r$, then the backward light cone of one output qubit in depth $d$ contains only $O(r^d)$ gates. A finite gate alphabet therefore induces at most $\exp(O((d+1)r^d))$ distinct output channels. On the other hand, approximating all $R_z$ rotations to precision $\varepsilon$ requires $\Omega(1/\varepsilon)$ distinct channels. Comparing these bounds gives $d=\Omega(\log\log(1/\varepsilon))$.

For polynomial-size circuit families and inverse-polynomial target accuracy, the single-qubit synthesis bounds can be applied independently to every single-qubit gate. Thus, for any depth bound $D(n)\geq1$, we obtain
$\textsf{Unitary}(\QNC[D])\subseteq\textsf{Unitary}(\UQNC[O(D\log\log n)])$,
with the same inclusion for $\QAC$, $\QTC$, and $\QNC_f$. When both generalized Toffoli and Fan-Out are available, the constant-depth synthesizer instead preserves depth up to a constant factor, giving
$\textsf{Unitary}(\QAC_f[D])=\textsf{Unitary}(\UQAC_f[D])$
and
$\textsf{Unitary}(\QTC_f[D])=\textsf{Unitary}(\UQTC_f[D])$
(\Cref{lem:direct-unitary-compilation}).

For bounded-error decision computation, bounded-arity constant-depth circuits admit a stronger conclusion. The measured output qubit of a $\QNC^0$ circuit has only a constant-size backward light cone, so only constantly many arbitrary single-qubit gates can affect its acceptance probability. These gates therefore need only constant-accuracy approximation, yielding $\QNC^0=\UQNC^0$.
This argument is specific to decision computation and does not give a constant-depth approximation of the full unitary. It also fails in the presence of Fan-Out, where the backward light cone of a single output qubit may contain polynomially many qubits. In particular, whether
$\QNC_f^0=\UQNC_f^0$
remains open. 

\subsubsection{Depth-Preserving Compilation with Restricted Multi-Qubit Gates}
\label{sec:intro-one-wide-primitive}

The constant-depth synthesizer of \Cref{thm:intro-single-qubit-synthesis} uses both generalized Toffoli and Fan-Out. We now ask how much of this constant-depth behavior survives when only one of these primitives is available? With generalized Toffoli alone, we show that the auxiliary Fan-Out gates introduced by the synthesizer can themselves be eliminated, in constant-depth. This yields fully depth-preserving compilation for $\QAC$ and $\QTC$. With Fan-Out alone, we instead build on the reusable catalyst of Kim--Laakkonen~\cite{kim2025clifford}. Once the catalyst is supplied, their online catalytic procedure can be implemented in constant total depth in $\UQNC_f$. We, thus, give a fully coherent $O(\log\log(1/\varepsilon))$-depth $\UQNC_f$ procedure for preparing the catalyst from $\ket{0}$-initialized ancillae and then reuse the resulting catalyst bank throughout the computation, so that this preparation cost is incurred only once.

\paragraph{Generalized Toffoli without Fan-Out.}
\label{sec:intro-small-fanout}
The only obstacle to implementing our constant-depth synthesizer directly in $\UQAC$ is its use of polylogarithmic-width Fan-Out. Although such Fan-Out is already known to admit constant-depth implementations in the standard shallow-depth hierarchy, existing constructions rely on unrestricted single-qubit rotations. Rosenthal gave an exponential-size constant-depth $\QAC^0$ approximation of Fan-Out~\cite{rosenthal2021bounds}, and Grier, Morris, and Wu later made this construction exact~\cite{grier2026mathsf}. Restricting these constructions to polylogarithmic width yields polynomial size, but they still require width-dependent single-qubit rotations, including rotations with angles scaling inversely with the Fan-Out width. Therefore, they do not directly yield implementations over a fixed universal gate set. Nevertheless, we show that such Fan-Out can be implemented in constant depth using only Hadamard and generalized Toffoli gates. The formal statement is \Cref{cor:inverse-poly-polylog-fanout}.

\begin{lemma}[Polylogarithmic Fan-Out, informal]
\label{lem:intro-small-fanout}
For every fixed $c,\beta>0$, Fan-Out on $O(\log^c n)$ targets has a polynomial-size constant-depth implementation over Hadamard and generalized Toffoli gates with coherent error at most $n^{-\beta}$.
\end{lemma}

\noindent The proof uses a recent construction of Grier and Morris~\cite{grier2025quantum} for preparing an approximate \emph{nekomata}, a cat-like state whose target register is concentrated on the two strings $\ket{0^q}$ and $\ket{1^q}$. We specialize their construction so that its only nontrivial phase operation marks $\ket{1^q}$, which can be implemented exactly by a generalized Toffoli gate.  Overall, the nekomata preparation uses only Hadamard and generalized Toffoli gates. To turn this state into Fan-Out, we use a reduction of Rosenthal~\cite{rosenthal2021bounds} that converts any such nekomata-preparation circuit into a coherent Parity gate using only the preparation circuit, its inverse, and additional Hadamard and generalized Toffoli gates. Thus the entire reduction remains within the same fixed gate set. Conjugating Parity by Hadamards then gives Fan-Out. Choosing $q=\Theta(\log n)$ yields logarithmic-width Fan-Out with inverse-polynomial error and polynomial resources, and a constant-depth copying tree extends the construction to any fixed polylogarithmic width (\Cref{lem:uqac-logarithmic-fanout} and \Cref{fact:nekomata-to-parity}).

Because the approximation holds coherently on arbitrary inputs and approximately restores the implementation ancillae, these Fan-Out gadgets can be substituted directly into the single-qubit synthesizer. This removes primitive Fan-Out while preserving constant depth and yields the fixed-gate compiler of \Cref{thm:qac-fixed-gate-simulation}. Consequently, arbitrary single-qubit gates can be compiled into the fixed universal basis without asymptotically increasing the depth of $\QAC$ or $\QTC$ circuits. The formal statements are
\Cref{cor:qac-uqac-collapse,cor:qtc-uqtc-collapse}.

\begin{theorem}[Depth-preserving universalization, informal]
\label{thm:intro-depth-preserving-universalization}
For every depth bound $D(n)\geq1$,
\begin{align*}
    \textsf{Unitary}(\QAC[D])
    &=
    \textsf{Unitary}(\UQAC[D]),\\
    \textsf{Unitary}(\QTC[D])
    &=
    \textsf{Unitary}(\UQTC[D]).
\end{align*}
\end{theorem}

\paragraph{Fan-Out without generalized Toffoli.}
\label{sec:intro-catalyst}
Takahashi and Tani showed that $\QNC_f^0=\QAC_f^0$, in the standard hierarchy, by implementing generalized Toffoli gates from Fan-Out in constant depth~\cite{takahashi2016collapse}. Their construction, however, relies on arbitrary single-qubit gates and therefore does not directly imply the corresponding equality over a fixed universal gate set. With only bounded-arity gates and Fan-Out, we do not know how to obtain the same constant-depth compilation. Instead, we build on the reusable eigenstate catalyst of Kim--Laakkonen~\cite{kim2025clifford}. Once the catalyst is supplied, a controlled-linear-map implementation of (\Cref{lem:controlled-linear-map}) realizes their grid-phase kickback procedure in constant total depth in $\UQNC_f$ and returns the catalyst unchanged after each use. However, their construction does not provide a fully coherent preparation of the required eigenstate catalyst from standard $\ket{0}$-initialized ancillae. We supply this missing end-to-end ingredient by giving a fully coherent $O(\log\log(1/\varepsilon))$-depth $\UQNC_f$ procedure for preparing the catalyst from $\ket0$ ancillae. Since the catalyst is reusable, this preparation cost is incurred once, for the entire computation.

The key observation is that, before a fixed basis relabeling, the catalyst is exponentially close to the product state
\begin{align}
    \ket{\chi_b}
    &=
    \frac1{\sqrt{2^b}}
    \sum_{j=0}^{2^b-1}e^{-2\pi i j/(2^b-1)}\ket j
    =
    \bigotimes_{t=0}^{b-1}
    \frac{\ket0+e^{-2\pi i2^t/(2^b-1)}\ket1}{\sqrt2}.
    \label{eq:intro-catalyst-product}
\end{align}
Its phase factors can therefore be synthesized in parallel. We then implement the required basis relabeling in depth $O(\log b)$ in $\UQNC_f$ (\Cref{lem:log-width-permutation}). Combining these two steps results in a catalyst factory (\Cref{lem:kim-catalyst-factory}), which prepares a bank large enough to service one parallel circuit layer in
$O(\log\log((w+2)/\varepsilon))$ depth, where $w$ is the maximum number of simultaneous phase gates. The same catalyst bank can then be reused throughout the computation. Since the ideal catalyst is restored after each use, its preparation error is incurred only once rather than accumulating with the number of circuit layers. \Cref{thm:catalytic-universalization} formalizes this amortization and gives an additive-depth compiler, with the catalyst-preparation cost added only once to the depth of the original circuit. For polynomial-size families and inverse-polynomial accuracy, this gives
$\textsf{Unitary}(\QNC_f[D])=\textsf{Unitary}(\UQNC_f[D])$
whenever $D(n)=\Omega(\log\log n)$
(\Cref{cor:qnc-f-universalization}).

Combining this with the exact constant-depth simulations of Takahashi and Tani~\cite{takahashi2016collapse} recovers the collapse of the entire Fan-Out hierarchy over a fixed universal gate set. Their simulations reduce generalized Toffoli and Threshold gates to bounded-arity gates, Fan-Out, and arbitrary single-qubit gates, while our catalytic universalization removes the remaining arbitrary single-qubit gates once the additive $O(\log\log n)$ cost can be absorbed. The formal statement is \Cref{cor:fanout-hierarchy-universalization}. 

\begin{theorem}[Collapse of the Universal Fan-Out hierarchy, informal]
\label{thm:intro-universal-fanout-hierarchy}
For every depth bound $D(n)=\Omega(\log\log n)$,
\[
    \QNC_f[D]
    =
    \UQNC_f[D]
    =
    \QAC_f[D]
    =
    \UQAC_f[D]
    =
    \QTC_f[D]
    =
    \UQTC_f[D].
\]
\end{theorem}
\noindent Thus, the standard Fan-Out hierarchy collapse survives the restriction to a fixed universal gate set at doubly logarithmic depth and above. At constant depth, we already have
$\QAC_f^0=\UQAC_f^0$
and
$\QTC_f^0=\UQTC_f^0$
from the depth-preserving compilation results above. The remaining gap is therefore specific to bounded-arity circuits with Fan-Out, where it remains open whether
$\QNC_f^0=\UQNC_f^0$ and, consequently, whether the full Takahashi--Tani collapse persists over a fixed universal gate set.

\subsubsection{Depth-Preserving Real Simulation}
\label{sec:intro-real-simulation}

We next turn to depth-preserving real simulation. As discussed above, the standard realification of quantum computation does not preserve parallelism because its shared ancilla can serialize otherwise disjoint gates. Prior work of McKague, Mosca, and Gisin showed that real simulation can instead preserve locality by distributing the encoding across the system~\cite{mckague2009simulating}. Their construction, however, replaces each complex local gate by a corresponding real local transformation and assumes that this induced real transformation can be implemented directly. It therefore does not compile the simulation into a fixed finite real gate set such as Hadamard and Toffoli. We now show that real simulation can preserve the parallel structure of the original circuit while compiling all the way to the fixed real gate sets defining our shallow circuit classes. The formal statements are
\Cref{thm:h-only-depth-preserving-simulation,thm:bounded-arity-real-decision-simulation}.

\begin{theorem}[Depth-preserving real simulation, informal]
\label{thm:intro-real-simulation}
For every $\varepsilon\in(0,1)$, a depth-$d$ circuit of gate count and width at most $s\geq2$ in
$\UQAC$, $\UQAC_f$, $\UQTC$, or $\UQTC_f$ has a size-$\operatorname{poly}(s,1/\varepsilon)$ simulation in the corresponding real model
$\HQAC$, $\HQAC_f$, $\HQTC$, or $\HQTC_f$, respectively, of depth $O(d+1)$ that preserves the acceptance probability of every classical input up to error $\varepsilon$. For $\UQNC$ and $\UQNC_f$, the corresponding simulations in $\HQNC$ and $\HQNC_f$ have depth
$O(d+\log\log(s/\varepsilon))$
and polynomial size and workspace.
\end{theorem}

To preserve the parallel structure of the original circuit, we need a real encoding in which local gates remain local. We achieve this by representing the density matrix through its Pauli coefficients, using the standard Pauli-transfer representation~\cite{wood2015tensor,huang2022bloch,kunold2023vectorization}. For a pure state $\rho=\ket\psi\bra\psi$, define
\begin{align}
    \ket\rho_{\mathrm P}
    &:=
    2^{-q/2}
    \sum_{P\in\{I,X,Y,Z\}^{\otimes q}}
    \operatorname{Tr}(P\rho)\ket P.
    \label{eq:intro-pauli-encoding}
\end{align}
The amplitudes are real, and purity makes the encoded state normalized. If a gate $U$ acts on $r$ original qubits, then conjugation by $U$ induces the real orthogonal transformation
\begin{align}
    [O_U]_{P,Q}
    &:=
    2^{-r}\operatorname{Tr}(P\,UQU^\dagger),
    \label{eq:intro-pauli-transfer}
\end{align}
on only the corresponding $2r$ Pauli-label qubits. Hence disjoint gates remain disjoint after encoding, avoiding the serialization of the standard shared-ancilla realification. Computational-basis inputs have depth-one product-state encodings, and the original acceptance probability can be recovered exactly by measuring only the two label qubits corresponding to the output qubit
(\Cref{lem:real-pauli-encoding}).

The remaining challenge is to implement each induced Pauli-transfer transformation using only gates available in the corresponding real circuit class. Clifford gates such as Hadamard and Fan-Out act as signed permutations of the Pauli labels, and these permutations can be realized directly within $\HQAC_f$ or the relevant real model
(\Cref{prop:encoded-hadamard,lem:real-encoded-fanout}). The non-Clifford $T$ gate instead becomes a controlled real rotation
(\Cref{prop:encoded-t}). In the wide-gate models, we implement this rotation using our Hadamard-only $R_y$ synthesis, keeping the simulation within $\HQAC$ (\Cref{lem:h-only-real-rotations}). In the bounded-arity models, every gate other than $T$ acts on only constantly many original qubits, so its Pauli-transfer action involves only constantly many label qubits and can therefore be implemented exactly by a constant-size real circuit. Encoded $T$ is handled instead using a reusable Hadamard-eigenstate catalyst. Preparing this catalyst once and reusing it throughout the computation yields only an additive $O(\log\log(s/\varepsilon))$ depth overhead (\Cref{thm:bounded-arity-real-decision-simulation}).

The main difficulty is generalized Toffoli. After Hadamard-conjugating its target, it becomes a generalized controlled-$Z$, whose Pauli-transfer action we show is a reflection with a particularly simple negative eigenspace (\Cref{prop:encoded-generalized-cz}). Fixing the first Pauli-label string partitions the positions into those carrying $I/Z$ labels and those carrying $X/Y$ labels. The reflection applies a minus sign exactly when every $I/Z$ position satisfies a local $\ket{-}$ condition and all $X/Y$ positions have the same $Y$-eigenvalue, either all $\ket{+i}$ or all $\ket{-i}$. The local conditions can be checked independently, so the only nonlocal task is testing agreement among the $Y$-eigenvalues. We do this without comparing every pair. Instead, at several bucket scales, we randomly assign the active positions to paired buckets and look for a pair of singleton buckets whose isolated $Y$-labels disagree. Whenever both signs are present, an appropriate scale exposes such a pair with constant probability, and parallel repetition reduces the failure probability. The bucket tests, routing, and $Y$-label comparisons can all be implemented coherently in constant depth, yielding the encoded generalized-Toffoli construction of \Cref{thm:real-encoded-toffoli}.

Finally, we handle Threshold circuits indirectly through the Fan-Out model. Takahashi and Tani give the standard constant-depth reduction $\QTC[D]\subseteq\QNC_f[D]$~\cite{takahashi2016collapse}, while Grier and Morris give constant-depth real Fan-Out within $\HQTC$~\cite{grier2025quantum}. Together with our fixed-gate real simulation, these reductions yield
$\UQTC[D]=\HQTC[D]$ and $\UQTC_f[D]=\HQTC_f[D]$
(\Cref{fact:real-threshold-fanout} and \Cref{cor:real-threshold-simulation}).

These real simulations, together with our universalization results, give the following relations among the standard, universal, and real hierarchies. The formal statements are
\Cref{cor:qac-uqac-hqac-collapse,cor:real-wide-hierarchies,cor:full-real-fanout-hierarchy}.

\begin{theorem}[Real shallow-depth hierarchy, informal]
\label{thm:intro-real-hierarchy}
For every depth bound $D(n)\geq1$,
\[
    \QAC[D]=\UQAC[D]=\HQAC[D],
\]
and
\begin{align*}
    \QAC_f[D]
    &=
    \UQAC_f[D]
    =
    \HQAC_f[D]
    =
    \QTC[D]
    =
    \UQTC[D]
    =
    \HQTC[D]\\
    &=
    \QTC_f[D]
    =
    \UQTC_f[D]
    =
    \HQTC_f[D].
\end{align*}
If $D(n)=\Omega(\log\log n)$, then the bounded-arity Fan-Out classes also join this collapse,
\[
    \QNC_f[D]
    =
    \UQNC_f[D]
    =
    \HQNC_f[D].
\]
\end{theorem}
\noindent At constant depth without Fan-Out, we additionally obtain the decision-specific equality
$\QNC^0=\UQNC^0=\HQNC^0$
(\Cref{prop:real-qnc-zero-decision}). Fan-Out removes the constant-size light-cone property underlying this argument, and it remains open whether
$\QNC_f^0=\UQNC_f^0=\HQNC_f^0$.

\subsubsection{Structured Fourier Circuits and Scalar Forrelation}
\label{sec:intro-fourier-hierarchy}
\label{sec:intro-forrelation-hierarchy}

Because Hadamard is the only basis-changing gate in the real hierarchy, every circuit can be viewed as alternating Hadamard layers with reversible Boolean transformations. Let $S$ be a set of control qubits, let $t$ be a target qubit, and let $f:\{0,1\}^{|S|}\to\{0,1\}$ be a Boolean function. We denote by $V_f$ the reversible unitary that computes $f$ into the target by XOR,
\begin{align}
    V_f\ket{z_S,z_t}
    &=
    \ket{z_S,z_t\oplus f(z_S)}.
    \label{eq:intro-reversible-boolean-gate}
\end{align}
Writing $H_t$ for a Hadamard on the target qubit, conjugating $V_f$ by $H_t$ converts it into the diagonal phase
\begin{align}
    H_tV_fH_t\ket{z_S,z_t}
    &=
    (-1)^{z_tf(z_S)}\ket{z_S,z_t}.
    \label{eq:intro-gate-phase}
\end{align}
Fan-Out satisfies an analogous identity after conjugating all of its targets. Thus each multiqubit primitive in the real hierarchy corresponds to a restricted Boolean phase, as summarized in \Cref{tab:structured-phase-families}.

\begin{table}[t]
\centering
\small
\setlength{\tabcolsep}{4pt}
\renewcommand{\arraystretch}{1.5}
\begin{tabular}{@{}
    >{\raggedright\arraybackslash}m{0.20\linewidth}
    >{\centering\arraybackslash}m{0.40\linewidth}
    >{\centering\arraybackslash}m{0.34\linewidth}@{}}
\toprule
\textbf{Multiqubit gate}
& \textbf{Associated block phase $g(z_B)$}
& \textbf{Generated phase family}\\
\midrule
Bounded Toffoli
& $\displaystyle\prod_{i\in B}z_i$, \quad $2\leq|B|\leq3$
& $\mathcal A_3=\operatorname{Phase}(\mathcal G_3)$\\
Generalized Toffoli
& $\displaystyle\prod_{i\in B}z_i$, \quad $|B|\geq2$
& $\mathcal A=\operatorname{Phase}(\mathcal G_{\mathrm{AND}})$\\
Threshold
& $\displaystyle z_t\left[\sum_{i\in S}z_i\geq u\right]$
& $\mathcal T=\operatorname{Phase}(\mathcal G_{\mathrm{THR}})$\\
Fan-Out
& $\displaystyle z_c\left(\bigoplus_{i\in S}z_i\right)$
& $\operatorname{Phase}(\mathcal G_{\mathrm{FO}})$\\
\bottomrule
\end{tabular}
\caption{Block phases obtained by Hadamard conjugation of the
multiqubit gates. Each $\mathcal G$ denotes the collection of block
functions in its row. For Threshold, $z_t$ is the target bit,
$S$ indexes the controls, $B=S\mathbin{\dot\cup}\{t\}$, and
$1\leq u\leq|S|$. For Fan-Out, $z_c$ is the control bit,
$S\neq\varnothing$ indexes the targets, and
$B=\{c\}\mathbin{\dot\cup}S$.}
\label{tab:structured-phase-families}
\end{table}

A parallel layer of reversible gates therefore becomes a single diagonal phase whose Boolean function is the XOR of the corresponding block phases. Since gates within one layer act on disjoint supports, the nonlinear terms act on pairwise disjoint blocks. This motivates the following structured Fourier model.

\begin{definition}[Structured Fourier hierarchy]
\label{def:intro-structured-fourier}
A structured phase function has the form
\begin{align}
    h(z)
    &=
    \alpha_0
    \oplus
    \bigoplus_{i=1}^m \alpha_i z_i
    \oplus
    \bigoplus_{j=1}^p g_j(z_{B_j}),
    \label{eq:intro-structured-phase}
\end{align}
where the blocks $B_j$ are pairwise disjoint and each $g_j$ belongs to one of the block-function families in \Cref{tab:structured-phase-families}. Adding the Fan-Out blocks gives the corresponding families $\mathcal A_{3,f}$, $\mathcal A_f$, and $\mathcal T_f$. For any such phase family $\mathcal F$, let
$\mathsf{FH}_{\mathcal F}[D]$ denote the bounded-error single-output decision class represented by polynomial-size circuits
\begin{align}
    U
    &=
    H_{S_\ell}D_{h_\ell}\cdots
    H_{S_1}D_{h_1}H_{S_0},
    \qquad
    D_h\ket z=(-1)^{h(z)}\ket z,
    \label{eq:intro-structured-fourier-circuit}
\end{align}
with $\ell=O(D(n))$ and $h_j\in\mathcal F$.
\end{definition}
\noindent 
The disjoint-block restriction records the parallel structure of the original circuit and distinguishes these classes from Fourier hierarchies that permit unrestricted classical computation between Hadamard transforms~\cite{shi2005tradeoffs}. The formal statements are
\Cref{lem:forrelation-diagonal-primitives}.

\begin{theorem}[Structured Fourier normal forms, informal]
\label{thm:intro-structured-fourier}
For every depth bound $D(n)\geq1$,
\begin{align*}
    \HQNC[D]
    &=
    \mathsf{FH}_{\mathcal A_3}[D],
    &
    \HQAC[D]
    &=
    \mathsf{FH}_{\mathcal A}[D],
    &
    \HQTC[D]
    &=
    \mathsf{FH}_{\mathcal T}[D],
    \\
    \HQNC_f[D]
    &=
    \mathsf{FH}_{\mathcal A_{3,f}}[D],
    &
    \HQAC_f[D]
    &=
    \mathsf{FH}_{\mathcal A_f}[D],
    &
    \HQTC_f[D]
    &=
    \mathsf{FH}_{\mathcal T_f}[D].
\end{align*}
The correspondence is exact and preserves depth up to a constant factor in both directions.
\end{theorem}
\noindent Combining this exact correspondence with the hierarchy equalities above gives corresponding normal forms for the standard and universal classes
(\Cref{thm:standard-forrelation-normal-forms}). For example, for every $D(n)\geq1$,
$\QAC[D]=\UQAC[D]=\HQAC[D]=\mathsf{FH}_{\mathcal A}[D]$.

We now connect this normal form to Forrelation, introduced by Aaronson and Ambainis as a canonical quantum problem built from alternating Boolean phase oracles and global Hadamard transforms~\cite{aaronson2015forrelation}. At the circuit level, a $k$-fold Forrelation instance consists of diagonal Boolean phases interleaved with full $H^{\otimes m}$ layers, and its value is the resulting all-zero matrix element. Our structured Fourier circuits already have the same alternating phase--Hadamard form, except that their Hadamard layers may act on only subsets of the qubits. To obtain a genuine Forrelation circuit, we use the Aaronson--Ambainis construction~\cite{aaronson2015forrelation} to replace each partial Hadamard layer by full Hadamards together with additional diagonal phases. These added phases remain within the same structured families and increase the number of layers by only a constant factor. For Boolean phase functions $h_1,\ldots,h_k$, the resulting scalar Forrelation quantity is
\begin{align}
    \Phi_m(h_1,\ldots,h_k)
    &:=
    \bra{0^m}
    H^{\otimes m}D_{h_k}H^{\otimes m}
    \cdots D_{h_1}H^{\otimes m}
    \ket{0^m}.
    \label{eq:intro-scalar-forrelation}
\end{align}
\noindent The remaining step is to relate the usual single-output acceptance probability to the all-zero amplitude that defines Forrelation. We show that this can be done without losing the decision gap. The formal statement is \Cref{prop:forrelation-acceptance-bias}. 

\begin{proposition}[Scalar Forrelation representation, informal]
\label{prop:intro-forrelation-bias}
Let $U$ be a structured Fourier circuit with $\ell$ phase layers. For every classical input $x$, one can construct a Forrelation instance with $O(\ell+1)$ phase functions from the same structured family whose scalar value is exactly the acceptance bias of $U$,
\begin{align}
    \Phi_m(\boldsymbol g_x)
    &=
    1-2p_U(x).
    \label{eq:intro-forrelation-bias}
\end{align}
Moreover, the nonlinear phase functions are fixed by the circuit and do not depend on $x$. The input $x$ appears only through affine terms in the first and last phase layers.
\end{proposition}

\noindent The idea is to express the acceptance bias as the expectation of $Z$ on the output qubit. After converting the circuit to full-Hadamard form, write its unitary as $W$. Then
\begin{align}
    1-2p_U(x)
    &=
    \bra{u_x}W^\dagger Z_{\mathrm{out}}W\ket{u_x}.
\end{align}
The forward circuit, the output $Z$ phase, and the reversed circuit give the phase sequence of a Forrelation instance. To put this expression into the standard all-zero Forrelation form, we absorb the basis input $\ket{u_x}$ into the boundary phases. Thus each input $x$ gives a slightly different Forrelation instance $\boldsymbol g_x$. Crucially, however, the nonlinear structured phases are completely independent of $x$. The input dependence is confined to an affine boundary phase consisting only of single-qubit $Z$ phases, and increases the number of phase layers by at most a constant. In particular, a constant-depth circuit still gives a constant-fold Forrelation instance.

This distinction is important for interpreting the normal form. It does not represent the decision problem by one fixed Forrelation circuit that takes $x$ as an ordinary input. Nevertheless, any lower bound or structural characterization that applies uniformly to this class of structured Forrelation circuits, even in the presence of arbitrary affine boundary phases, would immediately constrain the corresponding shallow decision circuits. The reduction also preserves the bounded-error gap. If $U$ accepts with probability at least $2/3$, then $\Phi_m(\boldsymbol g_x)\leq-1/3$, while if it accepts with probability at most $1/3$, then $\Phi_m(\boldsymbol g_x)\geq1/3$
(\Cref{prop:forrelation-acceptance-bias} and \Cref{cor:forrelation-scalar-normal-form}).

\subsection{Discussion and Open Problems}
\label{sec:discussion}

Taken together, our results show that much of the apparent flexibility of the standard shallow-depth quantum circuit model is not essential. Arbitrary single-qubit rotations can be compiled away without increasing depth in the wide-gate models and complex amplitudes can likewise be removed for bounded-error decision computation. The resulting real circuits have a substantially more rigid structure, admitting a structured Fourier and Forrelation description, in which the available multiqubit primitives directly determine the allowed phase functions. This raises two natural directions: understanding how far the depth-preserving simulations can be extended, and determining whether the resulting normal forms can be leveraged to prove new shallow quantum circuit lower bounds. We conclude with several open questions along these lines.

\begin{open}
\label{open:intro-fixed-fanout}
Do the equalities
$\QNC_f^0=\UQNC_f^0=\HQNC_f^0$
hold?
\end{open}
\noindent This question has two parts. First, the equality
$\QNC_f^0=\UQNC_f^0$ asks whether the Takahashi--Tani collapse survives over a fixed universal gate set. A positive answer would imply
\[
    \QNC_f^0=\UQNC_f^0=\QAC_f^0=\UQAC_f^0=\QTC_f^0=\UQTC_f^0.
\]
One possible route is a constant-depth, polynomial-size approximation of generalized Toffoli using bounded-arity gates, Fan-Out, and $\{H,T,X\}$. Another is to construct a reusable phase catalyst that can itself be prepared coherently in constant depth from $\ket{0}$-initialized ancillae. The Takahashi--Tani construction does not give the former because it uses unrestricted single-qubit rotations (\Cref{fact:takahashi-tani-simulation}), while our current catalyst preparation has depth $O(\log\log n)$ rather than $O(1)$.

Second, the equality $\UQNC_f^0=\HQNC_f^0$ asks whether complex amplitudes can also be removed at constant depth in the bounded-arity Fan-Out model. Together, the two equalities would give
\[
    \QNC_f^0=\UQNC_f^0=\HQNC_f^0=\QAC_f^0=\UQAC_f^0=\HQAC_f^0=\QTC_f^0=\UQTC_f^0=\HQTC_f^0.
\]
Conversely, a separation between the standard and fixed-gate decision classes would show that unrestricted single-qubit rotations or complex amplitudes provide a genuine constant-depth computational resource.

\begin{open}
\label{open:intro-parity}
Can one prove
$\mathrm{Parity},\mathrm{Majority}\notin\HQAC^0=\mathsf{FH}_{\mathcal A}[1]$?
\end{open}

\noindent By our real-simulation and structured Fourier reductions, establishing $\QAC^0=\HQAC^0$, either lower bound would immediately imply the corresponding longstanding lower bound $\mathrm{Parity},\mathrm{Majority}\notin\QAC^0$. The advantage of working in $\mathsf{FH}_{\mathcal A}[1]$ is that the circuit has a much more structuted form. It consists only of a constant number of Hadamard layers separated by phase layers built from disjoint AND blocks, with no arbitrary single-qubit gates or complex amplitudes. The scalar Forrelation representation further reduces its decision bias to a constant-fold Forrelation quantity whose nonlinear phase blocks are fixed independently of the input and whose input dependence appears only through affine boundary phases. For example, if such a circuit computes Parity with error at most $1/3$, then $(-1)^{\mathrm{Parity}(x)}\Phi_m(\boldsymbol g_x)\geq\frac13$ for every $x$.
One possible route is to establish Fourier-growth or correlation bounds for these structured Forrelation quantities, combining ideas from Forrelation lower bounds with Fourier-analytic techniques for $\QAC^0$~\cite{gretta2026fourier}. Unlike general Forrelation, however, the phase functions here have the additional disjoint-AND structure inherited from the circuit. Exploiting this restriction may make it possible to show that constant-fold structured Forrelation cannot maintain the high-order correlations required by Parity or Majority.

\begin{open}
\label{open:intro-resources}
Can the depth reductions developed here lead to practical reductions in fault-tolerant runtime or space--time cost?
\end{open}

\noindent Recent progress in magic-state preparation suggests that logical $T$ gates may become substantially cheaper, making $T$-depth, rather than only $T$-count, an increasingly relevant measure of fault-tolerant runtime~\cite{fowler2013time,litinski2019game,chamberland2020very,itogawa2024even,gidney2024magic,daguerre2025code}. This raises the possibility that the parallel synthesis techniques developed here could yield practical savings beyond their asymptotic depth bounds. A significant obstacle, however, is that our constant-depth constructions rely on generalized Toffoli gates and other wide reversible operations. When compiled into fault-tolerant Clifford+$T$ circuits, these operations may themselves incur substantial $T$-depth, $T$-count, routing, or magic-state costs. Working out the constants and end-to-end fault-tolerant implementations could identify regimes in which the increased parallelism genuinely reduces physical runtime or space--time volume.

\section{Acknowledgements and AI Statement}
The author is supported by NSF grant 2311733 and DOE grant DE-SC0024124. The author would like to thank Joe Aulicio, Adam Bene-Watts, Thiago Bergamaschi, John Bostanci, Arjan Cornelissen, Yongshan Ding, Uma Girish, Jackson Morris, Chris Pattison, and Gregory Rosenthal for insightful discussions. In particular, Gregory Rosenthal initially posed the question of a constant-depth Solovay-Kitaev gate-synthesis procedure, Thiago Bergamaschi suggested the question of studying ``fault-tolerant" \QACZ, and Adam Bene-Watts highlighted the connection to ``structured" constant-fold Forrelation. The author would also like to thank Angelos Pelecanos, Jack Spileki, Umesh Vazirani, and John Wright for helpful feedback on the manuscript.

\paragraph{AI Disclosure.} GPT-5.6 Sol and GPT-6 Astra were used throughout the project for technical discussion, literature searches, proof checking, working through details of author-proposed ideas, and writing support. The research questions, gate-synthesis constructions, broader complexity-theoretic directions, and overall organization of the paper originated with the author. The main exception is the depth-preserving real-simulation result. The author identified the goal of extending real simulation to shallow circuits and the shared-ancilla serialization in existing constructions as the central obstacle. GPT-6 Astra suggested the core parallel-simulation approach for avoiding this bottleneck and helped develop the technical ingredients needed to preserve depth. The result was subsequently refined and checked through further interaction with the author. GPT-6 was also used in writing the manuscript, largely through iterative exchanges in which the author supplied the technical content, structure, and intended emphasis while the models suggested revisions or alternative formulations. The final manuscript was heavily written, edited, and revised by the author, who independently verified all mathematical arguments, citations, and claims included in the paper.

\section{Block-Encoding Based Single-Qubit Gate Synthesis}
\label{sec:be_gate_synth}

We now develop our single-qubit synthesis construction. Rather than
approximating the target rotation directly, we first realize
$\frac12 R_z(\theta)$ as a block of a larger unitary and then apply a
single round of oblivious amplitude amplification to remove the factor
of $1/2$. This separates the synthesis problem into a simple
block-encoding construction followed by a fixed amplification
procedure.

Our approach is closely related to two standard paradigms for quantum
gate synthesis. Phase-kickback constructions, originating with
Kitaev--Shen--Vyalyi and subsequently developed in the fault-tolerant
synthesis literature, realize dyadic approximations to phase rotations
using Fourier-state ancillae and reversible arithmetic. More generally,
linear-combination-of-unitaries constructions embed weighted sums of
simple unitaries into a larger unitary, after which oblivious amplitude
amplification can recover the desired operation
~\cite{BCCKS14,BCCKS15}. 

Here we use a particularly simple dyadic
block encoding tailored to shallow circuits. Namely, since
\begin{align}
    R_z(\theta)
    &=
    \begin{pmatrix}
        e^{-i\theta/2}&0\\
        0&e^{i\theta/2}
    \end{pmatrix},
\end{align}
it suffices to approximate the single complex number
$\frac12e^{i\theta/2}$. We express this number using four nonnegative
dyadic weights associated with the fixed unitaries
\begin{align}
    I,\qquad -I,\qquad -iZ,\qquad iZ.
\end{align}
Because the four weights sum to one, they can be realized as the
relative sizes of four intervals partitioning a uniform superposition
over $2^m$ computational-basis addresses. If the corresponding
interval sizes are $n_0,n_1,n_2,n_3$, then the resulting distinguished
block is
\begin{align}
    \frac1{2^m}
    \left(
        n_0 I-n_1 I-i n_2 Z+i n_3 Z
    \right)
    &=
    \begin{pmatrix}
        \overline z&0\\
        0&z
    \end{pmatrix},
    \qquad
    z:=
    \frac{(n_0-n_1)+i(n_2-n_3)}{2^m}.
\end{align}
Thus choosing
$z\approx\frac12e^{i\theta/2}$ yields an approximate block encoding of
$\frac12R_z(\theta)$. The normalization by $1/2$ is chosen precisely so
that one round of oblivious amplitude amplification maps the ideal
block to $R_z(\theta)$.

The construction therefore has three ingredients. First, we choose
appropriate dyadic interval sizes $n_0,n_1,n_2,n_3$. Second, we
implement the corresponding interval-controlled SELECT operation using
reversible comparisons. Finally, we amplify the resulting block
encoding. As we will see, the synthesis depth is consequently governed
by the depth required for a constant number of reversible comparisons
and computational-basis reflections.

\subsection{Dyadic Approximation of \texorpdfstring{$\frac{1}{2}e^{i\theta/2}$}{the Complex Value}}

We begin with the purely numerical part of the construction, i.e. choosing
four integer interval sizes whose associated dyadic complex coefficient
approximates $\frac12e^{i\theta/2}$. The implementation of these
intervals as a block encoding is deferred to the later
subsections.

\begin{lemma}[Integer Approximation] \label{thm:int_approx}
Fix $\theta \in \mathbb{R}$ and $\eps\in(0,1)$. For any integer
$m \geq \lceil\log_2(2\sqrt{2}/\eps)\rceil$, there exist nonnegative
integers $n_0,n_1,n_2,n_3$ satisfying
\begin{align}
    n_0+n_1+n_2+n_3=2^m
\end{align}
such that
\begin{align}
    \left|
        \frac{(n_0-n_1)+i(n_2-n_3)}{2^m}
        -\frac{1}{2}e^{i\theta/2}
    \right|
    \leq \eps.
\end{align}
\end{lemma}

\begin{proof}[Proof of \Cref{thm:int_approx}]
The proof has two steps. We first express
$\frac12e^{i\theta/2}$ exactly via four nonnegative coefficients
that sum to one. We then round these coefficients to multiples of
$2^{-m}$ while preserving their total sum.

To begin, we define
\begin{align}
    p_0&:=\frac{1+\cos(\theta/2)}{4},
    &
    p_1&:=\frac{1-\cos(\theta/2)}{4},
    \notag\\
    p_2&:=\frac{1+\sin(\theta/2)}{4},
    &
    p_3&:=\frac{1-\sin(\theta/2)}{4},
\end{align}
such that the differences $p_0-p_1$ and
$p_2-p_3$ encode the real and imaginary parts of the desired
complex number, i.e.
\begin{align}
    p_0-p_1
    &=
    \frac12\cos(\theta/2),
    &
    p_2-p_3
    &=
    \frac12\sin(\theta/2).
\end{align}
At the same time, each $p_j$ is nonnegative and $p_0+p_1+p_2+p_3=1$.
Hence,
\begin{align}
    (p_0-p_1)+i(p_2-p_3)
    =
    \frac12e^{i\theta/2}.
    \label{eq:exact-four-phase-decomposition}
\end{align}
Thus it remains only to approximate the probability distribution
$(p_0,p_1,p_2,p_3)$ by one whose probabilities are integer multiples
of $2^{-m}$.

Set $Q:=2^m$. We choose nonnegative integers
$n_0,n_1,n_2,n_3$ summing exactly to $Q$ by the standard
largest-remainder rounding procedure. Namely, first set
\begin{align}
    a_j:=\lfloor Qp_j\rfloor,
\end{align}
and let
\begin{align}
    r:=Q-\sum_{j=0}^3 a_j\in\{0,1,2,3\}.
\end{align}
Increase by one the $r$ values $a_j$ having the largest fractional
parts of $Qp_j$, and call the resulting integers $n_j$. By
construction,
\begin{align}
    \sum_{j=0}^3 n_j=Q
    \qquad\text{and}\qquad
    |n_j-Qp_j|<1
    \label{eq:coordinate-rounding-error}
\end{align}
for every $j\in\{0,1,2,3\}$. The real and imaginary parts of the desired approximation therefore
each incur error less than $2/Q$:
\begin{align}
    \left|
        \frac{n_0-n_1}{Q}-(p_0-p_1)
    \right|
    &<
    \frac{2}{Q},
    \\
    \left|
        \frac{n_2-n_3}{Q}-(p_2-p_3)
    \right|
    &<
    \frac{2}{Q}.
\end{align}
Combining these two errors with
\Cref{eq:exact-four-phase-decomposition} gives
\begin{align}
    \left|
        \frac{(n_0-n_1)+i(n_2-n_3)}{2^m}
        -
        \frac12e^{i\theta/2}
    \right|
    &<
    \frac{2\sqrt2}{2^m}
    \leq
    \eps,
\end{align}
where the last inequality follows from
$m\geq\lceil\log_2(2\sqrt2/\eps)\rceil$.
\end{proof}

\subsection{Computing Value Thresholding Function \texorpdfstring{$\thresh(\x)=\mathbf{1}[\val(\x)< M]$}{f}}

For an $m$-bit binary string $\x=x_{m-1}\cdots x_1x_0\in\{0,1\}^m$
where $x_{m-1}$ is the most significant bit, define the integer encoded by the Boolean string as
\begin{equation}
    \val(\x):=\sum_{j=0}^{m-1}2^j x_j.
\end{equation}
For any hardwired integer $M\in\{0,1,\ldots,2^m\}$, we define the Boolean value thresholding function\footnote{This is not the typical thresholding function studied in circuit complexity, which instead considers whether the Hamming weight of the $m$-bit string $\x$ is less than an integer $M \in [m]$, i.e. $\text{thresh}_M^m(x)=\mathbf{1}[|\x|<M]$.} as
\begin{equation} \label{eq:lt-definition}
    \thresh(x) := \mathbf{1}[\val(\x)<M].
\end{equation}
A reversible circuit cleanly computes
$\thresh$ if, for every
$\x\in\{0,1\}^m$ and $b\in\{0,1\}$, it implements
\begin{equation}
    \ket{\x}\ket{b}\ket{0^a} \longmapsto \ket{\x}\ket{b\oplus \thresh(\x)}
    \ket{0^a}
    \label{eq:clean-lt}
\end{equation}
for some number of ancillae $a$. By linearity, this clean
circuit computes the value threshold coherently on arbitrary
superpositions.

\subsubsection{\texorpdfstring{$O(\log(m))$-Depth, $O(m)$-Space Implementation via $\{X, \cnot,\toff\}$}{Log-Depth Implementation}}

We first implement $\thresh$ exactly and cleanly over $\{X,\cnot,\toff\}$ in depth $O(\log m)$ and space $O(m)$. Standard Clifford+$T$ decompositions give the same asymptotic bounds over $\{H,T,CZ\}$.

\begin{lemma}[$\{X, \cnot, \toff\}$ Value Thresholding]
\label{lem:comparison-no-fanout}
For every $m\geq 1$ and every hardwired
$M\in\{0,1,\ldots,2^m\}$, the value thresholding function
$\thresh$ can be computed exactly and cleanly by a
reversible circuit containing only the gates $\{X, \cnot, \toff\}$ with $O(\log m)$ depth, $O(m)$ size, and $O(m)$ space.
\end{lemma}

\begin{proof}[Proof of \Cref{lem:comparison-no-fanout}]
The usual way to compare two binary integers is to scan from the most
significant bit downward until finding the first position at which they
differ. This is inherently sequential if implemented directly. Our
construction parallelizes this procedure using a balanced binary tree. The key observation is that, for any contiguous block of bits, we do
not need to remember the values of the two substrings themselves. It
suffices to retain a two-bit \emph{comparison summary} recording whether
the two substrings are equal and whether the substring of $\x$ is
smaller than the corresponding substring of $\M$. Two neighboring
summaries can then be combined in constant depth to obtain the summary
of their union. Repeating this combination in a balanced tree yields
the desired logarithmic depth.

The cases $M=0$ and $M=2^m$ give the constant functions zero and one,
respectively, so assume $0<M<2^m$. Set
$\ell:=\lceil\log_2m\rceil$ and $m':=2^\ell$, and pad $\x$ and the
binary expansion $\M$ of $M$ with leading zeroes to length $m'$.
For $0\leq r\leq\ell$ and $0\leq j<m'/2^r$, define the contiguous
block
\begin{align}
    B_{r,j}
    &:=
    \{j2^r,\ldots,(j+1)2^r-1\}.
\end{align}
Thus level $r$ consists of blocks of $2^r$ consecutive bits. For each
block, let
\begin{align}
    \val_{r,j}(\x)
    &:=
    \sum_{t=0}^{2^r-1}2^t x_{j2^r+t},
    \label{eqn:int}
\end{align}
and define its comparison summary by
\begin{align}
    E_{r,j}
    &:=
    \mathbf1[\val_{r,j}(\x)=\val_{r,j}(\M)],
    \label{eq:block-equal}\\
    L_{r,j}
    &:=
    \mathbf1[\val_{r,j}(\x)<\val_{r,j}(\M)].
    \label{eq:block-less}
\end{align}
The pair $(E_{r,j},L_{r,j})$ therefore encodes all three possible
relations between the two block values:
\[
    (1,0)\equiv {=},
    \qquad
    (0,1)\equiv {<},
    \qquad
    (0,0)\equiv {>}.
\]
At the root, the block contains the entire padded input, so
\begin{equation}
    L_{\ell,0}=\thresh(\x).
    \label{eq:root-lt}
\end{equation}

For a one-bit block, the comparison summary is immediate:
\begin{equation}
    (E_{0,j},L_{0,j})
    =
    \begin{cases}
        (\neg x_j,0),&M_j=0,\\
        (x_j,\neg x_j),&M_j=1.
    \end{cases}
    \label{eq:lt-base-case}
\end{equation}
Because $\M$ is hardwired, all leaf summaries can be computed in
parallel and in constant depth using only $X$ and CNOT gates.

For merging comparison summaries, consider two adjacent child blocks whose union forms a parent block.
The child with index $2j+1$ contains the more significant bits, while
the child with index $2j$ contains the less significant bits. Their
summaries determine the parent summary via
\begin{align}
    E_{r+1,j}
    &=
    E_{r,2j+1}\wedge E_{r,2j},
    \label{eq:eq-recursion}\\
    L_{r+1,j}
    &=
    L_{r,2j+1}
    \vee
    \bigl(E_{r,2j+1}\wedge L_{r,2j}\bigr).
    \label{eq:lt-recursion}
\end{align}
These equations are exactly the usual lexicographic comparison rule.
The two parent blocks are equal only when both corresponding halves are
equal. For the less-than relation, the more significant halves decide
the comparison whenever they differ. Only when those halves are equal
does the comparison defer to the less significant halves. This is the central reason that the comparison can be parallelized. Namely,
once a block has been replaced by its two-bit summary, its internal
bits are irrelevant to all subsequent levels of the computation. For every block, equality and strict inequality are mutually exclusive,
so
\begin{equation}
    E_{r,j}\wedge L_{r,j}=0.
    \label{eq:eq-lt-disjoint}
\end{equation}
In particular, the two alternatives on the right-hand side of
\Cref{eq:lt-recursion} can never simultaneously equal one. Hence their
OR is also their XOR:
\begin{equation}
    L_{r+1,j}
    =
    L_{r,2j+1}
    \oplus
    \bigl(E_{r,2j+1}\wedge L_{r,2j}\bigr).
    \label{eq:lt-recursion-xor}
\end{equation}
This form is convenient because it can be implemented reversibly
without computing an OR gate.

\paragraph{Circuit implementation.}
We now evaluate this comparison tree bottom-up. For every parent block,
append two fresh qubits initialized to zero to store
$(E_{r+1,j},L_{r+1,j})$. One Toffoli controlled by
$E_{r,2j+1}$ and $E_{r,2j}$ computes the parent equality bit.
To compute the parent less-than bit, first CNOT
$L_{r,2j+1}$ into its output qubit and then apply a Toffoli controlled
by $E_{r,2j+1}$ and $L_{r,2j}$. By
\Cref{eq:lt-recursion-xor}, the resulting target contains exactly
$L_{r+1,j}$.

The equality Toffoli and the first CNOT act on disjoint qubits and may
be performed simultaneously. The second Toffoli then forms the next
layer. Moreover, different parent blocks at the same tree level depend
on disjoint child summaries, so all merges at that level run in
parallel. Thus each level of the comparison tree costs only constant
depth. After $\ell$ merge levels, the root bit
$L_{\ell,0}=\thresh(\x)$ is available. We CNOT this bit into the
designated output register and then reverse the entire tree computation.
This restores every summary and ancilla to $\ket{0}$, while
preserving the output.

\paragraph{Resource analysis.}
The balanced tree has $m'$ leaves and $m'-1$ internal nodes. Each node
stores only a constant-size summary and requires constantly many gates,
so the forward computation uses $O(m')=O(m)$ gates and ancillae.
There are $\ell=\lceil\log_2m\rceil$
merge levels, each of constant depth. Reversing the computation changes
the depth and gate count by only a constant factor. Hence the resulting
clean reversible circuit requires $O(\log(m+1))$-depth $O(m)$-size, and $O(m)$-ancillae, as claimed.
\end{proof}

\subsubsection{Constant-Depth Implementation with Large Fan-Out and Toffoli}

\newcommand{\arity}{q_m(k)}

We next give an exact and clean implementation of $\thresh$ in constant depth, by leveraging many-qubit \FANOUT~and \toff~gates.

\begin{lemma}[Value Thresholding with Wide Gates]
\label{lem:comparison-fanout}
For every integer $m\geq 1$ and constant integer $k\geq 1$, let
\begin{align}
    b_{\max}
    &:=
    \left\lceil
    m^{2^{k-1}/(2^k-1)}
    \right\rceil.
\end{align}
For every $M\in\{0,1,\ldots,2^m\}$, the value
thresholding function $\thresh$ can be computed exactly and
cleanly by a reversible circuit containing only $X$ gates,
\FANOUT~gates with at most $b_{\max}$ targets, and \toff~gates
with at most $b_{\max}$ controls. The circuit uses $O(k)$ depth, $O_k(m)$ size, and $O_k(
m^{1+1/(2^k-1)}
)$ space.
\end{lemma}

\begin{proof}[Proof of \Cref{lem:comparison-fanout}]
As in the proof of \Cref{lem:comparison-no-fanout}, the cases
$M=0$ and $M=2^m$ are immediate. We therefore assume that
$0<M<2^m$.

We use the same block-comparison strategy as in
\Cref{lem:comparison-no-fanout}, but replace each binary merge by
a wider merge. For every level $i\in\{1,\ldots,k\}$, define the
branching factor
\begin{align}
b_i
&:=
\left\lceil
m^{2^{i-1}/(2^k-1)}
\right\rceil,
\label{eq:k-level-branching-factor}
\end{align}
and define the corresponding block widths by $w_0:=1$ (for the base case $i=0$) and, for every $i>0$,
\begin{align}
w_i
:=
b_iw_{i-1}
=
\prod_{t=1}^{i}b_t.
\label{eq:k-level-block-width}
\end{align}
We pad $\x$ and $\M$ with leading zeroes so that their common
length is $m':=w_k$.
Since
\begin{align}
\sum_{i=1}^{k}
\frac{2^{i-1}}{2^k-1}
&=
1,
\end{align}
and $\lceil z\rceil\leq 2z$ for $z\geq 1$, we have $m\leq m' \leq 2^k m$. Thus, for fixed $k$, the padding changes the input length by only a constant factor.

For every $i\in\{0,\ldots,k\}$ and
$j\in\{0,\ldots,m'/w_i-1\}$, define
\begin{align}
B_{i,j}
&:=
\{jw_i,jw_i+1,\ldots,(j+1)w_i-1\}.
\label{eq:k-level-block}
\end{align}
We define $\val_{i,j}$, $E_{i,j}$, and $L_{i,j}$ exactly as in
the proof of \Cref{lem:comparison-no-fanout}, using the blocks
$B_{i,j}$ above. In particular, the desired threshold value is simply $\thresh(\x)=L_{k,0}(\x,\M)$.

\paragraph{Base case)}
The one-bit summaries $(E_{0,j},L_{0,j})$ are computed exactly as
in \Cref{eq:lt-base-case}. Since $\M$ is hardwired, all base
summaries can be computed in constant depth using $X$ and CNOT
gates, where a CNOT is a \FANOUT~gate with one target.

\paragraph{Recursive step)}
Fix $i\in \{1,\ldots,k\}$. Each parent block $B_{i,j}$ is the
union of the $b_i$ consecutive child blocks:
\begin{align}
B_{i,j} = B_{i-1,jb_i} \cup B_{i-1,jb_i+1} \cup \ldots \cup B_{i-1,jb_i+b_i-1},
\end{align}
where larger child indices correspond to more significant
sub-blocks. The parent summaries satisfy
\begin{align}
E_{i,j}
&=
\bigwedge_{s=0}^{b_i-1}
E_{i-1,jb_i+s},
\label{eq:k-level-equality-recursion}
\\
L_{i,j}&=
\bigvee_{s=0}^{b_i-1}
\left(
L_{i-1,jb_i+s}
\wedge
\bigwedge_{t=s+1}^{b_i-1}
E_{i-1,jb_i+t}
\right),
\end{align}
where an empty conjunction is interpreted as one. The first
identity holds because the parent blocks are equal exactly when
all child blocks are equal. The second follows from the same
most-significant-disagreement argument used in
\Cref{lem:comparison-no-fanout}---i.e. the parent block is smaller
exactly when its most significant unequal child is smaller.

\paragraph{Circuit implementation.}
At level $i$, all parent merges are performed in parallel. For a
fixed parent $B_{i,j}$, append two fresh qubits that will store
$(E_{i,j},L_{i,j})$. For each $s\in\{0,\ldots,b_i-1\}$, define the intermediate clause
\begin{align}
C_{i,j,s}
&:=
L_{i-1,jb_i+s}
\wedge
\bigwedge_{t=s+1}^{b_i-1}
E_{i-1,jb_i+t},
\label{eq:k-level-clause}
\end{align}
such that the output less-than bit can be composed as
\begin{align}
L_{i,j}
&=
\bigvee_{s=0}^{b_i-1}
C_{i,j,s}.
\label{eq:k-level-clause-or}
\end{align}

The equality bit $E_{i-1,jb_i+t}$ is used in the parent equality
computation and in the $t$ clauses corresponding to less
significant children. We therefore apply a \FANOUT gate with $t$
targets, producing one physical copy for each additional
simultaneous use. The total number of copied equality bits for
one parent merge is
\begin{align}
\sum_{t=0}^{b_i-1}t
&=
\frac{b_i(b_i-1)}{2}.
\label{eq:k-level-copy-count}
\end{align}
Using these copies, we compute $E_{i,j}$ with one
$b_i$-controlled Toffoli gate and compute the $b_i$ clause bits
$C_{i,j,s}$ with at most $b_i$-controlled Toffoli gates. Because
the \FANOUT step supplies separate copies of every shared
equality bit, these Toffoli gates act on disjoint qubits and can
be applied in parallel.

To compute $L_{i,j}$, we negate the clause qubits and the fresh
less-than qubit, apply one $b_i$-controlled Toffoli gate, and then
undo the negations on the clause qubits. By De Morgan's law, the
target is mapped according to
\begin{align}
0
&\longmapsto
1
\oplus
\bigwedge_{s=0}^{b_i-1}
\neg C_{i,j,s}=
\bigvee_{s=0}^{b_i-1}
C_{i,j,s}
=
L_{i,j}.
\end{align}
We then reverse the clause computations and the local Fan-Out
operations. This returns all clause and copied-control ancillae
to $\ket{0}$ while preserving the child summaries and the newly
computed parent summary.

Since the branching factors are non-decreasing, every FAN-OUT gate
has at most $b_{\max}$ targets and every Toffoli gate has at
most $b_{\max}$ controls, where
\begin{align}
b_{\max}
&:=
b_k
=
\left\lceil
m^{2^{k-1}/(2^k-1)}
\right\rceil.
\end{align}
After completing all $k$ levels, we copy
$L_{k,0}=\thresh(\x)$ into the designated output qubit and run
the complete hierarchy in reverse. This returns every remaining
work qubit to $\ket{0}$ while preserving the threshold output.

\paragraph{Resource analysis.}
Each merge uses a constant number of layers, $O(b_i)$ gates, and
$O(b_i^2)$ temporary qubits---i.e., the $b_i$ clause bits and the
$b_i(b_i-1)/2$ control copies counted in \Cref{eq:k-level-copy-count}.
There are $m'/w_i$ merges at level $i$, so the gate count at that
level is $O(m'b_i/w_i)=O(m'/w_{i-1})=O_k(m)$, while its temporary
workspace is
\begin{align}
    O\!\left(\frac{m'}{w_i}b_i^2\right)
    &=O\!\left(\frac{m'b_i}{w_{i-1}}\right)
      =O_k\!\left(m^{1+1/(2^k-1)}\right).
    \label{eq:k-level-space}
\end{align}
The last equality uses
$b_i=O(m^{2^{i-1}/(2^k-1)})$ and
$w_{i-1}\geq m^{(2^{i-1}-1)/(2^k-1)}$.
Temporary ancilla workspace is cleaned and reused between levels. The
persistent summaries occupy
$2\sum_{i=0}^k m'/w_i=O_k(m)$ qubits. Including the base level,
output copy, and reverse computation therefore gives depth $O(k)$,
gate count $O_k(m)$, and workspace
$O_k(m^{1+1/(2^k-1)})$.

\end{proof}

\subsection{Oblivious Amplitude Amplification}
We will now give brief background on oblivious amplitude amplification (OAA), originally introduced and leveraged in the works of \cite{BCCKS14,BCCKS15}. Conceptually, OAA generalizes standard amplitude amplification to settings in which the desired output state itself need not be known. Namely, it suffices that the desired component can be identified by a fixed ancilla flag. By alternating the underlying unitary with reflections about this flagged ``good'' subspace, OAA coherently amplifies its amplitude without requiring any reflection about, or knowledge of, the state stored in the data register. This obliviousness is precisely what makes OAA useful for block encodings, where the desired operator appears in a distinguished ancilla block. We will then show that the same procedure can be applied robustly to the approximate block encoding required by our construction.

\subsubsection{\texorpdfstring{Exact OAA: $\frac12U\mapsto U$ up to a fixed sign}{Exact OAA}}

Let $A$ be an $a$-qubit ancilla register and $S$ the system register. Define the isometry that embeds the system into the all-zero ancilla subspace, together with the corresponding projector and reflection, as
\begin{align}
    V_0&:=\ket{0^a}_A\otimes I_S,
    &\Pi&:=V_0V_0^\dagger,
    &\mathcal R_\Pi&:=I-2\Pi.
\end{align}
Thus, $\mathcal R_\Pi$ applies a phase of $-1$ precisely when the ancilla register is $\ket{0^a}$. Equivalently,
\begin{align}
    \mathcal R_\Pi
    &=\left[X^{\otimes a}
       \left(I-2\ketbra{1^a}{1^a}\right)X^{\otimes a}\right]\otimes I_S.
\end{align}

For a unitary $W$, let its distinguished block be
$B:=\be(W):=V_0^\dagger W V_0$. We use the following convention for one step of oblivious amplitude amplification:
\begin{align}
    \mathsf{Amp}(W)
    &:=W\mathcal R_\Pi W^\dagger\mathcal R_\Pi W
      =W-2\Pi W-2W\Pi+4W\Pi W^\dagger\Pi W.
    \label{seq:def-one-step-oaa}
\end{align}
Taking the distinguished block yields the exact cubic transformation
\begin{align}
    \be\bigl(\mathsf{Amp}(W)\bigr)
    &=-3B+4BB^\dagger B.
    \label{eq:oaa-cubic-identity}
\end{align}
In particular, if $B=U/2$ for a unitary $U$, then the amplified block is $-U$. The minus sign is a fixed global phase, independent of the input, and may be removed exactly by composing with $(XZ)^2=-I$ on any one qubit.

\begin{proposition}[Exact one-step oblivious amplitude amplification]
\label{prop:exact-one-step-oaa}
If $U$ is unitary and $\be(W)=U/2$, then
\begin{align}
    \mathsf{Amp}(W)V_0=-V_0 U.
    \label{eqn:desired_amp}
\end{align}
In particular, the ancillas return exactly to $\ket{0^a}$, without measurement or postselection.
\end{proposition}

\begin{proof}
By \Cref{eq:oaa-cubic-identity},
$\be(\mathsf{Amp}(W))=-3U/2+U/2=-U$. Hence, on every unit input, the projection of the amplified output onto the all-zero ancilla subspace already has norm one. Since $\mathsf{Amp}(W)$ is unitary, there can be no component orthogonal to this subspace. This gives
$\mathsf{Amp}(W)V_0=-V_0U$, which is the standard one-step OAA argument~\cite{BCCKS14,BCCKS15}.
\end{proof}

\subsubsection{\texorpdfstring{Robust OAA for an approximate $\frac12R_z(\theta)$ block}{Robust OAA}}

Our block encoding is only approximate, so the exact argument above is not by itself sufficient. We must control both the error within the distinguished block and the amplitude that leaks outside the all-zero ancilla subspace. The special diagonal form of our block allows both quantities to be bounded linearly in the original approximation error.

\begin{lemma}[Robust amplification of an approximate $\frac12R_z(\theta)$ block]
\label{lem:robust-oaa-rz}
Let $\theta\in\mathbb R$, $\delta\in[0,1/8]$, and let $W_z$ be a unitary with distinguished block
\begin{align}
    \be(W_z)=B_z:=\begin{pmatrix}\overline z&0\\0&z\end{pmatrix},
    \qquad
    \left|z-\tfrac12e^{i\theta/2}\right|\leq\delta.
    \label{eq:z-approximation-assumption}
\end{align}
Then the fixed phase $\varphi=\pi$ satisfies
\begin{align}
    \left\|\be\bigl(\mathsf{Amp}(W_z)\bigr)+R_z(\theta)\right\|
       &\leq\frac{13}{4}\delta,
    \label{eq:amplified-block-error}\\
    \left\|\mathsf{Amp}(W_z)V_0+V_0 R_z(\theta)\right\|_{\mathrm{op}}
       &\leq\sqrt{13}\,\delta\leq5\delta.
    \label{eq:coherent-rz-implementation}
\end{align}
Moreover, for every unit input, the squared norm of the component outside the all-zero ancilla subspace is at most $13\delta^2$.
\end{lemma}

\begin{proof}
Write $z=re^{i\phi}$, set $t:=\theta/2$, and define
$V_\phi:=\operatorname{diag}(e^{-i\phi},e^{i\phi})$, so that
$B_z=rV_\phi$. Since $z$ is $\delta$-close to $\frac12e^{it}$, the reverse triangle inequality gives
\begin{align}
    \Delta:=r-\tfrac12,\qquad |\Delta|\leq\delta,
    \qquad \tfrac38\leq r\leq\tfrac58.
    \label{eq:r-range}
\end{align}
Let
$e:=|re^{i\phi}-e^{it}/2|$. Separating the radial and angular errors gives
\begin{align}
    e^2=\Delta^2+r\bigl(1-\cos(\phi-t)\bigr)
       =\Delta^2+\frac r2|e^{i\phi}-e^{it}|^2.
    \label{eq:radial-angular-decomposition}
\end{align}
Since $e\leq\delta$, this implies
\begin{align}
    |e^{i\phi}-e^{it}|
    &\leq\sqrt{2/r}\,\delta
       \leq\frac4{\sqrt3}\delta.
    \label{eq:phase-error-bound}
\end{align}

We first bound the distinguished block after amplification. Since
$B_zB_z^\dagger=r^2I$, \Cref{eq:oaa-cubic-identity} gives
\begin{align}
    \be\bigl(\mathsf{Amp}(W_z)\bigr)&=-\alpha(r)V_\phi,
    \qquad \text{where~~} \alpha(r):=3r-4r^3.
    \label{eq:amplified-kz-exact}
\end{align}
The polynomial $\alpha(r)$ equals one at the ideal value $r=1/2$. Writing $r=1/2+\Delta$ gives
\begin{align}
    1-\alpha(r)=\Delta^2(6+4\Delta),\qquad
    0\leq1-\alpha(r)\leq\frac{13}{2}\delta^2,
    \qquad 0<\alpha(r)\leq1.
    \label{eq:q-error-bound}
\end{align}
Thus the amplification suppresses the radial error to second order, while the remaining first-order contribution comes from the phase error. Using \Cref{eq:phase-error-bound},
\begin{align}
    \|\alpha(r)V_\phi-R_z(\theta)\|
    &\leq\frac{13}{2}\delta^2+\frac4{\sqrt3}\delta
     \leq\left(\frac{13}{16}+\frac4{\sqrt3}\right)\delta
     <\frac{13}{4}\delta.
    \label{eq:proof-amplified-block-error}
\end{align}
This proves \Cref{eq:amplified-block-error}, with the fixed global sign $-1$.

It remains to control the complete output isometry, including leakage outside the distinguished ancilla subspace. Define the sign-corrected isometry
$K:=-\mathsf{Amp}(W_z)V_0$ and set $U:=R_z(\theta)$. Both $K$ and $V_0U$ are isometries, while
$V_0^\dagger K=\alpha(r)V_\phi$. Therefore
\begin{align}
    (K-V_0 U)^\dagger(K-V_0 U)
    &=\bigl(2-2\alpha(r)\cos(\phi-t)\bigr)I.
\end{align}
Substituting the radial--angular decomposition from
\Cref{eq:radial-angular-decomposition} gives
\begin{align}
    \|K-V_0 U\|^2
    &=(6-8r^2)e^2+8(r+\tfrac12)^2\Delta^2\\
    &\leq\bigl(6-8r^2+8(r+\tfrac12)^2\bigr)\delta^2\\
    &=(12+8\Delta)\delta^2\leq13\delta^2.
\end{align}
Both coefficients in the first line are nonnegative for
$r\in[3/8,5/8]$. Taking square roots proves
\Cref{eq:coherent-rz-implementation}. In particular, this is an operator-norm statement with one fixed global phase, so it holds uniformly for all inputs, including inputs entangled with an arbitrary reference system.

Finally, for any unit vector $\ket\psi$, decompose the amplified output into its all-zero-ancilla component and its orthogonal remainder, i.e.
\begin{align}
    \mathsf{Amp}(W_z)V_0\ket\psi
      &=-V_0\alpha(r)V_\phi\ket\psi+\ket{\mathrm{bad}_\psi},
    \label{eq:oaa-good-bad-decomposition}
\end{align}
where $\Pi\ket{\mathrm{bad}_\psi}=0$. By unitarity and orthogonality,
\begin{align}
    \|\ket{\mathrm{bad}_\psi}\|^2
       &=1-\alpha(r)^2
       \leq2(1-\alpha(r))
       \leq13\delta^2.
    \label{eq:oaa-leakage-bound}
\end{align}
This gives the claimed leakage bound.
\end{proof}

\subsubsection{Coherent Implementation of the OAA Procedure}

The OAA procedures in
\Cref{prop:exact-one-step-oaa} and \Cref{lem:robust-oaa-rz} are fully coherent, i.e.
they require neither measurement nor postselection. Their only
additional primitive is the reflection about the all-zero ancilla
subspace,
\begin{align}
    \mathcal R_0
    &:=
    I-2\ketbra{0^a}{0^a}
    =
    X^{\otimes a}
    \left(I-2\ketbra{1^a}{1^a}\right)
    X^{\otimes a}.
\end{align}
Thus, implementing $\mathcal R_0$ reduces to applying a phase of $-1$
to the all-one basis state. The opposite reflection
$-\mathcal R_0$ may be used equivalently, since the OAA iterate contains
two reflections and the two additional minus signs cancel.

With generalized Toffoli gates, the all-one phase
$I-2\ketbra{1^a}{1^a}$ is simply an $a$-qubit controlled-$Z$, obtained
by Hadamard-conjugating the target of a generalized Toffoli. Hence the
reflection has constant depth in $\QAC$. With bounded-arity gates, we
instead compute the AND of the $a$ ancilla bits using a balanced
Toffoli tree, apply a $Z$ gate to the resulting root bit, and uncompute
the tree. This gives an exact implementation of $\mathcal R_0$ using
depth $O(\log(a+1))$ and $O(a)$ gates and workspace.

There is one important distinction between the exact and approximate
settings. In the exact case of
\Cref{prop:exact-one-step-oaa}, the amplification ancillas return
exactly to $\ket{0^a}$. In the approximate setting of
\Cref{lem:robust-oaa-rz}, a small component may remain outside the
all-zero ancilla subspace. Namely, by \Cref{eq:oaa-leakage-bound}, its squared
norm is at most $13\delta^2$. Equivalently, measuring the ancillae
would produce a nonzero outcome with probability at most
$13\delta^2$. We do not perform such a measurement. Instead, the
construction is analyzed through the coherent isometry guarantee of
\Cref{eq:coherent-rz-implementation}, which controls the complete
joint state of the system and ancillas.

\subsection{Efficient Synthesis of Single-Qubit Unitaries}

We now assemble the preceding block-encoding and amplification primitives into complete single-qubit synthesis circuits. We first treat $R_z$ rotations and then extend the construction to arbitrary single-qubit unitaries.

\subsubsection{Approximate Synthesis of Arbitrary \texorpdfstring{$\Rz$}{Rotation} Gates}

For $R_z$ rotations, the dyadic approximation determines a block encoding of $\frac12R_z(\theta)$, reversible comparisons implement the corresponding SELECT operation, and one step of oblivious amplitude amplification recovers $R_z(\theta)$. With bounded-width gates, the comparison and reflection steps have depth $O(\log m)$ for $m=\Theta(\log(1/\eps))$, yielding the following $O(\log\log(1/\eps))$-depth construction. 
\begin{theorem}[Approximate $\Rz$ Synthesis with Bounded-Width Gates]
\label{thm:rz-synthesis-qnc}
Fix $\theta\in\mathbb{R}$ and $\eps\in(0,1)$. Define the bounded-width gate set
\begin{align}
    \mathcal{G}_{\textnormal{\QNC}}
    :=
    \{
        H,X,Z,S,S^\dagger,\cnot,CZ,\toff
    \},
\end{align}
where $\toff$ denotes the ordinary three-qubit Toffoli gate.
Then, there exists an
angle-dependent, fully unitary circuit $C_{\theta,\eps}$ over the
gate set $\mathcal{G}_{\textnormal{\QNC}}$, using
$a=O(\log(1/\eps))$ ancilla qubits initialized to $\ket{0^a}$,
such that the single input-independent phase $\varphi=\pi$ satisfies
\begin{align}
    \left\|
        C_{\theta,\eps}(\ket{0^a}\otimes I)
        +
        \ket{0^a}\otimes\Rz
    \right\|_{\mathrm{op}}
    &\leq \eps.
    \label{eq:qnc-rz-coherent-error}
\end{align}
The circuit $C_{\theta,\eps}$  has depth $O(\log\log(1/\eps))$, gate count
$O(\log(1/\eps))$, and space $O(\log(1/\eps))$.
\end{theorem}

\begin{proof}[Proof of \Cref{thm:rz-synthesis-qnc}]
Define, respectively, the error and space parameters:
\begin{align}
    \delta
    &:=
    \min\left\{
        \frac{\eps}{5},
        \frac18
    \right\}
    \quad \text{and} \quad
    m:=
    \left\lceil
        \log_2\left(
            \frac{2\sqrt{2}}{\delta}
        \right)
    \right\rceil.
    \label{eq:rz-choice-m}
\end{align}
Let $A$ be an $m$-qubit address register and let $s$ denote the
single system qubit. The input is $\ket{0^a}\ket{\psi}$ for an arbitrary system state $\ket{\psi}$, where $a=O(m)$.

\paragraph{Precomputation.} By \Cref{thm:int_approx}, there exist nonnegative integers
$n_0,n_1,n_2,n_3$ satisfying
\begin{align}
    n_0+n_1+n_2+n_3
    =
    2^m
\end{align}
such that the complex number
\begin{align}
    z
    &:=
    \frac{
        (n_0-n_1)+i(n_2-n_3)
    }{2^m}
\end{align}
satisfies
\begin{align}
    \left|
        z-\frac12e^{i\theta/2}
    \right|
    \leq
    \delta.
    \label{eq:rz-z-approximation}
\end{align}
The integers $n_0,n_1,n_2,n_3$, and hence the comparison
thresholds below, are classically computed based on the angle $\theta$ and hardwired into the circuit. We will also use these integers to define the cumulative thresholds
\begin{align}
    M_0
    &:=
    n_0, \quad 
    M_1:=
    n_0+n_1,
    \quad \text{and}\quad
    M_2:=
    n_0+n_1+n_2.
\end{align}

\paragraph{Circuit Implementation.} With these parameters, we can now describe the circuit implementation. Begin by applying $H^{\otimes m}$ to $m$-qubits initialized to the state $\ket{0^m}$ in the ancilla register to prepare the uniform superposition state
\begin{align}
    \frac{1}{\sqrt{2^m}}
    \sum_{x\in\{0,1\}^m}
    \ket{x}_A.
\end{align}
We will next leverage our Boolean value thresholding circuit from  \Cref{lem:comparison-no-fanout} to compute, for $j\in\{0,1,2\}$,
\begin{align}
    t_j(x)
    &:=
    f^m_{M_j}(x)
    =
    \mathbf{1}[\val(x)<M_j],
\end{align}
into three fresh ancillae initialized to $\ket*{0^3}$ to obtain the state $\ket{t_0}\ket{t_1}\ket{t_2}$. Recall that each threshold computation
uses only $X$, $\cnot$, and $\toff$ gates.Since there are only three thresholds, they may be computed sequentially
using a common clean comparison workspace without changing the
asymptotic resource bounds. Since,
for every $x$,
$t_0(x)\leq t_1(x)\leq t_2(x)$ we can compute the four interval-indicator bits in
constant depth using only $X$ and $\cnot$ gates:
\begin{align}
    e_0
    &:=
    t_0,
    &
    e_1
    &:=
    t_1\oplus t_0,
    &
    e_2
    &:=
    t_2\oplus t_1,
    &
    e_3
    &:=
    1\oplus t_2.
    \label{eq:rz-one-hot-definitions}
\end{align}
Computing these bits into four fresh ancillae gives the one-hot register
$\ket{e_0}\ket{e_1}\ket{e_2}\ket{e_3}$. Exactly one of
$e_0,e_1,e_2,e_3$ equals one, with
\begin{align}
    e_j(x)=1
    \quad\Longleftrightarrow\quad
    x\in I_j,
\end{align}
where
\begin{align}
    I_0
    :=
    [0,n_0),
    \quad
    I_1
    :=
    [n_0,n_0+n_1),
    \quad
    I_2
    :=
    [n_0+n_1,n_0+n_1+n_2),
    \quad
    I_3
    :=
    [n_0+n_1+n_2,2^m).
\end{align}

We next apply a constant-depth SELECT operation controlled by the
one-hot interval register $E=(e_0,e_1,e_2,e_3)$. Specifically, apply
a $Z$ gate to $e_1$, apply an $S^\dagger$ gate to $e_2$ followed by
$CZ_{e_2,s}$, and apply an $S$ gate to $e_3$ followed by
$CZ_{e_3,s}$. No gate is associated with $e_0$. Thus, the SELECT
operation is
\begin{align}
    \operatorname{SELECT}
    &:=
    Z_{e_1}
    \left(
        S^\dagger_{e_2}CZ_{e_2,s}
    \right)
    \left(
        S_{e_3}CZ_{e_3,s}
    \right).
    \label{eq:rz-select-gate-implementation}
\end{align}
We now verify the induced operation on the system register $s$. Since the
interval register $(e_0,e_1,e_2,e_3)$ is one-hot, for each $j\in\{0,1,2,3\}$ and every
system state $\ket{\psi}$,
\begin{align}
    \operatorname{SELECT}
    \left(
        \ket{e_j}_E\ket{\psi}_S
    \right)
    =
    \ket{e_j}_E U_j\ket{\psi}_S,
\end{align}
where
\begin{align}
    U_0
    &:=
    I,
    &
    U_1
    &:=
    -I,
    &
    U_2
    &:=
    -iZ,
    &
    U_3
    &:=
    iZ.
    \label{eq:rz-select-unitaries}
\end{align}
Indeed, conditioned on $e_0=1$ no gate acts nontrivially. On the $e_1=1$
branch, $Z_{e_1}$ contributes the global branch phase $-1$. On the
$e_2=1$ branch, $S^\dagger_{e_2}$ contributes the phase $-i$, while
$CZ_{e_2,s}$ applies $Z$ to the system qubit, yielding $-iZ$.
Similarly, on the $e_3$ branch, $S_{e_3}$ contributes the phase $i$
and $CZ_{e_3,s}$ applies $Z$ to the system qubit, yielding $iZ$.

Let $L$ denote the reversible computation of the threshold and
interval-indicator bits, and write $H_A:=H^{\otimes m}$ on the
address register. The complete block-encoding circuit is
\[
    W_z:=H_A L^\dagger\operatorname{SELECT}L H_A.
\]
For every address $x$, SELECT applies $U_{j(x)}$ without changing
$x$ or its labels, where $j(x)$ is the unique interval containing
$\val(x)$. Thus $L^\dagger$ restores every threshold, indicator, and
comparison-workspace qubit exactly to zero, leaving
\begin{align}
    \frac1{\sqrt{2^m}}
    \sum_{x\in\{0,1\}^m}\ket x_A U_{j(x)}\ket\psi_S
       \otimes\ket{0^{a_{\mathrm{work}}}}
    \label{eq:state-after-select-uncomputation}
\end{align}
before the final Hadamards. Suppressing this clean workspace and
using $\bra{0^m}H_A\ket x=2^{-m/2}$, the distinguished block is
\begin{align}
    \be(W_z)=\frac1{2^m}\sum_{j=0}^3 n_jU_j=\frac{(n_0-n_1)I-i(n_2-n_3)Z}{2^m}
      =\begin{pmatrix}\overline z&0\\0&z\end{pmatrix}.
    \label{eq:rz-block-encoding}
\end{align}
By \Cref{eq:rz-z-approximation}, this block is $\delta$-close to
$\frac12R_z(\theta)$.

We now apply one step of oblivious amplitude amplification. Define
\begin{align}
    C_{\theta,\eps}
    &:=
    \mathsf{Amp}(W_z):= W_z\mathcal{R}_0W_z^\dagger\mathcal{R}_0W_z,
    \label{eq:qnc-main-amplification}
\end{align}
where $\mathcal{R}_0:=I-2\ketbra{0^m}{0^m}_A$.
All comparison and interval workspace is clean at the points at
which $\mathcal{R}_0$ is applied. We may therefore regard $W_z$ as
a unitary on the address and system registers, suppressing the
clean workspace from the block-encoding notation. By \Cref{lem:robust-oaa-rz}, and specifically by choosing the fixed
global phase $\varphi=\pi$, for every unit vector $\ket{\psi}$,
\begin{align}
    \left\|
        C_{\theta,\eps}
        \ket{0^m}_A\ket{\psi}_S
        +
        \ket{0^m}_A\Rz(\theta)\ket{\psi}_S
    \right\|
    \leq
    5\delta
    \leq
    \eps.
    \label{eq:qnc-rz-oaa-error}
\end{align}
The same phase $\varphi=\pi$ applies to every input
$\ket{\psi}$. Appending the clean comparison and interval workspace
to both states in \Cref{eq:qnc-rz-oaa-error} proves
\Cref{eq:qnc-rz-coherent-error}.

It remains to implement $\mathcal{R}_0$. First apply
$X^{\otimes m}$ to the address register. Using a balanced binary
tree of ordinary Toffoli gates, compute the AND of the $m$ address
bits into a clean root ancilla. Apply $Z$ to the root, reverse the
Toffoli tree, and apply $X^{\otimes m}$ again. This implements $I-2\ketbra{0^m}{0^m}_A$
exactly using only $X$, $\toff$, and $Z$ gates.

\paragraph{Resource Analysis.} We conclude by verifying the resource bounds. By
\Cref{lem:comparison-no-fanout}, each clean threshold computation
has $O(\log m)$ depth, $O(m)$ gate count, and  $O(m)$ space.
Computing three thresholds and their inverses changes these bounds
by only a constant factor. The preparation of the address state,
the interval-indicator computation and uncomputation, and the
SELECT operation have constant depth, $O(m)$ gate count, and
$O(m)$ space.

Each reflection $\mathcal{R}_0$ uses a balanced Toffoli tree of
 $O(\log m)$ depth,  $O(m)$ gate count, and $O(m)$ clean workspace.
One amplification step makes three calls to $W_z$ or $W_z^\dagger$
and applies two reflections. Since $m=O(\log(1/\eps))$,
\begin{align}
    \text{depth}(C_{\theta,\eps})
    &=
    O(\log m) = O(\log\log(1/\eps)),
    \\
    \text{size}(C_{\theta,\eps})
    &=
    O(m)=O(\log(1/\eps)),
    \\
    \text{space}(C_{\theta,\eps})
    &=
    O(m)=O(\log(1/\eps)).
\end{align}
\end{proof}

\noindent We now show how the procedure can be made constant-depth with access to multi-qubit gates.

\begin{corollary}[Constant-Depth Approximate $\Rz$ Synthesis with
Wide Gates]
\label{cor:rz-synthesis-qacf}
Fix a constant $\gamma\in(0,1)$, an angle
$\theta\in\mathbb{R}$, and $\eps\in(0,1)$. Define the wide gate set
\begin{align}
    \mathcal{G}_{\textnormal{\QAC}_f}
    :=
    \{
        H,X,Z,S,S^\dagger,\cnot,CZ,\FANOUT,\GTof
    \},
\end{align}
where, here, \FANOUT~denotes up to $O_\gamma(
\log^{1/2+\gamma}(1/\eps))$-size \FANOUT~gate
targets and GToffoli denotes up to $O(\log(1/\eps))$-size Toffoli gates. Then there exists an angle-dependent, fully unitary
circuit $C^{(\gamma)}_{\theta,\eps}$ over
$\mathcal{G}_{\textnormal{\QAC}_f}$, using $a=O_\gamma(\log^{1+\gamma}(1/\eps))$
ancillae initialized to $\ket{0^a}$, such that the same fixed phase
$\varphi=\pi$ works for all inputs, i.e.
\begin{align}
    \left\|
        C^{(\gamma)}_{\theta,\eps}(\ket{0^a}\otimes I)
        +\ket{0^a}\otimes\Rz(\theta)
    \right\|_{\mathrm{op}}
    &\leq\eps.
    \label{eq:qacf-rz-coherent-error}
\end{align}

Overall, the circuit has depth $O_\gamma(1)$, gate count
$O_\gamma(\log(1/\eps))$, and space
$O_\gamma(\log^{1+\gamma}(1/\eps))$.
\end{corollary}

\begin{proof}
Use the construction of \Cref{thm:rz-synthesis-qnc}, with the same
parameters $\delta,m$ from \Cref{eq:rz-choice-m}. Only the
comparison circuits and ancilla reflections are changed.

Choose a fixed integer $k=k(\gamma)$ such that
\begin{equation}
    \frac1{2^k-1}\leq\gamma.
    \label{eq:qacf-choice-k}
\end{equation}
By \Cref{lem:comparison-fanout}, each comparison has depth
$O_\gamma(1)$, gate count $O_\gamma(m)$, and workspace
\begin{equation}
    O_k\!\left(m^{1+1/(2^k-1)}\right)
       =O_\gamma(m^{1+\gamma}).
    \label{eq:qacf-threshold-space}
\end{equation}
Its Fan-Out gates have at most
\begin{equation}
    b_{\max}=\left\lceil m^{2^{k-1}/(2^k-1)}\right\rceil
       =O_\gamma(m^{1/2+\gamma})
    \label{eq:qacf-threshold-arity}
\end{equation}
targets, and its generalized Toffoli gates have at most
$b_{\max}\leq m$ controls. Each ancilla reflection is $R_0=X^{\otimes m}C^{m-1}Z\,X^{\otimes m}$, where $C^{m-1}Z=H_t(C^{m-1}X)H_t$ and $t$ is one address qubit. Thus a reflection has constant depth,
$O(m)$ gates, and no additional workspace.

These substitutions leave the block encoding unchanged. Hence
$C^{(\gamma)}_{\theta,\eps}:=\mathsf{Amp}(W_z)$ satisfies
\Cref{eq:qacf-rz-coherent-error} by \Cref{lem:robust-oaa-rz}, with
error $5\delta\leq\eps$ and the same fixed phase $\varphi=\pi$.
There are only constantly many comparisons, reflections, and calls
to $W_z$ or $W_z^\dagger$. All remaining steps have constant depth
and $O(m)$ gates. The total depth, gate count, and workspace are
therefore $O_\gamma(1)$, $O_\gamma(m)$, and
$O_\gamma(m^{1+\gamma})$, respectively. Substituting
$m=O(\log(1/\eps))$ gives all claimed resource bounds.

\end{proof}

\subsubsection{Approximate Synthesis of Any Single-Qubit Unitary}

We now extend the preceding $R_z$ synthesis results to arbitrary
single-qubit unitaries using the standard Euler-angle decomposition.
The three synthesized rotations may reuse the same workspace; a
telescoping argument below shows that their clean-input errors simply
add, even though the workspace is only approximately restored after
each rotation.

\begin{lemma}[Reduction from Single-Qubit Synthesis to $\Rz$ Synthesis]
\label{lem:single-qubit-to-rz}
Let $\mathcal{G}$ be a gate set containing $H$. Suppose that, for
every $\theta\in\mathbb{R}$ and $\eta\in(0,1)$, there exists a fully
unitary circuit $C^z_{\theta,\eta}$ over $\mathcal{G}$, using
$a(\eta)$ clean ancillae, with depth $d(\eta)$ and gate count
$s(\eta)$, such that there exists a phase
$\varphi_{\theta,\eta}\in\mathbb{R}$ satisfying
\begin{align}
    \left\|
        C^z_{\theta,\eta}\ket*{0^{a(\eta)}}\ket{\psi}
        -
        e^{i\varphi_{\theta,\eta}}
        \ket*{0^{a(\eta)}}\Rz(\theta)\ket{\psi}
    \right\|
    \leq
    \eta
    \label{eq:rz-synthesis-reduction-assumption}
\end{align}
for every state $\ket{\psi}$. Then, for every single-qubit unitary $U$
and $\eps\in(0,1)$, there exists a fully unitary circuit
$C_{U,\eps}$ over $\mathcal{G}$ using $a(\eps/3)$ clean ancillae such
that, for some $\varphi\in\mathbb{R}$ and every state $\ket{\psi}$,
\begin{align}
    \left\|
        C_{U,\eps}\ket*{0^{a(\eps/3)}}\ket{\psi}
        -
        e^{i\varphi}
        \ket*{0^{a(\eps/3)}}U\ket{\psi}
    \right\|
    \leq
    \eps.
    \label{eq:arbitrary-single-qubit-error}
\end{align}
Moreover,
\begin{align}
    \operatorname{depth}(C_{U,\eps})
    &\leq
    3d(\eps/3)+2,
    \\
    \operatorname{size}(C_{U,\eps})
    &\leq
    3s(\eps/3)+2,
    \\
    \operatorname{space}(C_{U,\eps})
    &=
    a(\eps/3)+1.
    \label{eq:single-qubit-reduction-resources}
\end{align}
\end{lemma}

\begin{proof}
By the Euler-angle decomposition
\cite[Theorem~4.1]{nielsen2010quantum} and
$R_x(\gamma)=HR_z(\gamma)H$, there exist real
$\alpha,\beta,\gamma,\delta$ such that
\begin{equation}
    U=e^{i\alpha}R_z(\beta)H R_z(\gamma)H R_z(\delta).
    \label{eq:zxz-euler-decomposition}
\end{equation}
Set $\eta:=\eps/3$ and let
$V_0:=\ket{0^{a(\eta)}}\otimes I$ denote the embedding into the
all-zero workspace subspace. For $\theta\in\{\beta,\gamma,\delta\}$,
define
\begin{equation}
    \widetilde R_\theta
    :=
    e^{i\varphi_{\theta,\eta}}R_z(\theta).
\end{equation}
Then \Cref{eq:rz-synthesis-reduction-assumption} is equivalently
\begin{equation}
    \left\|
        C^z_{\theta,\eta}V_0-V_0\widetilde R_\theta
    \right\|_{\mathrm{op}}
    \leq\eta.
    \label{eq:rz-clean-isometry-bound}
\end{equation}
Using the same $a(\eta)$-qubit workspace for all three synthesized
rotations, define
\begin{equation}
    C_{U,\eps}
    :=
    C^z_{\beta,\eta}H C^z_{\gamma,\eta}H C^z_{\delta,\eta}.
    \label{eq:arbitrary-single-qubit-circuit}
\end{equation}
A three-step hybrid argument, together with $HV_0=V_0H$, gives
\begin{align}
    &\left\|
        C_{U,\eps}V_0
        -
        V_0\widetilde R_\beta H
        \widetilde R_\gamma H
        \widetilde R_\delta
    \right\|_{\mathrm{op}}\leq
    \left\|
        C^z_{\delta,\eta}V_0
        -
        V_0\widetilde R_\delta
    \right\|_{\mathrm{op}}
    +
    \left\|
        C^z_{\gamma,\eta}V_0
        -
        V_0\widetilde R_\gamma
    \right\|_{\mathrm{op}}
    +
    \left\|
        C^z_{\beta,\eta}V_0
        -
        V_0\widetilde R_\beta
    \right\|_{\mathrm{op}}\leq
    3\eta
    =
    \eps.
    \label{eq:single-qubit-hybrid}
\end{align}
Here unitary invariance of the operator norm ensures that the error from
each replacement is unchanged by the surrounding gates. In particular,
the hybrid compares each synthesized rotation with its ideal action on
the all-zero workspace subspace, so the same workspace may be reused
throughout. Finally,
\begin{align}
    \widetilde R_\beta H
    \widetilde R_\gamma H
    \widetilde R_\delta
    &=
    e^{i(
        \varphi_{\beta,\eta}
        +\varphi_{\gamma,\eta}
        +\varphi_{\delta,\eta}
        -\alpha)}
    U.
\end{align}
Thus \Cref{eq:arbitrary-single-qubit-error} holds with
\begin{align}
    \varphi
    &:=
    \varphi_{\beta,\eta}
    +\varphi_{\gamma,\eta}
    +\varphi_{\delta,\eta}
    -\alpha.
\end{align}
The circuit makes three sequential calls to the $R_z$ synthesizer and
applies two Hadamard gates. Since the same workspace is reused for all
three calls, this gives the depth, gate-count, and space bounds in
\Cref{eq:single-qubit-reduction-resources}.
\end{proof}
\section{A Lower-Bound Establishing \texorpdfstring{\UQNC}{Q-U-NC} Depth Optimality}
\label{sec:depth_lower_bound}

We will now prove that the $O(\log\log(1/\eps))$-depth achieved by our construction is asymptotically optimal
for circuits over any finite set of constant-width gates, in the worst case over the target rotation. It is not a lower bound for every individual angle. Arbitrarily many clean input ancillas are allowed, and they need not be returned to their initial state.

In terms of notation, we will consider a circuit $C$ with one input qubit, a designated output qubit $\textsf{out}$, and
$a$ ancillae initialized to $\ket{0^a}$. We will define the induced output channel, where all registers other than $\textsf{out}$ are discarded, as
\begin{align}
    \mathcal{E}_C(\rho)
    :=
    \operatorname{Tr}_{\overline{\textsf{out}}}
    \left[
        C
        \left(
            \rho\otimes\ketbra{0^a}{0^a}
        \right)
        C^\dagger
    \right].
    \label{eqn:induced_output_channel}
\end{align}
Furthermore, we will denote the $\Rz$-rotation channel as 
\begin{align}
    \mathcal{R}_\theta(\rho)
    :=
    R_z(\theta)\rho R_z(\theta)^\dagger.
\end{align}

\begin{theorem}[Depth Lower-Bound]
\label{thm:depth_lower_bound}
Let $\mathcal{G}$ denote a finite set of constant-width gates. In particular, let the total number of gates be bounded as $|\mathcal{G}|\leq g$, for some constant $g\geq 1$. Furthermore, assume each gate in $\mathcal{G}$ acts on at most $r\geq 2$ qubits. 
Let $\mathcal{C}_d(\mathcal{G})$ denote the
set of depth-at-most-$d$ circuits over $\mathcal{G}$ with clean ancillae
and a designated output qubit. For error $\eps\in(0,1/16]$, define the minimum depth circuit over alphabet $\mathcal{G}$ required to approximate every $\Rz$ gate to diamond-norm error at most $\eps$, as
\begin{align}
    D_\eps
    :=
    \min\left\{
        d :
        \forall\theta\in[0,2\pi),\
        \exists C_\theta\in
        \mathcal{C}_d(\mathcal{G})
        \text{ s.t. }
        \left\|
            \mathcal{E}_{C_\theta}-\mathcal{R}_\theta
        \right\|_\diamond
        \leq\eps
    \right\}.
\end{align}
If no such depth exists, set $D_\eps=+\infty$. Then,
\begin{align}
    D_\eps
    =
    \Omega_{g,r}
    \left(
        \log\log\frac{1}{\eps}
    \right).
\end{align}
\end{theorem}

\begin{proof}
The proof consists of three main steps: 1) a light-cone argument to upper-bound the number of distinct channels that can be induced on the output qubit by depth-$d$ circuits over $\mathcal{G}$, 2) a counting argument to lower-bound the number of distinct channels necessary to approximate all $\Rz$ gates to diamond norm $\eps$, and 3) combining the upper- and lower-bounds to achieve an overall depth lower-bound relative to the desired accuracy $\eps$.

\paragraph{Light-Cone Upper-Bound.} Consider the backwards light-cone of the designated output qubit \textsf{out}.
Since every gate in $\mathcal{G}$ is of width at most $r$, at distance $\ell$ from the output, the light-cone contains at most $r^\ell$ wires. Therefore, the light-cone of any depth-$d$ circuit contains at most $r^d$
input wires and at most
\begin{align}
    1+r+\cdots+r^{d-1}
    \leq
    \frac{r^d}{r-1} = O_r(r^d)
\end{align}
 total gates. Importantly, gates outside the light-cone may be removed without changing the induced output channel.

After canonically relabeling the light-cone wires, also record which of its at most $r^d$ input wires is the distinguished data input, or record that this input is absent. This contributes a factor at most $r^d+1$, absorbed below. Each gate is specified by its layer, type, and support. There are at most $d$ choices for
the layer, $g$ choices for the gate type, and
\begin{align}
    \sum_{k=1}^r (r^d)^k
    \leq
    r^{dr+1}
\end{align}
choices for its ordered support. Therefore, the total number of ways a single gate occurrence can be specified within the light-cone is upper-bounded by 
\begin{align}
    dg\,r^{dr+1}
    =
    \exp\bigl(O_{g,r}(d+1)\bigr).
\end{align}
Since the light-cone contains only $O_r(r^d)$ gates, the total number of unique depth-$d$ circuit light-cones can be upper-bounded as
\begin{align} \label{eqn:number_depth_d_circuits}
    M_d
    \leq
    \exp\left(
        O_{g,r}\bigl((d+1)r^d\bigr)
    \right).
\end{align}

\paragraph{Packing and conclusion.}
We now lower-bound the number of distinct output channels required to approximate every $R_z$ rotation to error $\eps$. The family $\{\mathcal R_\theta:\theta\in[0,2\pi)\}$ forms a one-parameter continuum of channels, and at precision $\eps$ we can choose $\Theta(1/\eps)$ rotations that are pairwise more than $2\eps$ apart in diamond norm. A single channel cannot be $\eps$-close to two such rotations, so any collection of circuits that approximates every $R_z(\theta)$ must induce at least $\Theta(1/\eps)$ distinct output channels.

To witness the separation between two rotations, consider the input state $\rho_+:=\ketbra{+}{+}$. This choice is natural because $R_z(\theta)$ changes only the relative phase between $\ket0$ and $\ket1$. Computational-basis inputs are therefore insensitive to $\theta$, whereas
$R_z(\theta)\ket+=(e^{-i\theta/2}\ket0+e^{i\theta/2}\ket1)/\sqrt2$
depends explicitly on the rotation angle. Hence the trace distance between the two outputs on $\ket+$ lower-bounds the diamond distance between the corresponding channels:
\begin{align}
    \|\mathcal R_\theta-\mathcal R_\phi\|_\diamond
    &\geq
    \|\mathcal R_\theta(\rho_+)-\mathcal R_\phi(\rho_+)\|_1
    =
    |e^{i\theta}-e^{i\phi}|
    =
    2\left|\sin\frac{\theta-\phi}{2}\right|.
    \label{eqn:rotation_channel_separation}
\end{align}

Now choose
$N:=\lfloor1/(4\eps)\rfloor$ and
$\theta_j:=2\pi j/N$ for $0\leq j<N$. Since $\eps\leq1/16$, we have $N\geq4$. For distinct $j,k$, the smallest angular separation is achieved by neighboring points. Using $\sin x\geq2x/\pi$ on $[0,\pi/2]$:
\begin{align}
    \|\mathcal R_{\theta_j}-\mathcal R_{\theta_k}\|_\diamond
    &\geq
    2\sin\frac{\pi}{N}
    \geq
    \frac{4}{N}
    \geq
    16\eps
    >
    2\eps.
    \label{eqn:pairwise_rotation_separation}
\end{align}
Thus no channel can be within diamond distance $\eps$ of two rotations in this collection. Therefore, any depth-$d$ circuit family that $\eps$-approximates every $R_z$ rotation must induce at least
$N=\Theta(1/\eps)$ distinct output channels.

If $D_\eps=+\infty$, the theorem is immediate. Otherwise, set $d:=D_\eps$. Combining the packing lower bound with the light-cone counting bound from \Cref{eqn:number_depth_d_circuits} yields
\begin{align}
    \left\lfloor\frac{1}{4\eps}\right\rfloor
    =N
    \leq
    M_d
    \leq
    \exp\!\left(O_{g,r}\bigl((d+1)r^d\bigr)\right).
\end{align}
The left-hand side grows as $\Theta(1/\eps)$, while the number of channels realizable at depth $d$ is at most doubly exponential in $d$. Taking logarithms twice therefore gives
\begin{align}
    \log\log\frac{1}{\eps}
    &\leq
    O_{g,r}(d+1),
\end{align}
and hence
$d=\Omega_{g,r}(\log\log(1/\eps))$.
\end{proof}

\section{Preparing the Kim--Laakkonen Catalyst}
\label{sec:kim-comparison}

Kim showed that, given a suitable reusable catalytic resource, a large family of phase rotations can be implemented by a constant-depth $\QNC_f^0$ circuit~\cite{kim2026catalytic}. Kim and Laakkonen later gave a related single-catalyst construction with the same basic feature~\cite{kim2025clifford}. Thus, once the catalyst is available, the online gate-synthesis procedure is already constant depth. The controlled-linear-map implementation of \Cref{lem:controlled-linear-map} realizes the online phase-kickback step in constant total depth over $\UQNC_f$. The remaining resource to construct is therefore the catalyst itself. To turn these catalytic procedures into ordinary circuit constructions starting from standard $\ket{0}$-initialized ancillae, we must understand the circuit complexity of preparing the required catalytic state.

Kim's original construction uses an angle-independent bank of $b$ eigenstates, occupying $b^2$ qubits in total. Kim and Laakkonen later showed that the same family of phase gates can instead be implemented using a single eigenstate $\ket{\psi_k}$, provided $\gcd(k,N)=1$. We work with this more economical single-catalyst formulation, reducing the preparation problem to synthesizing one $b$-qubit eigenstate rather than an entire bank. In particular, we fix $k=1$ and give a direct shallow preparation of $\ket{\psi_1}$, deriving its depth, size, and workspace complexity.

In \Cref{sec:catalytic-circuit-simulation}, we combine this preparation with the constant-depth catalytic phase-kickback procedure to obtain an ordinary circuit simulation with no pre-supplied catalyst, and then derive the resulting relationships between the standard and universal shallow-depth hierarchies. We first recall the eigenstate construction and how a single catalyst implements the desired family of grid phases.

\subsection{The Kim--Laakkonen Eigenstate Catalyst}
Fix $b\geq3$, let $N:=2^b-1$ and $\omega_N:=e^{2\pi i/N}$, and let
$C_f\in\mathrm{GL}_b(\mathbb F_2)$ be the companion matrix of a
primitive degree-$b$ polynomial. For any fixed $v\neq0$, the map
\begin{align}
U_f\ket z
&:=
\ket{C_fz}
\end{align}
cycles through all $N$ nonzero computational-basis states. Its
eigenstates are
\begin{align}
\ket{\psi_k}
&:=
\frac{1}{\sqrt N}
\sum_{j=0}^{N-1}
\omega_N^{-jk}
\ket{C_f^jv},
\label{eq:kim-eigenstates}
\end{align}
for $k\in\mathbb Z_N$. Shifting the summation index gives
\begin{align}
U_f^q\ket{\psi_k}
&=
\omega_N^{qk}\ket{\psi_k}.
\label{eq:kim-eigenvalue}
\end{align}
Thus a controlled application of $U_f^q$ applies the phase
$\omega_N^{qk}$ to the control while returning the eigenstate
unchanged.

For any integer $a$ coprime to $N$, Kim originally proposed an
angle-independent bank of $b$ such catalyst states,
\begin{align}
\ket{\Gamma_{a,b}}
&:=
\bigotimes_{t=0}^{b-1}
\ket{\psi_{a2^t\bmod N}},
\label{eq:kim-bank}
\end{align}
occupying $b^2$ qubits in total~\cite{kim2026catalytic}. Later, Kim
and Laakkonen observed that a single eigenstate $\ket{\psi_k}$ with
$\gcd(k,N)=1$ suffices~\cite{kim2025clifford}. Specifically, to implement a grid
phase $\omega_N^d$, one applies $U_f^q$ for
$q=dk^{-1}\bmod N$. Their preparation procedure uses phase estimation
and measurement to obtain a random label $k$, repeating until it is
coprime to $N$.

We instead fix $k=1$ from the outset and prepare $\ket{\psi_1}$
directly. The key observation is that, before applying the orbit
relabeling $j\mapsto C_f^jv$, the phases appearing in
\Cref{eq:kim-eigenstates} form an almost-product state. Define
\begin{align}
\ket{\chi_b}
&:=
\frac{1}{\sqrt{2^b}}
\sum_{j=0}^{2^b-1}
\omega_N^{-j}\ket j.
\label{eq:chi-b}
\end{align}
Writing $j=\sum_{t=0}^{b-1}j_t2^t$ shows that
\begin{align}
\ket{\chi_b}
&=
\bigotimes_{t=0}^{b-1}
\frac{\ket0+\omega_N^{-2^t}\ket1}{\sqrt2}.
\label{eq:chi-b-factorization}
\end{align}
Hence $\ket{\chi_b}$ can be prepared by synthesizing $b$ independent
single-qubit phase states in parallel. Each factor is obtained from
$\ket+$ by an $R_z$ rotation, up to an irrelevant global phase, i.e.
\begin{align}
R_z(\alpha)\ket+
&=
e^{-i\alpha/2}
\frac{\ket0+e^{i\alpha}\ket1}{\sqrt2}.
\label{eq:rz-prepares-phase-state}
\end{align}
We use \Cref{thm:rz-synthesis-qnc} to approximate these $b$ rotations
simultaneously.

It remains only to convert the computational labels in
$\ket{\chi_b}$ into the orbit labels defining $\ket{\psi_1}$. Since
$N=2^b-1$, exactly the first $N$ labels correspond to the desired
orbit. Apply a fixed basis permutation satisfying, for all $0\leq j<N$, $\ket j \longmapsto
\ket*{C_f^jv}$
and map the single remaining label $\ket N$ to $\ket{0^b}$. Ignoring
the approximation in the single-qubit rotations, the resulting state is
\begin{align}
\sqrt{\frac{N}{2^b}}\,\ket{\psi_1}
+
\frac{e^{i\varphi}}{\sqrt{2^b}}\ket{0^b}
\label{eq:psi-one-padded}
\end{align}
for some phase $\varphi$. Because the orbit of $v$ consists of all
nonzero strings, $\ket{0^b}$ is orthogonal to the support of
$\ket{\psi_1}$. The single padded basis vector therefore contributes
trace-distance error $2^{-b/2}$.

If each of the $b$ single-qubit factors in
\Cref{eq:chi-b-factorization} is prepared to error at most $\eta$, a
hybrid argument contributes an additional error at most $b\eta$.
Choosing $b=\Theta(\log(1/\eps))$ and
$\eta=\Theta(\eps/b)$ therefore gives total preparation error at most
$\eps$ and depth $O(\log\log(1/\eps))$. The only nonlocal step is the final basis permutation. Its
Fan-Out-assisted implementation uses $O(b2^b)$ gates and workspace,
including Fan-Out operations with as many as $2^b$ targets. In the next
subsection, we describe the circuit implementation and resource bounds in
detail.

\subsection{Basis Permutations}
The catalyst preparation requires a fixed relabeling of $b$-bit basis
states. We implement it by computing the input's one-hot encoding,
performing a hardwired lookup, and uncomputing the work. The
$b$-literal conjunctions take depth $O(\log b)$ with bounded-arity
Toffoli trees, or constant-depth with generalized Toffoli gates.

\begin{lemma}[Logarithmic-Width Basis Permutations]
\label{lem:log-width-permutation}
Every permutation $\pi:\{0,1\}^b\to\{0,1\}^b$ admits an exact clean $\UQNC_f$ implementation of depth $O(\log b)$ and gate count and workspace $O(b2^b)$.
\end{lemma}

\begin{proof}
We first construct the clean XOR oracle
\begin{align}
    O_\pi\ket{x}\ket{y}
    &:=
    \ket{x}\ket{y\oplus\pi(x)}.
\end{align}
The construction proceeds by reversibly converting the binary string
$x\in\{0,1\}^b$ into its one-hot unary encoding.  For
$a\in\{0,1\}^b$, let $e_a$ denote the unary basis string whose $a$th
coordinate is one and whose remaining coordinates are zero.  We implement $\ket{x}\ket*{0^{2^b}} \longmapsto \ket{x}\ket{e_x}$.

To compute the $a$-th unary bit, test whether $x=a$.  For each
$i\in[b]$, define the corresponding literal
\begin{align}
    \ell_{a,i}(x)
    &:=
    \begin{cases}
        x_i, & a_i=1,\\
        \neg x_i, & a_i=0.
    \end{cases}
\end{align}
The $a$th unary bit is then
\begin{align}
    (e_x)_a
    &=
    \bigwedge_{i=1}^b \ell_{a,i}(x),
\end{align}
which equals one exactly when $x=a$.  Using Fan-Out, distribute each input
bit to the work registers associated with all $2^b$ possible strings
$a$.  For each $a$, compute the conjunction of its $b$ literals using a
balanced tree of Toffoli gates.  Each tree has depth $O(\log b)$ and uses
$O(b)$ gates and work qubits.  Since the $2^b$ trees act on disjoint
registers, they are evaluated in parallel.  The binary-to-unary conversion
therefore has depth $O(\log b)$ and gate count and workspace
$O(b2^b)$.

The unary representation makes the evaluation of $\pi$ immediate.  For
each output coordinate $j\in[b]$, the truth table of $\pi$ gives
\begin{align}
    \pi(x)_j
    &=
    \bigoplus_{\substack{a\in\{0,1\}^b\\ \pi(a)_j=1}}
    (e_x)_a.
    \label{eq:unary-permutation-lookup}
\end{align}
Indeed, the unary register contains a 1 only in the coordinate
$a=x$.  After distributing its bits to disjoint work registers, all $b$
parities in \Cref{eq:unary-permutation-lookup} can be XORed into the target
register in parallel.  Each Parity is implemented in constant depth from
Fan-Out by conjugating with Hadamard gates.  Reversing the distribution and
binary-to-unary conversion returns all ancillary registers to zero.
Consequently, the clean oracle $O_\pi$ has depth $O(\log b)$ and gate
count and workspace $O(b2^b)$.

Finally, introduce a fresh $b$-qubit zero register, use $O_\pi$ to compute
$\pi(x)$ into it, and swap the two $b$-qubit registers.  The register swap
consists of $b$ parallel SWAP gates, each decomposable into three CNOT
gates, and therefore has constant depth.  Applying $O_{\pi^{-1}}$ then
erases the register containing $x$, since $\pi^{-1}(\pi(x))=x$.  This
implements
\begin{align}
    \ket{x}\ket{0^b}
    &\longmapsto
    \ket{\pi(x)}\ket{0^b}
\end{align}
exactly.  The two oracle calls are sequential but each has depth
$O(\log b)$, so the complete clean in-place implementation retains depth
$O(\log b)$ and gate count and workspace $O(b2^b)$.
\end{proof}

The remaining ingredient in the catalyst preparation is the basis relabeling that maps the computational labels of $\ket{\chi_b}$ to the orbit labels defining $\ket{\psi_1}$. Since this relabeling is a fixed permutation of the $2^b$ computational-basis states, it suffices to show that an arbitrary $b$-bit basis permutation can be implemented in constant depth. We will not require this implementation to be polynomial in $b$, since in our application $b=\Theta(\log(1/\eps))$. So, a cost of $O(b2^b)$ remains polynomial in $1/\eps$. The following corollary gives exactly the required implementation.

\begin{corollary}[Constant-depth basis permutations]
\label{cor:constant-depth-permutation}
Every permutation $\pi:\{0,1\}^b\to\{0,1\}^b$ admits an exact clean
$\UQAC_f$ implementation of depth $O(1)$ and gate count and workspace
$O(b2^b)$.
\end{corollary}

\begin{proof}
In the construction of \Cref{lem:log-width-permutation}, replace each
balanced bounded-arity Toffoli tree computing
\begin{align}
    (e_x)_a &= \bigwedge_{i=1}^b\ell_{a,i}(x)
\end{align}
by a single unbounded Toffoli gate. All $2^b$ conjunctions remain
parallel, while the input distribution, parity lookups, uncomputation,
and register swaps already have constant depth with Fan-Out. The gate
count and workspace are unchanged asymptotically.
\end{proof}

\subsection{Catalyst Preparation}
We now combine the phase-state preparation and basis permutation to
prepare the catalyst
\begin{align}
    \ket{\psi_1}
    &=\frac{1}{\sqrt N}\sum_{j=0}^{N-1}
        \omega_N^{-j}\ket*{C_f^jv}.
    \label{eqn:catalyst_psi}
\end{align}
The following bound includes all preparation workspace and allows
several independent copies to be prepared in parallel.

\begin{lemma}[Shallow preparation of the universal catalyst]
\label{lem:kim-catalyst-factory}
For integers $b\geq3$, $w\geq1$, and $\eta\in(0,1)$, there is a
fully unitary $\UQNC_f$ circuit, initialized entirely to zero, whose
joint catalyst--workspace output $\widetilde\rho_{b,w}$ satisfies

\begin{align}
    \frac{1}{2} \left\|
        \widetilde\rho_{b,w}-
        \bigl(\ketbra{\psi_1}{\psi_1}\bigr)^{\otimes w}
        \otimes
        \ketbra{0^a}{0^a}
    \right\|_1
    &\leq
    \sqrt{\frac{w}{2^b}}
    +
    wb\eta,
    \label{eq:kim-catalyst-factory-error}
\end{align}
Here $a$ counts the workspace qubits other than the $wb$ catalyst
qubits. The depth is $O(\log b+\log\log(1/\eta))$ and the size is  $w2^{O(b)}+O(wb\log(1/\eta))$.
\end{lemma}

\begin{proof}
The construction has two stages. First, we prepare a product-state
approximation to a slightly padded version of the catalyst. Second, we
apply an exact basis permutation that converts almost all of this state
into $\ket{\psi_1}$. The only two errors are therefore the synthesis
error in preparing the product state and the amplitude of the single
extra padded basis vector.

Recall that, for  $N=2^b-1$,
\begin{align}
    \ket{\chi_b}
    &=
    \bigotimes_{t=0}^{b-1}
    \frac{\ket0+\omega_N^{-2^t}\ket1}{\sqrt2}.
\end{align}
By \Cref{eq:rz-prepares-phase-state}, the $t$th factor is prepared,
up to global phase, by applying
$R_z(-2\pi 2^t/N)$ to $\ket+$. Thus each copy of
$\ket{\chi_b}$ consists of $b$ independent single-qubit rotations,
all of which can be synthesized in parallel.
Prepare $w$ copies simultaneously, approximating each of the $wb$
rotations to clean-input error $\eta$ using
\Cref{thm:rz-synthesis-qnc}. A hybrid argument over the $wb$
replacements gives a joint state $\rho_\chi$ satisfying
\begin{align}
    \frac12\left\|
        \rho_\chi
        -
        (\ketbra{\chi_b}{\chi_b})^{\otimes w}
        \otimes\ketbra{0^a}{0^a}
    \right\|_1
    &\leq wb\eta.
    \label{eq:padded-state-synthesis-error}
\end{align}
Here $a$ includes all synthesis workspace as well as the clean
registers used by the subsequent basis permutations.
It remains to convert the computational labels of $\ket{\chi_b}$ into
the orbit labels defining the catalyst. Because $C_f$ is the companion
matrix of a primitive polynomial, the orbit
$v,C_fv,\ldots,C_f^{N-1}v$ contains every nonzero $b$-bit string
exactly once. Hence the map
\begin{align}
    \pi_f\ket j
    &:=
    \begin{cases}
        \ket*{C_f^jv},&0\leq j<N,\\
        \ket{0^b},&j=N
    \end{cases}
    \label{eq:kim-preparation-permutation}
\end{align}
is a permutation of the full computational basis.

The reason for padding by the single label $j=N$ is now clear.
Namely, the first $N$ terms are mapped exactly to the desired catalyst, while
the one additional term is sent to the only basis state outside its
support. Using $\omega_N^{-N}=1$,
\begin{align}
    \ket{\widetilde\psi_1}
    &:=
    \pi_f\ket{\chi_b}
    =
    \sqrt{1-2^{-b}}\ket{\psi_1}
    +
    2^{-b/2}\ket{0^b}.
    \label{eq:padded-to-kim-catalyst}
\end{align}
Since $\ket{\psi_1}$ is supported only on nonzero computational-basis
states, the two terms are orthogonal. Thus one ideal padded copy has
squared overlap $1-2^{-b}$ with $\ket{\psi_1}$, and $w$ independent
copies have squared overlap $(1-2^{-b})^w$. Their trace distance is
therefore
\begin{align}
    \frac12\left\|
        (\ketbra{\widetilde\psi_1}{\widetilde\psi_1})^{\otimes w}
        -
        (\ketbra{\psi_1}{\psi_1})^{\otimes w}
    \right\|_1
    &=
    \sqrt{1-(1-2^{-b})^w}
    \leq
    \sqrt{\frac{w}{2^b}},
    \label{eq:kim-padding-error}
\end{align}
where we used $1-(1-x)^w\leq wx$. Apply the exact clean implementation of $\pi_f$ from
\Cref{lem:log-width-permutation} to all $w$ copies in parallel.
Unitary invariance preserves the synthesis error in
\Cref{eq:padded-state-synthesis-error}, while
\Cref{eq:kim-padding-error} bounds the additional error caused by the
single padded label. The triangle inequality therefore gives
\begin{align}
    \frac12\left\|
        \widetilde\rho_{b,w}
        -
        (\ketbra{\psi_1}{\psi_1})^{\otimes w}
        \otimes\ketbra{0^a}{0^a}
    \right\|_1
    &\leq
    \sqrt{\frac{w}{2^b}}
    +
    wb\eta,
\end{align}
which is \Cref{eq:kim-catalyst-factory-error}.

Finally, all $wb$ single-qubit rotations are synthesized in parallel,
giving depth $O(\log\log(1/\eta))$ and total gate count and workspace
$O(wb\log(1/\eta))$. The $w$ basis permutations are also applied in
parallel and, by \Cref{lem:log-width-permutation}, contribute depth
$O(\log b)$ and total gate count and workspace
$O(wb2^b)=w2^{O(b)}$. Since these two stages are sequential, the
claimed resource bounds follow.
\end{proof}

The preceding construction prepares a single copy of $\ket{\psi_1}$ with error controlled by the phase-state approximation and the single padded basis vector. For circuit simulation, however, we will need a small bank of independent catalysts so that several phase gates can be implemented in parallel. Since the copies can be prepared simultaneously, increasing the number of catalysts affects the size and workspace but not the circuit depth. It therefore suffices to prepare each copy to error $O(\eps/w)$ and apply a hybrid bound over the $w$ copies. Choosing the precision parameter $b$ accordingly gives the following consequence.

\begin{corollary}[Constant-depth preparation of the universal catalyst]
\label{cor:constant-depth-kim-catalyst}
For every integer $w\geq1$ and $\eps\in(0,1)$, there is a choice
$b=O(\log(w/\eps))$ for which $w$ copies of the universal catalyst
$\ket{\psi_1}$ can be prepared by a $\UQAC_f^0$ circuit to trace-distance error at most $\eps$. The gate count and workspace are $\poly(w,1/\eps)$.
\end{corollary}

\begin{proof}
Use the construction of \Cref{lem:kim-catalyst-factory} with
$b=\max\{3,\lceil\log_2(4w/\eps^2)\rceil\}$ and
$\eta=\eps/(2wb)$. Its error analysis gives
\begin{align}
    \frac12\left\|
      \widetilde\rho_{b,w}
      -\bigl(\ketbra{\psi_1}{\psi_1}\bigr)^{\otimes w}
          \otimes\ketbra{0^a}{0^a}
    \right\|_1
    &\leq\sqrt{\frac{w}{2^b}}+wb\eta\leq\eps.
\end{align}
For the depth bound, fix $\gamma=1/4$ and replace each of the $wb$
parallel phase-factor synthesizers by the wide-gate construction of
\Cref{cor:rz-synthesis-qacf}. The same coherent accuracy $\eta$ is
available, so the error estimate is unchanged. Implement the $w$
permutations $\pi_f$ in parallel by
\Cref{cor:constant-depth-permutation}. Both stages have constant depth.
Their total gate count is
$O(wb2^b+wb\log(1/\eta))$, and the workspace is
\begin{align}
    O\bigl(wb2^b+wb\log^{5/4}(1/\eta)\bigr)
       =\poly(w,1/\eps).
\end{align}
Thus all resources have the claimed polynomial dependence and the
parameter $b$ is chosen with the target accuracy.
\end{proof}

\section{Approximate Polylogarithmic Fan-Out}
\label{sec:qac-uqac-collapse}

Several of our shallow synthesis constructions are naturally expressed as $\UQAC_f^0$ circuits, but use Fan-Out only on a fixed polylogarithmic number of targets. We show that these particular Fan-Out gates can be implemented in $\UQAC^0$, using only Hadamard and generalized Toffoli gates, with inverse-polynomial error, constant depth, and polynomial size. Note that the approximation is coherent and clean, i.e. approximately restores all zero-initialized workspace. Consequently, any $\UQAC_f^0$ synthesis construction whose Fan-Out gates have fixed polylogarithmic width can be compiled into a $\UQAC^0$ circuit with the same asymptotic depth and polynomial size. In \Cref{sec:catalytic-circuit-simulation}, we use this compilation to derive the resulting relationships between the standard and universal shallow-depth hierarchies. Importantly, the result applies only to fixed polylogarithmic-width Fan-Out and does not imply a general collapse $\UQAC_f^0=\UQAC^0$.

Approximate shallow Fan-Out has been studied previously. Rosenthal gave constant-depth approximations of exponential size in $\QAC$~\cite{rosenthal2021bounds}, while Grier, Morris, and Wu gave exact amplitude-amplification constructions, including polynomial-size implementations for every fixed polylogarithmic width~\cite[Cor.~8, Cor.~10]{grier2026mathsf}. These constructions use angle-dependent single-qubit rotations, and therefore do not directly provide the universal implementation needed here. Instead, we use the approximate-nekomata construction of Grier and Morris~\cite{grier2025quantum}, choosing its phase operation to mark only the all-ones basis state. With this choice, the construction uses only Hadamard and generalized Toffoli gates. A nekomata has two equally weighted branches, $\ket{0^q}$ and $\ket{1^q}$, on its target register. This structure yields coherent Parity by encoding the input parity as their relative phase, uncomputing the nekomata, and testing whether its preparation register returns to zero. Hadamard conjugation then gives Fan-Out.

We first obtain logarithmic-width Fan-Out from this construction, and then compose these circuits in a constant-height tree to reach $O(\log^c n)$ targets for any fixed $c$, while preserving constant depth, polynomial size, and inverse-polynomial accuracy.

\subsection{Prior Results}
\label{sec:fanout-prior-results}

We begin by specifying the operations and approximation guarantee.
On registers $B,X$, Parity XORs the parity of $X$ into $B$:
\begin{align}
    \mathsf{Parity}_L\ket b\ket{x_1,\ldots,x_L}
    &:=\ket{b\oplus x_1\oplus\cdots\oplus x_L}
       \ket{x_1,\ldots,x_L}.
    \label{eqn:parity-unitary-definition}
\end{align}
Fan-Out instead XORs the bit in $B$ into each qubit of $X$:
\begin{align}
    \FANOUT_L\ket b\ket{x_1,\ldots,x_L}
    &:=\ket b\ket{x_1\oplus b,\ldots,x_L\oplus b}.
    \label{eqn:fanout-unitary-definition}
\end{align}
Conjugating a CNOT by Hadamards reverses its control and target, giving
the standard Parity--Fan-Out duality~\cite[Prop.~1]{moore1999qac0}, i.e.
\begin{align}
\FANOUT_L
&=
H^{\otimes(L+1)}
\mathsf{Parity}_L
H^{\otimes(L+1)}.
\label{eqn:parity-fanout-hadamard-duality}
\end{align}
Thus, an implementation of Parity gives one of Fan-Out with only two
additional Hadamard layers.

We say that a unitary $V$ $\eps$-implements $U$ with clean
ancillas if, for every normalized input state $\ket{\psi}$,
\begin{align}
\left\|
V\bigl(\ket{0^A}\ket{\psi}\bigr)
-
\ket{0^A}U\ket{\psi}
\right\|
&\leq
\eps.
\label{eqn:clean-isometry-error-definition}
\end{align}
Equivalently, the implemented operation is close to $U$ while returning
the ancillary register approximately to $\ket*{0^A}$. Since the bound is
uniform over $\ket{\psi}$, it remains valid even when the data register
is entangled with an arbitrary reference system.

The connection to state preparation uses \emph{nekomata}. A
$q$-nekomata is a pure state of the form
\begin{align}
\ket\nu
&=
\frac{
\ket{0^q}\ket{\alpha_0}
+
\ket{1^q}\ket{\alpha_1}
}{\sqrt2},
\label{eqn:nekomata-definition}
\end{align}
where $\ket{\alpha_0}$ and $\ket{\alpha_1}$ are arbitrary normalized
garbage states. Following \cite[Def.~15]{grier2025quantum}, a pure state
$\ket\phi$ is a $\delta$-approximate $q$-nekomata if there exists such a
$\ket\nu$ satisfying
\begin{align}
|\braket{\nu}{\phi}|^2
&\geq
1-\delta.
\label{eqn:approximate-nekomata-definition}
\end{align}
Thus $\delta$ measures infidelity with an exact nekomata. The following
criterion shows that it is enough to place nearly half of the probability
on each of the two endpoint states $\ket{0^q}$ and $\ket{1^q}$.

\begin{fact}[Endpoint Criterion {\cite[Lem.~18]{grier2025quantum}}]
\label{fact:nekomata-endpoints}
Let $0\leq\eps<1/2$. If measuring the designated $q$ target qubits of a
pure state yields $0^q$ and $1^q$ with probability at least
$1/2-\eps$ each, then the state is a $2\eps$-approximate
$q$-nekomata.
\end{fact}
\noindent Indeed, one may normalize the garbage state associated with each endpoint
and give the two endpoint branches equal amplitude. If the original state
already places probability close to $1/2$ on each branch, then it has
large overlap with this exact nekomata.

Grier and Morris construct such states using a rectangular grid of
$q$-qubit columns. Each non-target column is prepared so that it is very
likely to be $\ket{1^q}$, less likely to be $\ket{0^q}$, and unlikely to
lie outside these two endpoint states. A row-wise generalized Toffoli
layer then combines the non-target columns into a distinguished target
column. By choosing the number of columns appropriately, the target
column is nearly equally likely to be $\ket{0^q}$ or $\ket{1^q}$, and
the complete grid state is therefore close to a nekomata by the
preceding fact. For our application, we use the simplest case of their construction, in
which the phase operation marks only the all-ones basis state. In the
notation of Grier and Morris, this corresponds to $S_q:=\{1^q\}$,
where $S_q$ denotes the set of basis strings that acquire a minus sign.
The associated phase reflection is therefore
\begin{align}
U_{S_q}
&:=
I-2\ketbra{1^q}{1^q}.
\label{eqn:singleton-phase-gate}
\end{align}
This reflection has an exact constant-depth implementation over our
restricted gate set. Taking the last qubit as the target,
\begin{align}
U_{S_q}
&=
\left(I^{\otimes(q-1)}\otimes H\right)
\GTof_{q-1}
\left(I^{\otimes(q-1)}\otimes H\right),
\label{eqn:singleton-phase-gate-implementation}
\end{align}
where $\GTof_{q-1}$ has the first $q-1$ qubits as controls and the
last as its target. Indeed, conjugating the target bit flip by Hadamards turns it
into a $Z$ phase conditioned on all $q-1$ controls being one, so the
unique basis state that acquires a minus sign is $\ket{1^q}$.

\begin{fact}[Singleton Grid Construction
{\cite[Thm.~19]{grier2025quantum}}]
\label{fact:singleton-grid-construction}
Fix $q\geq4$ and an integer $t\geq1$, and use the reflection
$U_{S_q}$ from \Cref{eqn:singleton-phase-gate}.
Let $B_q:=H^{\otimes q}U_{S_q}H^{\otimes q}$,
$\ket{s_q}:=B_q\ket{0^q}$, and
$\gamma_j:=|\braket{j^q}{s_q}|^2$ for $j\in\{0,1\}$.
Arrange $q(t+1)$ qubits in $q$ rows and $t+1$ columns. Prepare each
of $t$ non-target columns in
$(I-2\ketbra{s_q}{s_q})\ket{1^q}$, leaving the target column in
$\ket{0^q}$. In each row, apply a generalized Toffoli controlled on
the non-target bits and targeting the remaining bit. The target column equals $1^q$ with probability
\begin{align}
    p(1^q)
    &=(1-2\gamma_1)^{2t}.
    \label{eqn:grid-all-one-probability}
\end{align}
Writing $p_{\mathrm{other}}$ for the probability that the target is
neither $0^q$ nor $1^q$, we have
\begin{align}
    p_{\mathrm{other}}
    &\leq4t\gamma_1(1-\gamma_0-\gamma_1).
    \label{eqn:grid-other-probability}
\end{align}
The circuit has constant depth, uses $q(t+1)$ qubits, and contains
$O(q(t+1))$ Hadamard and generalized Toffoli gates.
\end{fact}

The bounds in the preceding fact follow from independence across the
non-target columns together with a union bound on the probability of
leaving the two endpoint states. We now combine this approximate
nekomata preparation with the standard reduction from nekomata to
coherent Parity.

\begin{fact}[Coherent Nekomata-to-Parity Reduction
{\cite[Thm.~3.1 and Sec.~3.2]{rosenthal2021bounds}}]
\label{fact:nekomata-to-parity}
Suppose an $A$-qubit unitary circuit $C$ prepares a
$\delta$-approximate $L$-nekomata from $\ket{0^A}$, where
$0\leq\delta<1$. Then the standard reduction produces a circuit
$\mathcal P_L(C)$ such that, for every state $\ket{\psi}$ on the
Parity input and output registers,
\begin{align}
\left\|
\mathcal P_L(C)
\bigl(\ket{0^A}\ket{\psi}\bigr)
-
\ket{0^A}\mathsf{Parity}_L\ket{\psi}
\right\|
&\leq
4\sqrt{2\delta}.
\label{eqn:nekomata-parity-hybrid-error}
\end{align}
The construction uses two calls to $C$, two to $C^\dagger$, and
$O(A+L)$ additional Hadamard and generalized Toffoli gates in constant
additional depth. In particular, it preserves the Hadamard--Toffoli
gate set and approximately returns the preparation register to
$\ket{0^A}$.
\end{fact}

The reduction is illustrated explicitly in
\cite[App.~A, Fig.~2]{grier2025quantum}. Conceptually, it writes the
input parity as a relative phase between the two nekomata branches,
uncomputes the preparation, and records whether the preparation register
returns to zero. The required $CZ$ and reversible OR operations use only
Hadamards and generalized Toffoli gates.
The $O(\sqrt{\delta})$ error arises when converting nekomata infidelity
into coherent state error. Indeed, if
$\ket{\phi}=C\ket{0^A}$ has squared overlap at least $1-\delta$ with an
exact nekomata $\ket{\nu}$, then after choosing their relative phase,
\begin{align}
\|\ket{\phi}-\ket{\nu}\|_2
&\leq
\sqrt{2\delta}.
\label{eqn:preparation-unitary-distance}
\end{align}
Equivalently, for the analysis one may replace $C$ by a unitary
$\widehat C$ that prepares $\ket{\nu}$ exactly and differs from $C$ by at
most $\sqrt{2\delta}$ in operator norm. A hybrid over the four calls to
$C$ or $C^\dagger$ then gives the factor
$4\sqrt{2\delta}$ in \Cref{eqn:nekomata-parity-hybrid-error}.

We will also use a standard copying-tree construction to extend
logarithmic-width Fan-Out to larger polylogarithmic widths. The idea is
simple. Starting from one control qubit, repeatedly apply Fan-Out to
fresh zero ancillas in a tree. If each Fan-Out gate has at most
$B-1$ targets, then after $h$ levels we can create up to $B^h$
coherent copies of the control bit. Thus $M$ copies require depth only
$\lceil\log_B M\rceil$.

\begin{fact}[Restricted Fan-Out Tree
{\cite[Lem.~3.9]{rosenthal2021bounds}}]
\label{fact:restricted-fanout-tree}
For integers $B\geq2$ and $M\geq1$, there is a circuit $V$ on $M$
qubits, using at most $M-1$ Fan-Out gates with at most $B-1$ targets
per gate, of depth $\lceil\log_B M\rceil$, such that, for
$b\in\{0,1\}$,
\begin{align}
V\bigl(\ket b\ket{0^{M-1}}\bigr)
&=
\ket b^{\otimes M}.
\end{align}
\end{fact}
\noindent For our application, we take $B=\Theta(\log n)$ and
$M=\operatorname{polylog}(n)$. Hence $\log_B M=O(1)$, so the copying
tree has constant depth. We then replace each width-$B$ Fan-Out gate in
the tree by our constant-depth universal implementation from
\Cref{lem:uqac-logarithmic-fanout}. Gates within each tree level act on
disjoint registers and therefore run in parallel. On a superposition, the tree produces the corresponding cat state
rather than independent copies of an unknown quantum state.

\subsection{Universal Polylogarithmic Fan-Out}
\label{sec:fanout-fixed-gate-construction}

We now combine the preceding ingredients to obtain universal Fan-Out.
We first specialize the Grier--Morris grid construction to produce
logarithmic-width Fan-Out using only Hadamard and generalized Toffoli
gates, with inverse-polynomial coherent error and polynomial resources.
We then compose these logarithmic-width gadgets in the copying-tree
construction to reach any fixed polylogarithmic number of targets while
preserving constant depth and polynomial size.

\begin{lemma}[Fixed-Gate Simulation of Logarithmic-Width Fan-Out]
\label{lem:uqac-logarithmic-fanout}
For every pair of integers $q\geq4$ and $1\leq L\leq q$, there is a
constant-depth circuit $G_{L,q}$ over $H$ and generalized Toffoli gates,
acting on the $L+1$ data qubits of $\FANOUT_L$ together with
$A_q=\Theta(q4^q)$ clean ancilla qubits, such that, for every input state
$\ket{\psi}$,
\begin{align}
\left\|
G_{L,q}
\bigl(\ket{0^{A_q}}\ket{\psi}\bigr)
-
\ket{0^{A_q}}\FANOUT_L\ket{\psi}
\right\|
&\leq
16\cdot 2^{-q/2}.
\label{eqn:fixed-gate-fanout-error}
\end{align}
The circuit uses $O(q4^q)$ gates and qubits. Its largest generalized Toffoli gate has
$O(q4^q)$ controls.
\end{lemma}

\begin{proof}
Apply \Cref{fact:singleton-grid-construction}, whose gate set is
already Hadamard and generalized Toffoli by
\Cref{eqn:singleton-phase-gate-implementation}. It remains to choose
the number of columns so that measuring the target gives $0^q$ and
$1^q$ with nearly equal probability, with little probability on any
other string. The endpoint criterion will then give an approximate
nekomata.

Let $u:=2^{-q}$
denote the probability weight of any individual basis string in the
uniform superposition on $q$ qubits. Since every amplitude of
$H^{\otimes q}\ket{1^q}$ has magnitude $2^{-q/2}=\sqrt{u}$, expanding
\begin{align}
\ket{s_q}
&=
\ket{0^q}
-
2^{1-q/2}H^{\otimes q}\ket{1^q}
\end{align}
shows that the probabilities of its two endpoint strings are
\begin{align}
\gamma_0
&:=
|\braket{0^q}{s_q}|^2
=
(1-2u)^2
\label{eqn:singleton-endpoint-probabilities}
\end{align}
and
\begin{align}
\gamma_1
&:=
|\braket{1^q}{s_q}|^2
=
4u^2.
\label{eqn:singleton-one-probability}
\end{align}
These two endpoint probabilities determine, through
\Cref{fact:singleton-grid-construction}, how the final target distribution
depends on the number of auxiliary columns. In particular, with $t$ non-target columns, the target is $\ket{1^q}$ with
probability
\begin{align}
p(1^q)
&=
(1-2\gamma_1)^{2t}.
\end{align}
We therefore choose
\begin{align}
t_q
&:=
\left\lceil
\frac{\ln2}{-2\ln(1-2\gamma_1)}
\right\rceil
=
\Theta(4^q),
\label{eqn:singleton-column-count}
\end{align}
so that this probability is approximately $1/2$. The asymptotic estimate
follows from $\gamma_1=\Theta(4^{-q})$ and
$-\ln(1-2\gamma_1)=\Theta(\gamma_1)$. Without the ceiling, this choice
would give $p(1^q)=1/2$ exactly. Rounding upward gives
\begin{align}
\frac12-2\gamma_1
&\leq
\frac{(1-2\gamma_1)^2}{2}
\leq
p(1^q)
\leq
\frac12.
\label{eqn:singleton-target-one}
\end{align}
Thus the all-one branch already has probability within exponentially
small error of $1/2$.

It remains to show that little probability lies outside the two endpoint
strings. Since $-\ln(1-2\gamma_1)\geq2\gamma_1$, our choice of $t_q$
implies
\begin{align*}
4t_q\gamma_1
&\leq
\ln2+4\gamma_1
<
1
\end{align*}
for $q\geq4$. Applying
\Cref{eqn:grid-other-probability} therefore gives
\begin{align}
p_{\mathrm{other}}
&\leq
1-\gamma_0-\gamma_1
=
4u-8u^2
\leq
4u.
\end{align}
Hence the total probability on target strings other than
$\ket{0^q}$ and $\ket{1^q}$ is at most $4u$. Let $p(0^q)$ denote the probability of the all-zero target string. Since
$p(1^q)\leq1/2$,
\begin{align}
p(0^q)
&=
1-p(1^q)-p_{\mathrm{other}}
\geq
\frac12-4u.
\label{eqn:singleton-target-zero}
\end{align}
Moreover, $2\gamma_1=8u^2\leq4u$, so
\Cref{eqn:singleton-target-one} also gives
\begin{align}
p(1^q)
&\geq
\frac12-4u.
\end{align}
Thus, both endpoint branches have probability at least
$1/2-4u$.

Let $C_q$ denote the resulting grid preparation. Since the grid contains
$q$ qubits in each of its $t_q+1$ columns, it acts on
\begin{align}
A_q
&:=
q(t_q+1)
=
\Theta(q4^q)
\end{align}
qubits. Let $\delta_q:=8u=2^{3-q}$ and
\(\ket{\phi_q}:=C_q\ket{0^{A_q}},\)
\Cref{fact:nekomata-endpoints} with $\eps=4u$ implies that there is an
exact $q$-nekomata $\ket{\nu_q}$ satisfying
\begin{align}
|\braket{\nu_q}{\phi_q}|^2
&\geq
1-\delta_q.
\label{eqn:singleton-nekomata-infidelity}
\end{align}
Therefore the grid prepares a $q$-nekomata with exponentially small
infidelity using $\Theta(q4^q)$ qubits. Finally, for any $L\leq q$, the
unused $q-L$ target qubits can be absorbed into the garbage register, so
the same prepared state is also a $\delta_q$-approximate
$L$-nekomata.

Apply \Cref{fact:nekomata-to-parity} to $C_q$, then conjugate the data
qubits by Hadamards as in \Cref{eqn:parity-fanout-hadamard-duality}.
These Hadamards preserve the error, so the resulting Fan-Out circuit
has clean-input error at most
\begin{align*}
    4\sqrt{2\delta_q}
    &=16\,2^{-q/2}.
\end{align*}
The grid uses $O(q4^q)$ gates in constant depth. The reduction makes
four preparation or inverse-preparation calls on the same register
and adds $O(A_q+L)$ gates in constant additional depth. Its reversible
OR has $A_q$ controls, and no other gate has larger arity. This proves
all resource bounds.
\end{proof}

The preceding lemma has cost exponential in $q$, so we now restrict to
the regime $q=O(\log n)$. In this range, the size and workspace
$O(q4^q)$ remain polynomial in $n$, while the error
$O(2^{-q/2})$ can be made inverse polynomial. This gives the following
logarithmic-width Fan-Out implementation.

\begin{corollary}[Inverse-Polynomial Logarithmic-Width Fan-Out]
\label{cor:inverse-poly-logarithmic-fanout}
Let $L(n)=O(\log n)$ and fix a constant $\beta>0$. There is a
polynomial-size, constant-depth circuit family over $H$ and generalized Toffoli gates that $n^{-\beta}$-implements
$\FANOUT_{L(n)}$ with clean ancillas.
\end{corollary}

\begin{proof}
The case $L(n)=0$ is immediate. Otherwise, set
$q(n):=\max\{L(n),4,\lceil2\log_2(16n^\beta)\rceil\}$
and apply \Cref{lem:uqac-logarithmic-fanout}.
This meets the lemma's hypotheses and makes its error at most
$n^{-\beta}$. Since $q(n)=O(\log n)$, the gate count and workspace
$q(n)4^{q(n)}$ are polynomial in $n$.
\end{proof}

\noindent For larger widths, setting $q\geq L(n)$ directly would lose
polynomial size. Instead, we compose logarithmic-width gadgets using
\Cref{fact:restricted-fanout-tree}. The tree needs only constantly
many levels to reach any fixed polylogarithmic number of targets.

\begin{corollary}[Inverse-Polynomial Polylogarithmic Fan-Out]
\label{cor:inverse-poly-polylog-fanout}
For every fixed $c,\beta>0$ and every $L(n)=O(\log^c n)$, there is a
polynomial-size, constant-depth circuit $G_n$ over Hadamard and
generalized Toffoli gates, using $A(n)=\operatorname{poly}(n)$
clean ancilla qubits, such that, for every input state $\ket{\psi}$,
\begin{align}
\left\|
G_n
\bigl(\ket*{0^{A(n)}}\ket{\psi}\bigr)
-
\ket*{0^{A(n)}}\FANOUT_{L(n)}\ket{\psi}
\right\|
&\leq
n^{-\beta}.
\label{eqn:polylog-fanout-error}
\end{align}
\end{corollary}

\begin{proof}
The case $L(n)=0$ is immediate, so assume $L(n)\geq1$. Let
$B:=\max\{2,\lceil\log_2 n\rceil\}$ be the branching factor, and
let $M:=L(n)+1$ count the original control together with one private
copy for each target. By \Cref{fact:restricted-fanout-tree}, a copying
circuit $V$ produces these copies using $1\leq K\leq L(n)$ Fan-Out
gates, each with at most $B-1$ targets. Its depth is
$\lceil\log_B M\rceil=O_c(1)$ because $M=O(\log^c n)$ and
$B=\Theta(\log n)$.

Let $D$ be the parallel CNOT layer from the private control copies
to their respective targets. The ideal implementation is
$V^\dagger D V$: copy the control, act on the targets, and undo the
copies. Writing $F:=\FANOUT_{L(n)}$, this works because
\begin{align}
    D V\bigl(\ket*{0^{L(n)}}\ket\psi\bigr)
    &=V\bigl(\ket*{0^{L(n)}}F\ket\psi\bigr)
    \label{eqn:polylog-copy-act-intertwining}
\end{align}
for every normalized state $\ket\psi$ of the original control and
targets. Both sides copy the same control value and XOR it into the
targets, only in a different order. Equality on basis states extends
linearly to arbitrary inputs.

Give each tree gate private implementation workspace and local error
$\eps_{\mathrm{loc}}:=n^{-\beta}/(2K)$. All replacements from
\Cref{lem:uqac-logarithmic-fanout} can use the parameter
\begin{align}
    q&:=\max\left\{B,4,
          \left\lceil2\log_2\frac{16}{\eps_{\mathrm{loc}}}
          \right\rceil\right\}=O(\log n).
\end{align}
Here the logarithmic bound follows from $K=O(\log^c n)$. Thus each
local gadget has polynomial size and workspace, rather than the
quasipolynomial cost that would result from setting $q\geq L(n)$
directly.

Let $\widetilde V$ be the resulting tree, and let $A$ count the
copy qubits and all implementation workspace. Extend $V$ and $D$
by the identity on the added workspace. A hybrid over the $K$
replacements gives
\begin{align}
    \left\|(\widetilde V-V)
       \bigl(\ket{0^A}\ket\psi\bigr)\right\|_2
    &\leq K\eps_{\mathrm{loc}}=\frac{n^{-\beta}}2=: \zeta
    \label{eqn:copying-tree-error}
\end{align}
for every normalized $\ket\psi$. Each comparison uses an ideal
prefix, so the replaced gate's private workspace is zero; the gate's
data may be entangled with the rest of the tree.

Set $G_n:=\widetilde V^\dagger D\widetilde V$. To bound its error,
we use \Cref{eqn:polylog-copy-act-intertwining} rather than assume
that the approximate forward tree has left its workspace clean.
Unitary invariance and the triangle inequality give
\begin{align}
    &\left\|G_n\bigl(\ket{0^A}\ket\psi\bigr)
                   -\ket{0^A}F\ket\psi\right\|_2\\
    &\quad=\left\|D\widetilde V\bigl(\ket{0^A}\ket\psi\bigr)
                   -\widetilde V\bigl(\ket{0^A}F\ket\psi\bigr)
             \right\|_2\\
    &\quad\leq
       \left\|D(\widetilde V-V)\bigl(\ket{0^A}\ket\psi\bigr)
             \right\|_2
       +\left\|(V-\widetilde V)\bigl(\ket{0^A}F\ket\psi\bigr)
             \right\|_2\\
    &\quad\leq2\zeta=n^{-\beta}.
\end{align}
The two terms use \Cref{eqn:copying-tree-error} on $\ket\psi$ and
$F\ket\psi$, respectively, both with zero-initialized workspace.
This also bounds the error in restoring all ancillary qubits.

The tree has constantly many levels, with disjoint gadgets within
each level. It contains only $K=O(\log^c n)$ polynomial-size
gadgets, so its total size and workspace are polynomial. Reversing
it and adding the CNOT layer preserves constant depth. Finitely
many small input lengths can be handled separately.
\end{proof}

\section{Relations Between the Standard and Universal Hierarchies}
\label{sec:shallow-circuit-implications}

We will now leverage the previously described synthesis results to obtain a number of circuit simulations and relations between the standard and universal shallow-depth quantum circuit hierarchies. Although the universal hierarchy restricts the allowed single-qubit gate set, we will show that this restriction can often be imposed with little or no loss in depth. In particular, our direct synthesis results yield universal simulations with at most a doubly logarithmic depth overhead, while the polylogarithmic Fan-Out and catalytic constructions allow this overhead to be removed or made merely additive in several important shallow-depth regimes.

First, the depth-optimal single-qubit synthesis results of
\Cref{sec:be_gate_synth} already give several direct relations between
the standard and universal hierarchies. With bounded-arity gates,
arbitrary single-qubit gates can be synthesized to inverse-polynomial
accuracy in depth $O(\log\log n)$. At the level of full unitary
synthesis, this gives
\begin{align*}
\textsf{Unitary}(\QNC^0)
&\subseteq
\textsf{Unitary}(\UQNC[O(\log\log n)]),
\\
\textsf{Unitary}(\QAC^0)
&\subseteq
\textsf{Unitary}(\UQAC[O(\log\log n)]),
\\
\textsf{Unitary}(\QNC_f^0)
&\subseteq
\textsf{Unitary}(\UQNC_f[O(\log\log n)]).
\end{align*}
For bounded-error $\QNC^0$ decision computation, the single measured
output has only a constant-size backward light cone, so constant
accuracy suffices and the conclusion strengthens to
$\QNC^0=\UQNC^0$. This stronger constant-depth statement is
decision-specific and does not approximate the full unitary. Finally,
when both generalized Toffoli and Fan-Out are available, the
constant-depth synthesis construction gives the full-unitary equality
$\textsf{Unitary}(\QAC_f^0)=
\textsf{Unitary}(\UQAC_f^0)$, and hence also
$\QAC_f^0=\UQAC_f^0$ for bounded-error decision problems.

Second, \Cref{sec:qac-uqac-collapse} strengthens the constant-depth
synthesis result for circuits with generalized Toffoli gates. In
particular, the construction of \Cref{sec:be_gate_synth} uses only
polylogarithmic-width Fan-Out in addition to the fixed universal basis
and generalized Toffoli gates, while \Cref{sec:qac-uqac-collapse}
shows that these Fan-Out gates can themselves be approximated in
constant depth and polynomial space using only Hadamards and generalized
Toffoli gates. Combining the two gives a depth-preserving approximation
of the full unitary implemented by any $\QAC$ circuit, on arbitrary
inputs with clean ancillas. Thus, for every fixed $k\geq0$,
$$
    \textsf{Unitary}(\QAC^k)
    =
    \textsf{Unitary}(\UQAC^k),
$$
and, as an immediate consequence for bounded-error decision problems,
$\QAC^k=\UQAC^k$.

Finally, the end-to-end catalyst synthesis of
\Cref{sec:kim-comparison} gives a complementary universalization result
for circuits with unbounded Fan-Out. Once the Kim--Laakkonen catalyst
bank is available, each phase gate is implemented by a constant-depth
phase-kickback gadget, and the ideal catalyst is returned unchanged.
The same bank can therefore be reused throughout the computation, so
the only nonconstant-depth cost is its one-time
$O(\log\log n)$ preparation. For a depth-$D$ circuit, this yields a
simulation of the complete induced quantum channel in depth
$O(D+\log\log n)$, including on inputs entangled with a reference
system. For every fixed $k\geq1$, however,
the additive overhead is absorbed into $O(\log^k n)$, giving
$\QNC_f^k=\UQNC_f^k$ for bounded-error computation. Combining this with
the Takahashi--Tani collapse of the unrestricted Fan-Out hierarchy yields, for every fixed $k\geq1$,
$$
    \QNC_f^k=\UQNC_f^k=\QAC_f^k=\UQAC_f^k
    =\QTC_f^k=\UQTC_f^k.
$$

\subsection{Immediate Consequences of Single-Qubit Synthesis}
\label{sec:direct-unitary-compilation}

We first apply the single-qubit synthesis results of
\Cref{sec:be_gate_synth} independently to each single-qubit gate, while
leaving the multiqubit primitives unchanged. The following lemma
collects the resulting resource bounds, which will be used for both the
unitary-synthesis and bounded-error decision consequences.

\begin{lemma}[Direct Unitary Compilation]
\label{lem:direct-unitary-compilation}

Let $C$ be a size-$s$, depth-$d$ circuit in one of the models
$\QNC$, $\QAC$, $\QTC$, or $\QNC_f$. For every $\eps\in(0,1)$,
there is a circuit in the corresponding universal model that
approximates $C$ to error $\eps$, in the sense of
\Cref{eqn:qac-fixed-gate-simulation}, with depth $O\left(d\log\log(s/\eps)\right)$
and gate count and additional workspace
$O(s\log((s+2)/\eps))$. In particular, at inverse-polynomial accuracy,
\begin{align}
\textsf{Unitary}(\QNC^0)
&\subseteq
\textsf{Unitary}(\UQNC[O(\log\log n)]),
\label{eq:direct-qnc-unitary-inclusion}\\
\textsf{Unitary}(\QAC^0)
&\subseteq
\textsf{Unitary}(\UQAC[O(\log\log n)]),
\label{eq:direct-qac-unitary-inclusion}\\
\textsf{Unitary}(\QTC^0)
&\subseteq
\textsf{Unitary}(\UQTC[O(\log\log n)]),
\label{eq:direct-qtc-unitary-inclusion}\\
\textsf{Unitary}(\QNC_f^0)
&\subseteq
\textsf{Unitary}(\UQNC_f[O(\log\log n)]).
\label{eq:direct-qncf-unitary-inclusion}
\end{align}
When both generalized Toffoli and Fan-Out are available, the compilation
can be performed with no asymptotic depth overhead. More precisely, for
every fixed $\gamma\in(0,1)$, a size-$s$, depth-$d$ circuit in
$\QAC_f$ or $\QTC_f$ can be approximated to error $\eps$ in the
corresponding universal model with depth $O(d)$, gate count
$O_\gamma(s\log((s+2)/\eps))$, and additional workspace $O_\gamma\left( s\log^{1+\gamma}(s/\eps)\right)$
Consequently,
\begin{align}
\textsf{Unitary}(\QAC_f^0)
&=
\textsf{Unitary}(\UQAC_f^0),
\label{eq:direct-qacf-unitary-equality}\\
\textsf{Unitary}(\QTC_f^0)
&=
\textsf{Unitary}(\UQTC_f^0).
\label{eq:direct-qtcf-unitary-equality}
\end{align}
\end{lemma}

\begin{proof}
Set $\eta:=\eps/s$ and let $g_1,\ldots,g_m$, with $m\leq s$, denote the
single-qubit gates of $C$. For each $g_j$, choose a universal
implementation $\widetilde g_j$ with private clean workspace and fixed
phase $\varphi_j$ such that
\begin{align}
    \left\|
        \widetilde g_j
        (\ket{0^{a_j}}\otimes I)
        -
        e^{i\varphi_j}
        (\ket{0^{a_j}}\otimes g_j)
    \right\|_{\mathrm{op}}
    &\leq\eta.
    \label{eq:direct-compilation-local-error}
\end{align}
By \Cref{thm:rz-synthesis-qnc} and \Cref{lem:single-qubit-to-rz}, in the
bounded-width models each replacement has depth
$O(\log\log(1/\eta))$ and uses $O(\log(1/\eta))$ gates and clean
ancillae. When generalized Toffoli and Fan-Out are available,
\Cref{cor:rz-synthesis-qacf} instead gives depth $O_\gamma(1)$,
$O_\gamma(\log(1/\eta))$ gates, and
$O_\gamma(\log^{1+\gamma}(1/\eta))$ clean ancillae. We leave all
multiqubit gates unchanged. The same compilation applies in the
Threshold models because generalized Toffoli is a special case of a
Threshold gate.

Let $\widetilde C$ denote the resulting circuit. A standard hybrid
argument, replacing the $m$ single-qubit gates one at a time, gives
\begin{align}
    \left\|
        \widetilde C(\ket{0^A}\otimes I)
        -
        e^{i\varphi}
        (\ket{0^A}\otimes C)
    \right\|_{\mathrm{op}}
    &\leq
    \sum_{j=1}^m \eta
    \leq
    s\eta
    =
    \eps,
    \label{eq:direct-compilation-total-error}
\end{align}
where $\varphi:=\sum_{j=1}^m\varphi_j$. Here the operator-norm
guarantee in \Cref{eq:direct-compilation-local-error} remains valid
after tensoring with the rest of the circuit, so each replacement may
act on a qubit entangled with the remaining registers. Moreover, each
gate has private initially clean workspace, and the hybrid compares
against the ideal clean execution, so the local guarantees apply
directly.

It remains to account for resources. Replacements of gates belonging
to the same circuit layer act on disjoint data qubits and private
workspace and may therefore be executed in parallel. Hence, in the
bounded-width models, each original layer expands to depth
$O(\log\log(1/\eta))$, giving total depth $O\!\left(d\log\log(s/\eps)\right)$.
Summing the per-gate costs over at most $s$ replacements gives gate
count and additional workspace $O\!\left(s\log(s/\eps)\right).$
When generalized Toffoli and Fan-Out are available, each replacement
has constant depth, so the total depth remains $O(d)$, while the gate
count and additional workspace are, respectively, $O_\gamma\!\left(s\log(s/\eps)\right)$ and $O_\gamma\!\left(s\log^{1+\gamma}(s/\eps)\right)$.

Finally, for polynomial-size circuit families and inverse-polynomial
target accuracy, $\eta=\eps/s$ remains inverse polynomial. Thus the
bounded-width compilation incurs $O(\log\log n)$ depth per original
layer, while the wide-gate compilation incurs only constant depth.
The reverse inclusions in the resulting wide-gate equalities follow
immediately from the gate-set inclusions.

\end{proof}

\begin{corollary}[Direct restricted-gate simulations]
\label{cor:shallow-restricted-gate-simulations}
For bounded-error decision problems,
\begin{align}
\QNC^0&=\UQNC^0,
\label{eq:qnc0-restricted-inclusion}\\
\QAC^0&\subseteq\UQAC[O(\log\log n)],
\label{eq:qac0-restricted-inclusion}\\
\QTC^0&\subseteq\UQTC[O(\log\log n)],
\label{eq:qtc0-restricted-inclusion}\\
\QNC_f^0&\subseteq\UQNC_f[O(\log\log n)],
\label{eq:qncf-restricted-inclusion}\\
\QAC_f^0&=\UQAC_f^0,
\label{eq:qacf-restricted-equality}\\
\QTC_f^0&=\UQTC_f^0.
\label{eq:qtcf-restricted-equality}
\end{align}
\end{corollary}

\begin{proof}
We first consider $\QNC^0$. Because every gate has bounded arity and
the circuit has constant depth, the backward light cone of the
designated output qubit contains only $O(1)$ input bits. Hence, at each
input length, the acceptance probability depends on only constantly
many bits of the input. The bounded-error circuit therefore computes a
Boolean function of $O(1)$ bits, which has a constant-size truth table
and can be implemented exactly by a constant-size reversible circuit
over $X$, CNOT, and bounded Toffoli gates. Thus
$\QNC^0\subseteq\UQNC^0$, while the reverse inclusion follows
immediately from the gate sets.

For the remaining classes, we instead compile the entire circuit. First
amplify the source computation, using a constant number of parallel
repetitions followed by reversible majority, so that YES instances are
accepted with probability at least $9/10$ and NO instances with
probability at most $1/10$. Since the number of repetitions is
constant, both copying the classical input and computing the majority
incur only constant-depth and constant-factor size overhead.

Let $s(n)\geq1$ bound the size of the amplified circuit and set
$\eps_0:=1/20$. In the compilation of
\Cref{lem:direct-unitary-compilation}, synthesize each single-qubit gate
to accuracy $\eta_n:=\eps_0/s(n).$
Since at most $s(n)$ gates are replaced, a hybrid argument gives total
simulation error at most $s(n)\eta_n=\eps_0$.
Thus, on every input, the final state of the compiled circuit is within
Euclidean distance $\eps_0$ of the ideal amplified computation, up to
an irrelevant fixed global phase. The trace distance between the two
output states is therefore at most $\eps_0$, so any measurement
probability---in particular, the acceptance probability---changes by
at most $\eps_0$. Consequently, the compiled circuit has completeness
at least
\begin{align}
    \frac{9}{10}-\frac{1}{20}
    &=
    \frac{17}{20},
\end{align}
and soundness at most
\begin{align}
    \frac{1}{10}+\frac{1}{20}
    &=
    \frac{3}{20}.
\end{align}
Hence the bounded-error gap is preserved.

Finally, because $s(n)$ is polynomial, $\eta_n$ is inverse polynomial.
The bounded-width synthesis from
\Cref{thm:rz-synthesis-qnc,lem:single-qubit-to-rz} therefore contributes
$O(\log\log n)$ depth per original layer, yielding the stated
$\UQAC[O(\log\log n)]$, $\UQTC[O(\log\log n)]$, and
$\UQNC_f[O(\log\log n)]$ simulations. When generalized Toffoli and
Fan-Out are both available, \Cref{cor:rz-synthesis-qacf} gives
constant-depth replacements, so the $\QAC_f^0$ and $\QTC_f^0$
simulations remain constant depth. All size and workspace bounds remain
polynomial, and the reverse inclusions in these latter two equalities
follow directly from the gate-set inclusions.
\end{proof}

\subsection{Depth-Preserving Single-Qubit Gate Synthesis  in \texorpdfstring{$\UQAC^0$}{Q-U-AC0}}
\label{sec:qac-circuit-simulation}

We now combine the two ingredients developed above. Our constant-depth single-qubit synthesis procedure uses generalized Toffoli gates together with Fan-Out of only polylogarithmic width, while \Cref{sec:qac-uqac-collapse} shows that these auxiliary Fan-Out operations can themselves be approximated in constant depth using only Hadamard and generalized Toffoli gates. Consequently, Fan-Out can be removed from the synthesis construction without increasing its asymptotic depth.

This gives a depth-preserving compiler from $\QAC$ circuits with arbitrary single-qubit gates to the fixed universal basis consisting of $H$, $T$, and generalized Toffoli gates. The approximation is clean and at the level of the full unitary action, rather than only its acceptance probability. Namely, the compiled circuit approximates the original circuit on every input state while approximately returning all implementation ancillae to $\ket{0}$. We first state this compilation theorem and then record its single-qubit and complexity-class consequences.

\begin{theorem}[Fixed-Gate Simulation of $\QAC$ Circuits]
\label{thm:qac-fixed-gate-simulation}
Let $C$ be a size-$s$, depth-$d$ $\QAC$ circuit acting on $w$
qubits, where $s\geq1$ counts all gate occurrences. For every
$\eps\in(0,1)$, there is a circuit $\widetilde C$ over $H$, $T$,
and generalized Toffoli gates, using $A=\poly(s,1/\eps)$ additional
clean ancillae, such that, for one fixed phase $\varphi$ and every
normalized $w$-qubit input $\ket\psi$,
\begin{align}
\left|\widetilde C\bigl(\ket{0^A}\ket\psi\bigr)
-e^{i\varphi}\ket{0^A}C\ket\psi\right|_2
&\leq\eps.
\label{eqn:qac-fixed-gate-simulation}
\end{align}
The depth is $O(d)$ and the size and additional workspace are
$\poly(s,1/\eps)$.
\end{theorem}

\begin{proof}
The compilation proceeds in two stages. First, we replace every
arbitrary single-qubit gate by our constant-depth synthesis circuit,
temporarily retaining the polylogarithmic-width Fan-Out gates used by
that construction. We then replace each of these Fan-Out gates by its
constant-depth implementation over Hadamard and generalized Toffoli.
We allocate error $\eps/2$ to each stage.

We begin with compilation  of the single-qubit gates.
Fix $\gamma=1/4$ in \Cref{cor:rz-synthesis-qacf} and combine it with
\Cref{lem:single-qubit-to-rz}. For every single-qubit unitary $U$ and
accuracy $\eta$, this gives a constant-depth circuit $S_{U,\eta}$ and
an input-independent phase $\varphi_{U,\eta}$ such that
\begin{align}
    \sup_{\|\ket\psi\|_2=1}
    \left\|
        S_{U,\eta}(\ket{0^{a_\eta}}\ket\psi)
        -
        e^{i\varphi_{U,\eta}}
        \ket{0^{a_\eta}}U\ket\psi
    \right\|_2
    &\leq\eta.
    \label{eqn:wide-compiler-fixed-phase-guarantee}
\end{align}
Each such synthesizer uses $O(\log(1/\eta))$ Fan-Out gates, each with
$O(\log^{3/4}(1/\eta))$ targets. All other gates already have
constant-size exact implementations over the desired basis:
$Z=T^4$, $S=T^2$, $S^\dagger=T^6$, $X=HT^4H$. Note that CNOT is a
generalized Toffoli and CZ is obtained from CNOT by conjugating its
target with Hadamards.

Set $\eta:=\eps/(2s)$. Replace every single-qubit gate of $C$ by its
corresponding $S_{U,\eta}$ using private clean workspace, while leaving
the generalized Toffoli gates unchanged. Let $C'$ denote the resulting
circuit. Since at most $s$ gates are replaced, a standard hybrid
argument gives, for one fixed phase $\varphi$ and every normalized
input $\ket\psi$,
\begin{align}
    \left\|
        C'(\ket{0^A}\ket\psi)
        -
        e^{i\varphi}\ket{0^A}C\ket\psi
    \right\|_2
    &\leq
    s\eta
    =
    \frac{\eps}{2}.
    \label{eqn:wide-compiler-total-error}
\end{align}
Here $A$ includes both the synthesis workspace and the clean workspace
reserved for the Fan-Out replacements below. Adding unused zero
registers does not affect the bound.

We now elimate Fan-Out. The circuit $C'$ contains $M=O\!(s\log(s/\eps))$
Fan-Out gates, each acting on at most $L=O\!(\log^{3/4}(s/\eps))$
targets. If $M=0$, the first stage already proves the theorem.
Otherwise, assign error
$\rho:=\eps/(2M)$ to each Fan-Out replacement and choose
\begin{align}
    q
    &:=
    \max\left\{
        4,\,
        L,\,
        \left\lceil
            2\log_2\frac{16}{\rho}
        \right\rceil
    \right\}
    =
    O\!\left(
        \log\frac{s}{\eps}
    \right).
\end{align}
By \Cref{lem:uqac-logarithmic-fanout}, every such Fan-Out admits a
constant-depth implementation over Hadamard and generalized Toffoli
with clean-input error at most $\rho$ and
$O(q4^q)=\poly(s,1/\eps)$ gates and workspace.

Replace the $M$ Fan-Out occurrences one at a time in a second hybrid
argument. The local guarantee applies even when the Fan-Out data
register is entangled with the rest of the computation. Only its
private implementation workspace must begin in the all-zero state.
Thus the resulting circuit $\widetilde C$ satisfies
\begin{align}
    \left\|
        \widetilde C(\ket{0^A}\ket\psi)
        -
        C'(\ket{0^A}\ket\psi)
    \right\|_2
    &\leq
    M\rho
    =
    \frac{\eps}{2}.
    \label{eqn:fanout-replacement-total-error}
\end{align}
Combining
\Cref{eqn:wide-compiler-total-error,eqn:fanout-replacement-total-error}
with the triangle inequality proves
\Cref{eqn:qac-fixed-gate-simulation}.

Finally, the total cost of all Fan-Out replacements is
$O(Mq4^q)=\poly(s,1/\eps)$, and the first-stage synthesis circuits also
use only polynomially many gates and ancillae. Gates from the same
original circuit layer act on disjoint data registers and use private
workspace, so their compiled implementations run in parallel. Both
compilation stages increase the depth of each original layer by only a
constant factor, and therefore
$\operatorname{depth}(\widetilde C)=O(d)$.
\end{proof}

\noindent The theorem shows that the constant-depth synthesis result previously obtained in $\UQAC_f^0$ does not require Fan-Out as a primitive, with sustained access to generalized Toffoli. Thus, the polylogarithmic-width Fan-Out used by the synthesis can be removed in $\UQAC^0$.

\begin{corollary}[Depth-Preserving Universalization of $\QAC$]
\label{cor:qac-uqac-collapse}
For every depth $D(n)\geq1$,
\begin{align}
\textsf{Unitary}(\QAC[D])
&=
\textsf{Unitary}(\UQAC[D]).
\label{eq:unitary-qac-uqac-collapse}
\end{align}
For bounded-error decision problems,
\begin{align}
\QAC[D]
&=
\UQAC[D].
\label{eqn:qac-uqac-collapse}
\end{align}
In particular, these equalities hold for $\QAC^k$ and $\UQAC^k$ for every fixed $k\geq0$.
\end{corollary}

\begin{proof}
As an immediate special case of
\Cref{thm:qac-fixed-gate-simulation}, every single-qubit unitary admits
a constant-depth approximation over $H$, $T$, and generalized Toffoli
gates with polynomial overhead in the inverse accuracy.

More generally, apply \Cref{thm:qac-fixed-gate-simulation} to each
circuit in a depth-$O(D(n))$ $\QAC$ family, using sufficiently small
inverse-polynomial compilation error for unitary synthesis or
sufficiently small constant error for bounded-error decisions, after the constant-error amplification described in the proof of \Cref{cor:shallow-restricted-gate-simulations}. The
compiler increases depth by only a constant factor and size only
polynomially, so the resulting family remains depth $O(D(n))$.
This proves the forward inclusions in both statements. The reverse
inclusions follow immediately from the gate sets.
\end{proof}

\noindent
The same compiler applies to threshold circuits, since every generalized
Toffoli introduced by the compilation is already available as a threshold
operation.

\begin{corollary}[Depth-Preserving Universalization of $\QTC$]
\label{cor:qtc-uqtc-collapse}
For every depth bound $D(n)\geq1$,
\begin{align}
\textsf{Unitary}(\QTC[D])
&=
\textsf{Unitary}(\UQTC[D]).
\label{eq:unitary-qtc-uqtc-collapse}
\end{align}
For bounded-error decision problems,
\begin{align}
\QTC[D]
&=
\UQTC[D].
\label{eqn:qtc-uqtc-collapse}
\end{align}
In particular, these equalities hold for $\QTC^k$ and $\UQTC^k$ for every fixed $k\geq0$.
\end{corollary}

\begin{proof}
Apply the proof of \Cref{cor:qac-uqac-collapse}. An $m$-controlled
generalized Toffoli is a threshold operation with threshold $m$, so
every multiqubit gate introduced by the compiler is already available
in $\UQTC$. The depth and resource bounds are unchanged, and the
reverse inclusions again follow from the gate sets.
\end{proof}

\subsection{Additive-Depth Compilation in \texorpdfstring{$\UQNC_f$}{Q-U-NCf} with a Catalyst}
\label{sec:catalytic-circuit-simulation}

The preceding section showed that the constant-depth single-qubit
synthesis originally obtained in $\UQAC_f^0$ can in fact be
implemented in $\UQAC^0$. The polylogarithmic-width Fan-Out gates used
by the synthesis can be replaced in constant depth using generalized
Toffoli gates. The analogous question for $\QNC_f$ goes in the
opposite direction. Here Fan-Out is available as a primitive, but
generalized Toffoli is not. Takahashi and Tani showed that generalized Toffoli can be implemented
in constant depth from Fan-Out~\cite{takahashi2016collapse}, yielding
the standard equality $\QNC_f^0=\QAC_f^0$. Their construction,
however, uses arbitrary single-qubit gates, including angle-dependent
phase rotations, and therefore does not directly imply
$\QNC_f^0=\UQNC_f^0$. We leave constant-depth universalization of
$\QNC_f^0$ as an open question.

For larger depths, we instead use the catalyst prepared in
\Cref{sec:kim-comparison}. Once the catalyst bank is available, the
Kim--Laakkonen phase-kickback construction implements each phase layer
in constant depth using Fan-Out and bounded-arity gates. Crucially,
the ideal catalyst is returned unchanged and can therefore be reused
throughout the computation. We prepare the bank once, perform the
entire online computation, and then reverse the preparation circuit.
Thus the catalyst-preparation error is incurred only once, rather than
once per use, while the preparation depth contributes only an additive
$O(\log\log n)$ overhead.

To implement a grid phase $\omega_N^d$, we use the catalyst eigenvalue relation
$U_f^d\ket{\psi_1}=\omega_N^d\ket{\psi_1}$. Thus, if the phase is to be applied conditioned on a data qubit, the online circuit must implement the controlled transformation
$\ket c\ket x\mapsto\ket c\ket{C_f^{dc}x}$ on the catalyst register. Since $C_f^d$ is an invertible linear map over $\mathbb F_2$, the required online operations are controlled linear transformations. Kim and Laakkonen give circuits for such controlled linear maps with bounded Toffoli depth~\cite[Thm.~3]{kim2025clifford}. The next lemma records an explicit implementation in our $\UQNC_f$ model and keeps track of its total depth, gate count, and workspace.

\begin{lemma}[Controlled Linear Maps]
\label{lem:controlled-linear-map}
For every $A\in\mathrm{GL}_b(\mathbb F_2)$, $c\in\{0,1\}$, and $x\in\mathbb F_2^b$, the transformation
\begin{align}
\ket c\ket x
&\longmapsto
\ket c\ket{A^c x},
\qquad \text{where} \quad 
A^c=
\begin{cases}
I,&c=0,\\
A,&c=1,
\end{cases}
\label{eq:controlled-linear-map}
\end{align}
has an exact clean
$\UQNC_f$ implementation of depth $O(1)$ and size $O(b^2)$.
\end{lemma}

\begin{proof}
First compute $A^cx$ into a separate register. For each coordinate,
let $S_i:=\{j:A_{ij}\neq\delta_{ij}\}$. Over $\mathbb F_2$,
\begin{align}
(A^cx)_i
&=
x_i\oplus\bigoplus_{j\in S_i}cx_j.
\label{eq:controlled-linear-coordinate}
\end{align}
Fan out $c$ to $b$ private controls and compute all products $cx_j$
in parallel with ordinary Toffoli gates. Distribute each product to
its occurrences in the sets $S_i$, then XOR $x_i$ and those products
into the $i$th target using Parity. By
\Cref{eqn:parity-fanout-hadamard-duality}, these parity operations
have constant-depth implementations. Private copies make the
coordinate computations disjoint. Reversing the product and copying
steps clears the workspace, giving the exact oracle
\begin{align}
    O_A:\ket c\ket x\ket y
    \longmapsto
    \ket c\ket x\ket{y\oplus A^cx}.  
\end{align}
There are at most $b^2$ coordinate incidences, so the gate count and
workspace are $O(b^2)$, with constant depth.

For an in-place implementation, compute $A^cx$ into a fresh zero
register, swap it with the input register using parallel SWAPs, and
then apply $O_{A^{-1}}$ to erase the old input. This works because
$(A^{-1})^cA^c=I$ for both values of $c$, and preserves the stated
resource bounds.
\end{proof}
By \Cref{lem:single-qubit-to-rz}, every single-qubit gate can be
written, up to a fixed global phase, as a constant-length sequence of
Hadamard gates and phase gates
$P(\theta):=\operatorname{diag}(1,e^{i\theta})$. We round each phase
to the catalyst grid and implement the rounded phase by kickback from
$\ket{\psi_1}$.

\begin{theorem}[Additive-Depth Catalytic Universalization]
\label{thm:catalytic-universalization}
Let $C$ be a size-$s$, depth-$D$ $\QNC_f$ circuit. After decomposing
its single-qubit gates, up to global phase, into Hadamard and phase
gates, let $r$ be the total number of phase gates and let $w$ be the
maximum number appearing in any one layer. For every $\delta\in(0,1)$, there is a $\UQNC_f$ circuit
$\widetilde C$, using $A=\poly(s,1/\delta)$ additional clean ancillas,
such that for one fixed phase $\varphi$ and every normalized input
$\ket\psi$,
\begin{align}
\left|
\widetilde C\bigl(\ket{0^A}\ket\psi\bigr)
-
e^{i\varphi}\ket{0^A}C\ket\psi
\right|_2
&\leq
\delta.
\label{eq:catalytic-clean-unitary-guarantee}
\end{align}
The size and workspace are polynomial in $s$ and $1/\delta$, while the depth is
\begin{align}
O\!\left(
D+\log\log\frac{r+w+2}{\delta}
\right).
\label{eq:additive-universalization-depth}
\end{align}
\end{theorem}

\begin{proof}
The proof has three steps. First, we round every phase gate in the
decomposed circuit to the $N$-point phase grid associated with the
catalyst. Second, assuming an ideal bank of $w$ catalysts, we implement
the rounded circuit exactly by phase kickback, reusing the same bank in
every layer. Third, we replace the ideal bank by the output of the
approximate catalyst factory and uncompute the factory at the end. The
rounding error and catalyst-preparation error are budgeted separately;
the latter is incurred only once, independently of the number of
catalyst uses.

The case $r=0$ is immediate, so assume $r,w\geq1$. Let
$N:=2^b-1$. Choose
$b=\Theta(\log((r+w+2)/\delta))$ large enough that
$\pi r/N\leq\delta/2$ and
$\sqrt{w/2^b}\leq\delta/(8\sqrt2)$, and set
$\eta:=\delta/(8\sqrt2\,wb)$. By
\Cref{lem:kim-catalyst-factory}, there is a fully unitary factory
$F$ whose output has trace distance at most
\begin{align}
\xi
&:=
\sqrt{\frac{w}{2^b}}+wb\eta
\leq
\frac{\delta}{4\sqrt2}
\label{eq:catalyst-joint-error}
\end{align}
from $w$ ideal catalysts together with zero factory workspace. Because the factory is unitary and starts from a pure state, write its
output as $\ket*{\widetilde\Gamma}$ and the ideal
catalyst--workspace state as
$\ket\Gamma:=\ket{\psi_1}^{\otimes w}\ket{0^a}$. After fixing an
irrelevant global phase of $\ket\Gamma$, the relation between trace
distance and Euclidean distance for pure states gives
\begin{align}
\left|
\ket*{\widetilde\Gamma}-\ket\Gamma
\right|_2
&\leq
\sqrt2\,\xi.
\label{eq:catalyst-vector-error}
\end{align}

\paragraph{Step 1: round to the catalyst grid.}
For each phase gate $P(\theta)$, choose the nearest grid phase
$P(2\pi d/N)$. The angular error is at most $\pi/N$, and hence
\begin{align}
\left\|
P(\theta)
-
P\!\left(\frac{2\pi d}{N}\right)
\right\|_{\mathrm{op}}
&\leq
\frac{\pi}{N}.
\label{eq:phase-grid-error}
\end{align}
Let $C^{(b)}$ denote the circuit obtained by replacing all $r$ phase
gates by their rounded versions. A hybrid argument gives, for one fixed
global phase $\varphi$,
\begin{align}\left\|
C^{(b)}-e^{i\varphi}C
\right\|_{\mathrm{op}}
&\leq
\frac{\pi r}{N}
\leq
\frac{\delta}{2}.
\label{eq:rounded-unitary-error}
\end{align}
Thus it remains to implement $C^{(b)}$ using the catalyst bank.

\paragraph{Step 2: implement the rounded circuit with an ideal catalyst.}
A rounded phase $P(2\pi d/N)$ is implemented by controlling the linear
map $C_f^d$ on the corresponding data qubit. By
\Cref{eq:kim-eigenvalue},
\begin{align}
\operatorname{ctrl}(U_f^d)
\bigl(
(c_0\ket0+c_1\ket1)\ket{\psi_1}
\bigr)
&=
(c_0\ket0+c_1\omega_N^d\ket1)
\ket{\psi_1}.
\label{eq:compact-kim-kickback}
\end{align}
The desired phase is therefore transferred to the data while the
catalyst is returned exactly to $\ket{\psi_1}$. Assign distinct catalysts to the at most $w$ phase gates in any one
layer. By \Cref{lem:controlled-linear-map}, the corresponding
controlled maps have constant depth and use private clean workspace,
so all phase gates in a layer can be implemented in parallel. Their
workspace is then uncomputed. Since each ideal catalyst is restored
after use, the same bank can be reused in every layer. Let $V_b$ denote the complete online implementation of $C^{(b)}$, and
let $h$ denote its workspace. For every input $\ket\psi$,
\begin{align}
V_b
\bigl(
\ket\psi\ket\Gamma\ket{0^h}
\bigr)
&=
C^{(b)}\ket\psi
\ket\Gamma\ket{0^h}.
\label{eq:ideal-catalyst-reuse}
\end{align}
Thus the online computation has depth $O(D)$.

\paragraph{Step 3: replace the ideal catalyst by the prepared one.}
The actual compiled circuit applies the factory $F$, then $V_b$, and
finally $F^\dagger$. The inverse factory maps $\ket*{\widetilde\Gamma}=F\ket{0^{a+wb}}$ exactly to zero. We therefore compare the online output with $C^{(b)}\ket\psi\ket*{\widetilde\Gamma}\ket{0^h}$ before applying $F^\dagger$. The estimate below then bounds the final cleanup error. By unitary invariance,
\begin{align}
&\left|
(I_{\mathrm{data}}\otimes F^\dagger\otimes I_h)
V_b
\bigl(
\ket\psi\ket*{\widetilde\Gamma}\ket{0^h}
\bigr)
-
C^{(b)}\ket\psi\ket{0^{a+wb+h}}
\right|_2=
\left|
V_b
\bigl(
\ket\psi\ket*{\widetilde\Gamma}\ket{0^h}
\bigr)
-
C^{(b)}\ket\psi
\ket*{\widetilde\Gamma}\ket{0^h}
\right|_2.
\label{eq:catalyst-unprepare-reduction}
\end{align}
Insert the ideal state $\ket\Gamma$ into the two terms on the
right-hand side. Using
\Cref{eq:ideal-catalyst-reuse} and unitarity of $V_b$ gives
\begin{align}
&\left|
V_b
\bigl(
\ket\psi\ket*{\widetilde\Gamma}\ket{0^h}
\bigr)
-
C^{(b)}\ket\psi
\ket*{\widetilde\Gamma}\ket{0^h}
\right|_2\leq
2
\left|
\ket*{\widetilde\Gamma}-\ket\Gamma
\right|_2
\leq
2\sqrt2\,\xi
\leq
\frac{\delta}{2}.
\label{eq:catalyst-cleanup-error}
\end{align}
Importantly, this bound depends only on the initial catalyst
preparation error and not on the number of times the bank is reused.
Combining
\Cref{eq:rounded-unitary-error,eq:catalyst-cleanup-error} with the
triangle inequality proves
\Cref{eq:catalytic-clean-unitary-guarantee}.

\paragraph{Resources.}
The online simulation has depth $O(D)$. The catalyst factory and its inverse each contribute $O(\log b+\log\log(1/\eta))$ depth. Since $b=\Theta(\log((r+w+2)/\delta))$ and $\eta=\Theta(\delta/(wb))$, the overall depth is $O(D+\log\log((r+w)/\delta))$. For size and ancillae, the factory uses $w2^{O(b)}+O(wb\log(1/\eta))$ resources, while the online controlled linear maps use $O(rb^2)$ gates and $O(wb^2)$ reusable workspace. Because $b=O(\log((r+w+2)/\delta))$, the factor $2^{O(b)}$ is polynomial in $(r+w+2)/\delta$. Since $r,w=O(s)$, the total gate count and ancillae are therefore $\poly(s,1/\delta)$.
\end{proof}
\noindent This stronger clean-unitary guarantee immediately yields a universalization result for $\QNC_f$ once the additive catalyst-preparation depth can be absorbed.
\begin{corollary}[Catalytic Universalization Above Loglog Depth]
\label{cor:qnc-f-universalization}
Let $D(n)=\Omega(\log\log n)$. Then, for polynomial-size circuit families,
\begin{align}
\textsf{Unitary}(\QNC_f[D])
&=
\textsf{Unitary}(\UQNC_f[D]).
\label{eq:unitary-qnc-f-universalization}
\end{align}
For bounded-error decision problems,
\begin{align}
\QNC_f[D]
&=
\UQNC_f[D].
\label{eq:qnc-f-universalization}
\end{align}
In particular, for every fixed integer $k\geq1$, $\textsf{Unitary}(\QNC_f^k)=\textsf{Unitary}(\UQNC_f^k)$
and $\QNC_f^k=\UQNC_f^k$.
\end{corollary}

\begin{proof}
For a polynomial-size circuit family, the parameters $r$ and $w$ in
\Cref{thm:catalytic-universalization} are polynomially bounded. At
inverse-polynomial accuracy, the additive catalyst-preparation cost is
therefore $O(\log\log n)$. If $D(n)=\Omega(\log\log n)$, the compiled
depth is
$O(D(n)+\log\log n)=O(D(n))$, while the size remains polynomial. Applying \Cref{thm:catalytic-universalization} with sufficiently small
inverse-polynomial error gives the unitary-synthesis equality. For bounded-error decision problems, first amplify by a constant number of parallel repetitions as in \Cref{cor:shallow-restricted-gate-simulations}. A sufficiently small constant compilation error then preserves the stipulated completeness and soundness thresholds. The
reverse inclusions are immediate from the gate sets.
\end{proof}

\noindent To extend the result to the other Fan-Out models, we use the
following exact simulations from prior work.

\begin{fact}[Exact Fan-Out Simulations
{\cite{takahashi2016collapse}}]
\label{fact:takahashi-tani-simulation}
Generalized Toffoli and threshold gates have exact constant-depth,
polynomial-size implementations using bounded-arity gates, unbounded
Fan-Out, and arbitrary single-qubit gates.
\end{fact}
\noindent Because these simulations increase depth by only a constant factor,
they apply to any asymptotic depth bound.
\begin{fact}[Standard Fan-Out Hierarchy
{\cite{takahashi2016collapse}}]
\label{cor:standard-fanout-hierarchy}
For every depth function $D(n)\geq1$, the bounded-error decision classes
satisfy
\begin{align}
\QNC_f[D]
=
\QAC_f[D]
=
\QTC_f[D].
\label{eq:standard-fanout-collapse}
\end{align}
\end{fact}

\begin{proof}
Apply \Cref{fact:takahashi-tani-simulation} to every generalized
Toffoli and threshold gate, using disjoint workspace for disjoint
gates. Each original layer expands into only constantly many layers,
the size remains polynomial, and the simulation is exact. Hence a
depth-$O(D(n))$ circuit in either larger model becomes a
depth-$O(D(n))$ $\QNC_f$ circuit. The reverse inclusions follow
immediately from the gate sets.
\end{proof}

The Takahashi--Tani simulations may introduce arbitrary single-qubit
gates. By \Cref{cor:qnc-f-universalization}, these can be removed with
an additive $O(\log\log n)$ depth cost. Thus, whenever
$D(n)=\Omega(\log\log n)$, this overhead is absorbed into the original
depth bound.

\begin{corollary}[Universalization of the Fan-Out Hierarchies]
\label{cor:fanout-hierarchy-universalization}
Let $D(n)=\Omega(\log\log n)$. Then the bounded-error decision classes
satisfy
\begin{align}
\QNC_f[D]
=
\UQNC_f[D]
=
\QAC_f[D]
=
\UQAC_f[D]
=
\QTC_f[D]
=
\UQTC_f[D].
\label{eq:fanout-hierarchy-universalization}
\end{align}
In particular, this holds for $D(n)=\log^k n$ for every fixed
integer $k\geq1$.
\end{corollary}

\begin{proof}
A $\QAC_f[D]$ or $\QTC_f[D]$ family first becomes a $\QNC_f[D]$
family by \Cref{cor:standard-fanout-hierarchy}. Applying
\Cref{cor:qnc-f-universalization} then gives a $\UQNC_f$ simulation
of depth $O(D(n)+\log\log n)=O(D(n))$. Thus both unrestricted classes
are contained in $\UQNC_f[D]$. Each corresponding universal class
contains $\UQNC_f[D]$ and is contained in its unrestricted class, so
all six classes coincide.
\end{proof}

\section{The Real Shallow-Depth Circuit Hierarchy}
\label{sec:hadamard-only-collapse}
\label{sec:real-hierarchy-main-result}

We now remove $T$-gates for bounded-error decision computation. Recall that
$\HQNC$, $\HQAC$, and $\HQTC$ are the Hadamard-only versions of
$\UQNC$, $\UQAC$, and $\UQTC$, with the same multiqubit
primitives. The gates $X$ and $Z$ also remain available as the
single-qubit cases of the corresponding multiqubit operations, with
$Z=HXH$. In the bounded-arity models, i.e. $\HQNC$, we explicitly retain ordinary
three-qubit Toffoli in addition to CNOT. Throughout, language classes
are nonuniform and polynomial-size, with computational-basis inputs,
$\ket{0}$-initialized ancillae, and one measured output qubit.

Real simulations of quantum computation have a long history. The standard construction uses one additional two-level system to encode the real and imaginary parts of the amplitudes~\cite{berstein1993quantum,aharonov2003simple}. While efficient in width, this creates a shared auxiliary register touched by every nonreal gate, and therefore need not preserve circuit depth. McKague, Mosca, and Gisin later developed a distributed encoding that restores locality~\cite{mckague2009simulating}.
A more recent line of
work uses the Pauli--Liouville, or Pauli-transfer, representation of
density matrices and quantum channels, in which a Hermitian operator
is represented by its real Pauli coefficients
\cite{wood2015tensor,huang2022bloch,kunold2023vectorization}. In
particular, this representation has been used for classical circuit
simulation and for quantum simulation of density-matrix dynamics.

We use the Pauli-transfer representation for a different purpose, to obtain
a locality-preserving real simulation of shallow quantum circuits.
Rather than tracking the complex amplitudes of the evolving state
vector, we represent its density matrix by its Pauli coefficients.
The $4^q$ coefficients of a $q$-qubit pure density matrix form a
normalized real vector and can therefore be stored as the amplitudes
of only $2q$ ``label'' qubits. Computational-basis inputs have a simple
product encoding, which we propagate gate-by-gate.
Under this representation, conjugation by an $r$-qubit gate $U$
becomes the real orthogonal transformation
\begin{align}
UQU^\dagger
&=
\sum_P [O_U]_{P,Q}P,
\quad \text{where} \quad
[O_U]_{P,Q}
:=
2^{-r}\operatorname{Tr}(P\,UQU^\dagger),
\end{align}
so that, on encoded states, $\ket{\rho}_{\mathrm P}\longmapsto O_U\ket{\rho}_{\mathrm P} = \ket{U\rho U^\dagger}_{\mathrm P}$.
The matrix $O_U$ is referred to as the Pauli-transfer matrix of $U$. Crucially, when
$U$ acts on only a subset of the original qubits, $O_U$ acts only on
the corresponding label registers. Thus, the encoding preserves the locality and parallelism that are lost in the usual realification based on a single shared auxiliary register.

While the Pauli-transfer formalism itself is standard, the new ingredient is
to realize these induced transformations within the corresponding
Hadamard-only shallow circuit models. Fan-Out is Clifford and hence
acts as a signed permutation of Pauli labels. For generalized Toffoli,
whose induced action is both non-Clifford and unbounded-arity, we
identify the Pauli-transfer matrix as a structured reflection and
implement it by a coherent agreement test. Threshold circuits are
handled indirectly through the Fan-Out model. To our knowledge, this
use of the Pauli-transfer representation to obtain depth-preserving
real simulations and the resulting shallow-depth hierarchy collapses
has not appeared previously.

\begin{theorem}[Depth-Preserving Real Decision Simulation]
\label{thm:h-only-depth-preserving-simulation}
Let $(\mathcal U,\mathcal R)$ be one of the corresponding universal--real circuit pairs:
\begin{align*}
    (\UQAC,\HQAC),\qquad
    (\UQAC_f,\HQAC_f),\qquad
    (\UQTC,\HQTC),\qquad
    (\UQTC_f,\HQTC_f).
\end{align*}
Let $C$ be a depth-$d$ circuit in $\mathcal U$, and let $s\geq2$ bound its gate count and width. Then, for every $\varepsilon\in(0,1/2)$,
there exists a circuit $\widetilde C$ in the corresponding real class
$\mathcal R$, of depth $O(d+1)$ and size 
$\operatorname{poly}(s,1/\varepsilon)$, such that, for every $x\in\{0,1\}^n$,
\begin{align}
    |p_{\widetilde C}(x)-p_C(x)|
    &\leq\varepsilon
    \label{eqn:h-only-depth-preserving-acceptance}
\end{align}
\end{theorem}
\noindent The bounded-arity classes $\UQNC$ and $\UQNC_f$ are treated separately:
their real simulation uses a reusable catalyst whose preparation incurs an
additive $O(\log\log(s/\varepsilon))$ depth cost, as in
\Cref{sec:catalytic-circuit-simulation}; see
\Cref{thm:bounded-arity-real-decision-simulation}.

\paragraph{Proof outline.}
\Cref{sec:pauli-real-encoding} introduces the local Pauli encoding and
readout. \Cref{sec:h-only-real-primitives} implements the real rotation
induced by $T$. \Cref{sec:real-fanout-threshold} treats Fan-Out as a
linear label update plus a low-degree sign and records the threshold
reductions. The main technical step, \Cref{sec:encoded-toffoli-real},
expresses encoded Toffoli as a reflection whose negative eigenspace is
detected by local checks and randomized comparisons of selected
$Y$-eigenvalues. \Cref{sec:real-simulation-completion} assembles the
wide-gate simulations, while \Cref{sec:real-bounded-fanout} handles
bounded arity using exact synthesis and a reusable real catalyst.

\subsection{A Local Real Encoding}
\label{sec:pauli-real-encoding}

We now construct the real encoding used throughout the simulation. The
idea is to represent the evolving density matrix by its real Pauli
coefficients. Although there are $4^q$ such coefficients for a
$q$-qubit state, they are stored coherently as the amplitudes of only
$2q$ label qubits. Crucially, conjugation by a gate acting on a set of
qubits induces an operation only on the corresponding label qubits, so
the encoding preserves the locality and parallelism of the original
circuit.

Let $q$ denote the total number of qubits of $C$, including its
workspace, and write $C=L_d\cdots L_1$. On classical input $x$, let
\begin{align}
\ket{\psi_t(x)}
&=
L_t\cdots L_1\ket{x}\ket{0^{q-n}}
\qquad \text{and} \qquad
\rho_t(x)
:=
\ket{\psi_t(x)}\bra{\psi_t(x)}.
\label{eqn:original-state-evolution}
\end{align}
If $\Pi=\ket{1}\bra{1}_\textsf{out}\otimes I$ projects onto outcome one of the
designated output qubit, then
\begin{align}
p_C(x)
&=
\operatorname{Tr}\bigl(\Pi\rho_d(x)\bigr).
\label{eqn:original-acceptance-probability}
\end{align}
Since the computation is unitary prior to the final measurement, every
$\rho_t(x)$ is pure and satisfies $\rho_t(x)^2=\rho_t(x)$. We label the $q$-qubit Pauli operators by pairs of bit strings
$a,b\in\{0,1\}^q$. Define
\begin{align}
P(a,b)
&:=
\bigotimes_{j=1}^q i^{a_jb_j}X^{a_j}Z^{b_j},
\label{eqn:pauli-labels}
\end{align}
so the local labels $00,01,10,11$ correspond to $I,Z,X,Y$,
respectively. For a Hermitian operator $\rho$, define its
\emph{Pauli encoding} as
\begin{align}
\ket{\rho}_{\mathrm P}
&:=
2^{-q/2}
\sum_{a,b\in\{0,1\}^q}
\operatorname{Tr}(P(a,b)\rho)\ket{a,b}.
\label{eqn:pauli-vectorization}
\end{align}
The amplitudes are real because both $P(a,b)$ and $\rho$ are Hermitian.
Pauli orthogonality further gives, for Hermitian $A$ and $B$,
\begin{align}
\braket{A}{B}_{\mathrm P}
&=
\operatorname{Tr}(AB).
\label{eqn:pauli-hilbert-schmidt}
\end{align}
Hence
$|\ket{\rho}_{\mathrm P}|_2^2=\operatorname{Tr}(\rho^2)$, so each
$\ket{\rho_t(x)}_{\mathrm P}$ is a normalized real state.

The two label qubits associated with an original qubit should not be
viewed as separate real- and imaginary-part registers. Their joint
value instead labels one of $I,Z,X,Y$, and complex phase information
in the original state is encoded through real Pauli coefficients. For
example,
\begin{align}
\rho
&=
\frac12\left(I+r_xX+r_yY+r_zZ\right)
\end{align}
has entirely real coefficients, with the $Y$ coefficient $r_y$
carrying information that would appear as relative phase in a
state-vector description. Crucially, the simulator never computes the $4^q$
coefficients individually. Instead, it prepares their coherent encoding and
evolves it directly.

The following lemma proves three properties of this encoding that
will be necessary for the later analysis. In particular, the encoded input is easy to prepare, gates act locally on the
label registers, and the original acceptance probability can be read
out locally.

\begin{lemma}[Local Pauli Encoding and Readout]
\label{lem:real-pauli-encoding}
The Pauli encoding in \Cref{eqn:pauli-vectorization} satisfies:
\begin{enumerate}
\item \textbf{Input preparation.}
For every $z\in\{0,1\}^q$,
\begin{align}
\bigl|\ket{z}\bra{z}\bigr\rangle_{\mathrm P}
&=
\ket{0^q}_a\otimes H^{\otimes q}\ket{z}_b.
\label{eqn:encoded-basis-state}
\end{align}
Thus a computational-basis input has a depth-one encoded
preparation using only Hadamard gates.

\item \textbf{Local gate evolution.}
Suppose $U$ acts on a set $S$ of $r$ original qubits, with Pauli-transfer matrix
\begin{align}
[O_U]_{P,Q}
&:=
2^{-r}\operatorname{Tr}\!\left(P\,UQU^\dagger\right),
\label{eqn:local-pauli-transfer}
\end{align}
where $P,Q$ range over the $r$-qubit Paulis on $S$. Then, $O_U$ is
real orthogonal and, for all Hermitian $\rho$,
\begin{align}
\bigl|U\rho U^\dagger\bigr\rangle_{\mathrm P}
&=
(O_U)_S\ket{\rho}_{\mathrm P},
\label{eqn:local-pauli-evolution}
\end{align}
where $(O_U)_S$ acts only on the $2r$ label qubits associated with
$S$.

\item \textbf{Local readout.}
Let $\textsf{out}$ be the designated output qubit, and let
$(a_{\textsf{out}},b_{\textsf{out}})$ denote its two Pauli-label qubits.
For
$\Pi=\ket{1}\bra{1}_{\textsf{out}}\otimes I$, define the effect
\begin{align}
M_\Pi
&:=
\ket{0}\bra{0}_{a_{\textsf{out}}}
\otimes
\ket{-}\bra{-}_{b_{\textsf{out}}}
+
\frac12
\ket{1}\bra{1}_{a_{\textsf{out}}}
\otimes
I_{b_{\textsf{out}}}.
\label{eqn:encoded-readout-effect}
\end{align}
Then, for every pure state $\rho$,
\begin{align}
\bra{\rho}
M_\Pi
\ket{\rho}_{\mathrm P}
&=
\operatorname{Tr}(\Pi\rho).
\label{eqn:encoded-readout-probability}
\end{align}
Moreover, the measurement $\{M_\Pi,I-M_\Pi\}$ can be implemented
exactly by a constant-size $\HQNC^0$ circuit, followed by one computational-basis
measurement.

\end{enumerate}
\end{lemma}

\begin{proof}
For the first claim,
$\ket{z}\bra{z}
=2^{-q}\prod_{j=1}^q(I+(-1)^{z_j}Z_j)$, so only Pauli strings with
$a=0^q$ have nonzero coefficients. Therefore
\begin{align}
\bigl|\ket{z}\bra{z}\bigr\rangle_{\mathrm P}
&=
\ket{0^q}_a
\otimes
2^{-q/2}
\sum_{b\in\{0,1\}^q}
(-1)^{z\cdot b}\ket{b}
=
\ket{0^q}_a\otimes H^{\otimes q}\ket{z}_b.
\end{align}
Thus we may use the original computational-basis input as the
$b$-register, append $\ket{0^q}$ as the $a$-register, and apply one
parallel layer of Hadamards.

For the second claim, conjugation $A\mapsto UAU^\dagger$ preserves both
Hermiticity and the Hilbert--Schmidt inner product. Its matrix in the
normalized Pauli basis is therefore real orthogonal, with entries
given by \Cref{eqn:local-pauli-transfer}. Expanding $\rho$ in this
basis immediately gives \Cref{eqn:local-pauli-evolution}. Since $U$
acts trivially outside $S$, its Pauli-transfer matrix likewise acts as
the identity on all label registers outside $S$. In particular,
disjoint gates in the original circuit remain disjoint in the encoded
computation.

For the final claim, consider the Hilbert--Schmidt-space operator
corresponding to the Jordan map
\begin{align}
A
&\longmapsto
\frac12(\Pi A+A\Pi).
\end{align}
Its matrix in the normalized Pauli basis is precisely $M_\Pi$. Hence,
for any Hermitian $\rho$,
\begin{align}
\bra{\rho}
M_\Pi
\ket{\rho}_{\mathrm P}
&=
\operatorname{Tr}(\Pi\rho^2).
\label{eqn:pauli-output-probability}
\end{align}
When $\rho$ is pure, $\rho^2=\rho$, yielding
\Cref{eqn:encoded-readout-probability}.

On the output label pair, this effect takes the form
\Cref{eqn:encoded-readout-effect}. To implement it, apply $H$ to
$b_\textsf{out}$, prepare a fresh qubit $c$ in $\ket{+}$, and reversibly select
$b_\textsf{out}$ when $a_\textsf{out}=0$ and $c$ when $a_\textsf{out}=1$ into the measured output
qubit. This requires only a constant number of Hadamard and Toffoli
gates, meaning it can be implemented by a $\HQNC^0$ circuit.
\end{proof}

\noindent We also record the Pauli-transfer actions of the two nonclassical
single-qubit gates.

\begin{proposition}[Encoded Hadamard]
\label{prop:encoded-hadamard}
The Pauli-transfer matrix of the Hadamard gate acts on its two label
qubits $(a,b)$ as
\begin{align}
O_H\ket{a,b}
&=
(-1)^{ab}\ket{b,a}.
\label{eqn:encoded-h}
\end{align}
In particular, $O_H$ is implemented exactly by a SWAP followed by a CZ.
\end{proposition}

\begin{proof}
Conjugation by $H$ acts on the single-qubit Pauli basis as
\begin{align}
HIH&=I,
&
HZH&=X,
&
HXH&=Z,
&
HYH&=-Y.
\end{align}
Thus $H$ exchanges the $X$- and $Z$-components of the Pauli label,
corresponding to $(a,b)\mapsto(b,a)$, and introduces a minus sign
precisely for the $Y$ label $(a,b)=(1,1)$. This gives
\Cref{eqn:encoded-h}. The permutation $(a,b)\mapsto(b,a)$ is a SWAP,
while the phase $(-1)^{ab}$ is a CZ.
\end{proof}

\noindent For the $T$-gate, we use the notation
\begin{align}
\Lambda_a(V_b)
&:=
\ket{0}\bra{0}_a\otimes I_b
+
\ket{1}\bra{1}_a\otimes V_b
\label{eqn:controlled-label-operation}
\end{align}
for $V$ applied to the $b$-label conditioned on $a=1$.

\begin{proposition}[Encoded $T$-gate]
\label{prop:encoded-t}
The Pauli-transfer matrix of the $T$-gate acts on its two label qubits $(a,b)$ as
\begin{align}
O_T\ket{a,b}
&=
\Lambda_a \left(R_y(\pi/2)_b\right)\ket{a,b}.
\label{eqn:encoded-t}
\end{align}
Thus $O_T$ acts trivially when $a=0$ and applies $R_y(\pi/2)$ to the
$b$-label when $a=1$. In particular, the complex phase introduced by
$T$ becomes an ordinary controlled real rotation on the Pauli-label
space.
\end{proposition}

\begin{proof}
Conjugation by $T$ fixes $I$ and $Z$, while
\begin{align}
TXT^\dagger
&=
\frac{X+Y}{\sqrt2}
\qquad \text{and} \qquad
TYT^\dagger
=
\frac{-X+Y}{\sqrt2}.
\end{align}
The label $a=0$ corresponds to the $\{I,Z\}$ subspace, so $O_T$ acts
as the identity there. The label $a=1$ corresponds to the
$\{X,Y\}$ subspace, with $b=0$ and $b=1$ labeling $X$ and $Y$,
respectively. Since
\begin{align}
R_y(\pi/2)\ket{0}
&=
\frac{\ket{0}+\ket{1}}{\sqrt2}
\qquad \text{and} \qquad
R_y(\pi/2)\ket{1}
=
\frac{-\ket{0}+\ket{1}}{\sqrt2},
\end{align}
its action on the $b$-label exactly reproduces the above rotation of
$X$ and $Y$. Hence
$O_T=\Lambda_a(R_y(\pi/2)_b)$.
\end{proof}

\subsection{Real Rotations and the Encoded \texorpdfstring{$T$}{T} Gate}
\label{sec:h-only-real-primitives}

By \Cref{prop:encoded-t}, the Pauli-transfer action of a $T$-gate is
the controlled real rotation
\begin{align}
O_T
&=
\Lambda(R_y(\pi/2)).
\end{align}
This is where the restriction to decision computation becomes
essential. We do not reproduce the original complex $T$-gate on an
arbitrary quantum state. Instead, the Pauli encoding represents the
evolving density matrix by a real state, and preserving the final
acceptance probability requires only that we implement the induced
real transformation $O_T$.

It therefore suffices to synthesize real $y$-rotations without
reintroducing $T$. The next lemma shows more generally that arbitrary
real rotations, and their controlled versions, admit efficient
Hadamard-only implementations.

\begin{lemma}[Real Rotations]
\label{lem:h-only-real-rotations}
Fix $\theta\in\mathbb R$ and $\eta\in(0,1/2)$. For either
\begin{align}
U\in\left\{R_y(\theta),\Lambda(R_y(\theta))\right\},
\end{align}
there exists a constant-depth $\HQAC$ circuit $\widetilde U$ using
$A=\operatorname{poly}(1/\eta)$ clean ancillae and
$\operatorname{poly}(1/\eta)$ gates such that
\begin{align}
\sup_{\|\ket{\phi}\|_2=1}
\left\|
\widetilde U\bigl(\ket{\phi}\ket{0^A}\bigr)
-
(U\ket{\phi})\ket{0^A}
\right\|_2
&\leq \eta.
\label{eqn:real-rotation-approximation}
\end{align}
Moreover, there is an $\HQNC$ circuit satisfying the same guarantee
with depth $O(\log\log(4/\eta))$ and
$\operatorname{poly}(\log(1/\eta))$ gates and clean ancillae.
\end{lemma}

\noindent
The guarantee in \Cref{eqn:real-rotation-approximation} compares the
full data--ancilla isometries. It therefore remains valid when the
acted-on register is entangled with an arbitrary reference system,
allowing the synthesized rotations to be substituted gate-by-gate in
the encoded computation. Throughout the constructions below, each
approximate primitive is given fresh clean ancillae, including
occurrences used in an uncomputation. A standard telescoping hybrid
then bounds the total error by the sum of the individual clean-input
errors.

\begin{proof}
We first synthesize $R_y(\theta)$. The key observation is that almost
all of the synthesis construction from \Cref{sec:be_gate_synth}
already uses only real gates. The only place where complex
single-qubit gates enter is the SELECT operation implementing the
linear-combination branches. For a real $y$-rotation, those branches
can instead be chosen entirely from the Hadamard-only gate set.
Consequently, the dyadic block encoding, oblivious amplitude
amplification, and Fan-Out compilation from
\Cref{sec:be_gate_synth} otherwise carry over unchanged.

Set $J:=XZ$. Since $J^2=-I$,
\begin{align}
R_y(\theta)
&=
\cos(\theta/2)I+\sin(\theta/2)J.
\label{eqn:real-rotation-lcu}
\end{align}
We therefore apply the dyadic block-encoding construction of
\Cref{sec:be_gate_synth} with SELECT branches
\begin{align}
I,\qquad -I,\qquad J,\qquad -J.
\label{eqn:real-rotation-select}
\end{align}
These replace the complex branches used in the earlier $R_z$
construction. Crucially, all four are available without $T$.
Indeed, $Z=HXH$ and, with the rightmost gate applied first,
$\Lambda(J)=\operatorname{CNOT}\operatorname{CZ}$, while a $Z$
on the branch indicator supplies the minus sign for the negative
branches. Thus the only potentially complex part of the earlier
construction becomes Hadamard--Toffoli implementable.

By \Cref{thm:int_approx}, using $m=O(\log(1/\eta))$ address bits
gives a block encoding $W$ whose distinguished block has the form
\begin{align}
K=aI+bJ
\end{align}
and satisfies, for $\delta:=\eta/20$,
\begin{align}
\left|
K-\frac12R_y(\theta)
\right|_{\mathrm{op}}
&\leq\delta.
\label{eqn:real-rotation-block-error}
\end{align}
All remaining steps in constructing $W$ are exactly those of
\Cref{sec:be_gate_synth}. The interval predicates are classical
reversible computations, and the only additional primitives are the
small Fan-Out gates used to evaluate them in parallel. In particular,
no further $T$-gates or complex single-qubit operations appear.

The oblivious-amplification step and its error analysis likewise carry
over unchanged. Writing $K=rV$, where $V$ is a real rotation and
$|r-\tfrac12|\leq\delta$, one round of amplification, followed by the sign correction $J^2=-I$, produces a clean isometry within $4\delta$ of $V$.Moreover,
\begin{align}
|V-R_y(\theta)|_{\mathrm{op}}
&\leq
\frac{2\delta}{1/2-\delta}
<
5\delta.
\end{align}
Hence the ideal-Fan-Out implementation has clean-input error less than
$9\delta<\eta/2$. The amplification reflections are
computational-basis reflections, implemented using $X$ and generalized
Toffoli gates, and the required overall minus sign is realized by
$J^2=-I$. Thus this stage also remains entirely within the real gate
set.

Finally, the construction uses only $O(m)$ Fan-Out occurrences of
width $O(m)$. As in \Cref{sec:be_gate_synth}, allocate total error
at most $\eta/2$ to their replacement and apply
\Cref{lem:uqac-logarithmic-fanout}. Since
$m=O(\log(1/\eta))$, this yields a constant-depth $\HQAC$
implementation with size and ancillae
$\operatorname{poly}(1/\eta)$. The Fan-Out replacements themselves
use only Hadamard and generalized Toffoli gates, so they do not
reintroduce $T$.

For $\HQNC$, use the bounded-arity implementation of the same
construction from \Cref{sec:be_gate_synth}. The copying, interval
predicates, and reflections are implemented by balanced CNOT and
bounded-Toffoli trees, giving depth
$O(\log m)=O(\log\log(4/\eta))$ and
$\operatorname{poly}(m)$ size and ancillae.

Finally, with the rightmost operation applied first,
\begin{align}
\Lambda(R_y(\theta))
&=
(I\otimes R_y(\theta/2))
\operatorname{CNOT}
(I\otimes R_y(-\theta/2))
\operatorname{CNOT}.
\label{eqn:controlled-ry-decomposition}
\end{align}
For control zero the two rotations cancel, while for control one
conjugation by $X$ reverses the second angle. Approximating the two
unconditional rotations to error $\eta/2$ each proves the controlled
case.
\end{proof}

\noindent This immediately removes the only complex single-qubit
gate from the encoded computation.

\begin{corollary}[Encoded $T$ Gate]
\label{cor:h-only-encoded-t}
For every $\eta\in(0,1/2)$, the Pauli-transfer action $O_T$ of the
$T$-gate admits a constant-depth $\HQAC$ implementation of
$\operatorname{poly}(1/\eta)$ size and ancillae with clean-input
operator error at most $\eta$. It also admits an $\HQNC$
implementation with depth $O(\log\log(4/\eta))$ and
$\operatorname{poly}(\log(1/\eta))$ size and ancillae.
\end{corollary}

\begin{proof}
By \Cref{prop:encoded-t},
$O_T=\Lambda(R_y(\pi/2))$. Apply
\Cref{lem:h-only-real-rotations} with $\theta=\pi/2$.
\end{proof}

\subsection{Fan-Out and Threshold Gates}
\label{sec:real-fanout-threshold}

We next extend the real simulation to circuits with Fan-Out and
Threshold gates. The two cases are handled differently. Fan-Out admits
a direct and exact implementation in the Pauli encoding. Because it is
Clifford, conjugation maps each Pauli operator to a single Pauli
operator up to sign, so its Pauli-transfer action is a signed
permutation of the Pauli labels. A general Threshold gate need not
have such a simple encoded action. Rather than realifying Threshold
gates directly, we reduce Threshold circuits to the Fan-Out-assisted
generalized-Toffoli model and then replace Fan-Out by a constant-depth
real Threshold construction.

We begin with Fan-Out. Let $F_m$ denote Fan-Out with one control and
$m$ targets. Its encoded action has the form
\begin{align}
O_{F_m}\ket{a,b}
&=
(-1)^{\phi(a,b)}
\ket{L_a(a),L_b(b)},
\label{eqn:encoded-fanout-structural-action}
\end{align}
where $L_a$ and $L_b$ are binary linear maps and $\phi$ is a Boolean
polynomial of degree at most three. Thus, despite the unbounded arity
of $F_m$, its Pauli-transfer matrix creates no superposition over
Pauli labels: it only permutes the labels and attaches a sign.

Both parts of \Cref{eqn:encoded-fanout-structural-action} admit shallow
implementations for simple structural reasons. The map $L_a$ is
Fan-Out itself, while the only nonlocal operation in $L_b$ is Parity,
which is Hadamard-conjugated Fan-Out. The phase $\phi$ contains only
$O(m^2)$ monomials, so after making enough copies of the participating
label bits, the corresponding CZ and CCZ phases can be applied to
disjoint registers in parallel. The following lemma makes this action
explicit.

\begin{lemma}[Exact Encoded Fan-Out]
\label{lem:real-encoded-fanout}
For Fan-Out $F_m$ with one control and $m$ targets, $O_{F_m}$
has an exact constant-depth $\HQNC_f$ implementation using
$O(m^2)$ gates and clean ancillae.
\end{lemma}

\begin{proof}
Index the Fan-Out control by $0$ and the targets by
$1,\ldots,m$. The linear maps in
\Cref{eqn:encoded-fanout-structural-action} are, for
$1\leq j\leq m$,
\begin{align}
a'_0&=a_0,
&
a'_j&=a_j\oplus a_0,
&
b'_0&=b_0\oplus\bigoplus_{j=1}^m b_j,
&
b'_j&=b_j.
\label{eqn:real-fanout-label-update}
\end{align}
Thus $L_a$ is exactly Fan-Out on the $a$ labels. The only nontrivial
part of $L_b$ computes the parity of the target $b$ labels into
$b_0$. Since Parity is Hadamard-conjugated Fan-Out
\cite{moore1999qac0}, the complete label permutation has an exact
constant-depth $\HQNC_f$ implementation.

It remains to determine the sign in
\Cref{eqn:encoded-fanout-structural-action}. One obtains
\begin{align}
\phi(a,b)
&=
a_0\left(
\sum_{j=1}^m a_jb_j
+(b_0+1)\sum_{j=1}^m b_j
+\sum_{1\leq i<j\leq m}b_ib_j
\right)
\pmod2.
\label{eqn:real-fanout-sign}
\end{align}
To verify this expression, let $L$ denote the binary linear
transformation implemented by $F_m$. Conjugation maps $X^aZ^b$ to
$X^{La}Z^{L^{-T}b}$. Returning to the Hermitian Pauli convention
contributes the factor
\begin{align}
i^{a\cdot b-(La)\cdot(L^{-T}b)},
\end{align}
and expanding this even exponent modulo four gives
\Cref{eqn:real-fanout-sign}.

The polynomial $\phi$ contains $O(m^2)$ monomials, each of degree at
most three. Before applying the label permutation, use Fan-Out to make
separate copies of the bits participating in these monomials. The
corresponding CZ and CCZ phases can then be applied to disjoint copies
in parallel, after which the copies are uncomputed. Since these phases
are diagonal, the copied bits are unchanged and all ancillae return to
zero. Finally, CCZ is Hadamard-conjugated ordinary Toffoli. Hence both
the sign and the label permutation are implemented exactly in constant
depth using $O(m^2)$ gates and clean ancillae.
\end{proof}

We now turn to Threshold circuits. Rather than construct the
Pauli-transfer action of a general Threshold gate directly, we reduce
to the Fan-Out-assisted model. The required Fan-Out implementation is
another specialization of the Grier--Morris nekomata construction
used earlier in \Cref{sec:qac-uqac-collapse}. The difference is that
we now choose the central Hamming slice, whose phase can be implemented
directly by Threshold gates.

\begin{fact}[Fan-Out from Hadamard and Threshold
{\cite[Theorem~19, Lemma~16, and Corollary~20]{grier2025quantum}}]
\label{fact:real-threshold-fanout}
For every $m\geq1$ and $\eta\in(0,1/2)$, Fan-Out on $m$
targets has a constant-depth $\HQTC$ implementation using
$\operatorname{poly}(m,1/\eta)$ gates and clean ancillae, with
clean-input operator error at most $\eta$.
\end{fact}

\begin{proof}
We use the same Grier--Morris nekomata construction and coherent
nekomata-to-Parity reduction as in
\Cref{sec:qac-uqac-collapse}, specifically \Cref{fact:nekomata-to-parity}. It therefore
suffices to verify that the nekomata preparation can be implemented
using only gates available in $\HQTC$ and with polynomial resources
at the required accuracy.

Take the central Hamming slice
\begin{align}
S_q
&:=
\{z:|z|=q/2\}
\end{align}
for sufficiently large even $q$. Its phase satisfies
\begin{align}
(-1)^{[|z|=q/2]}
&=
(-1)^{[|z|\geq q/2]}
(-1)^{[|z|\geq q/2+1]}.
\label{eqn:real-threshold-slice-phase}
\end{align}
Hence the required phase operation is implemented exactly by two
Threshold phase kickbacks. The remaining operations in the
Grier--Morris preparation are Hadamards, $X$ gates, and
generalized-Toffoli reflections, all of which are available in
$\HQTC$. Thus their construction gives, in constant depth and with
polynomial resources, an approximate $q$-nekomata of infidelity $\Delta_q=O(q^{-1/2})$.

For $q\geq m$, retain any $m$ marked qubits and regard the remaining
ones as ancillae. Applying
\Cref{fact:nekomata-to-parity} and then the standard
Parity--Fan-Out Hadamard conjugation gives a Fan-Out implementation
with clean-input operator error at most
$4\sqrt{2\Delta_q}=O(q^{-1/4})$. Choosing a sufficiently large even
\begin{align}
q
&=
O(m+\eta^{-4})
\end{align}
makes this error at most $\eta$. The resulting circuit has constant
depth and uses $\operatorname{poly}(m,1/\eta)$ gates and clean
ancillae.
\end{proof}

\noindent The preceding Fan-Out construction lets us avoid realifying a general
Threshold gate directly. In fact, it gives a slightly stronger
simulation than is needed. Even if Fan-Out is present as a primitive
in the original Threshold circuit, it can be removed in the final
real simulation.

\begin{corollary}[Real Threshold Simulation]
\label{cor:real-threshold-simulation}
Let $C$ be a depth-$d$ $\UQTC_f$ circuit, and let $s\geq2$ bound its
gate count and width. For every $\varepsilon\in(0,1/2)$, there is an
$\HQTC$ circuit $\widetilde C$ of depth $O(d+1)$ and size and
ancillae $\operatorname{poly}(s,1/\varepsilon)$ such that, for every
$x\in\{0,1\}^n$,
\begin{align}
|p_{\widetilde C}(x)-p_C(x)|
&\leq
\varepsilon.
\label{eqn:real-threshold-simulation}
\end{align}
In particular, the same conclusion holds for $C\in\UQTC$.
\end{corollary}

\begin{proof}
We track the simulation through the corresponding circuit classes. Start with the depth-$d$ circuit $C\in\UQTC_f$. Replace each
Threshold gate by its exact constant-depth $\QNC_f$ implementation
from \Cref{fact:takahashi-tani-simulation}, retaining the original
Fan-Out and single-qubit gates. This gives an equivalent $\QNC_f$
circuit of depth $O(d)$ and polynomial size. Thus the
Takahashi--Tani simulation may introduce additional Fan-Out gates,
but it preserves the implemented unitary exactly.

We do not next invoke the $\QNC_f$ universalization theorem, since its
additive $O(\log\log n)$ catalyst-preparation cost would not preserve
constant depth. Instead, we simply enlarge the allowed gate set and
regard the same circuit as a $\QAC_f$ circuit. No circuit
transformation occurs in this step. The advantage is that the
depth-preserving fixed-gate compiler for $\QAC_f$ now applies.
Using \Cref{lem:direct-unitary-compilation} with clean-input error $\varepsilon/8$, we obtain a $\UQAC_f$ circuit $C_{\mathrm{univ}}$ of depth $O(d)$ and size $\operatorname{poly}(s,1/\varepsilon)$. This compilation changes each acceptance probability by at most $\varepsilon/4$.

We now realify this circuit using the
$\UQAC_f\to\HQAC_f$ case of
\Cref{thm:h-only-depth-preserving-simulation}, with acceptance error
$\varepsilon/4$
and implement each encoded Fan-Out exactly, using
\Cref{lem:real-encoded-fanout}. This gives an $\HQAC_f$ circuit
$C_{\mathrm{real}}$ of depth $O(d)$ and polynomial size and ancillae
such that, combining the compilation and realification errors, for every input $x$,
\begin{align}
|p_{C_{\mathrm{real}}}(x)-p_C(x)|
&\leq
\frac{\varepsilon}{2}.
\label{eqn:threshold-realification-error}
\end{align}

It remains only to remove primitive Fan-Out. Every generalized
Toffoli gate is already a Threshold gate, so these gates may be
retained unchanged. Let $M$ denote the number of Fan-Out occurrences
in $C_{\mathrm{real}}$. If $M=0$, there is nothing further to do.
Otherwise, replace each Fan-Out using
\Cref{fact:real-threshold-fanout} with clean-input operator error
\begin{align}
\eta
&:=
\frac{\varepsilon}{4M}.
\label{eqn:threshold-fanout-local-error}
\end{align}
Using fresh clean ancillae for every occurrence, a telescoping hybrid
bounds the Euclidean distance between the output states of the ideal
$\HQAC_f$ circuit and the resulting $\HQTC$ circuit
$\widetilde C$ by
\begin{align}
M\eta
&=
\frac{\varepsilon}{4}.
\label{eqn:threshold-fanout-total-state-error}
\end{align}
For normalized states $\ket{\phi},\ket{\psi}$ and any measurement
effect $0\leq E\leq I$,
\begin{align}
\left|
\bra{\phi}E\ket{\phi}
-
\bra{\psi}E\ket{\psi}
\right|
&\leq
2|\ket{\phi}-\ket{\psi}|_2.
\end{align}
Hence the Fan-Out replacements change the acceptance probability by
at most
\begin{align}
|p_{\widetilde C}(x)-p_{C_{\mathrm{real}}}(x)|
&\leq
\frac{\varepsilon}{2}.
\label{eqn:threshold-fanout-acceptance-error}
\end{align}
Combining
\Cref{eqn:threshold-realification-error,eqn:threshold-fanout-acceptance-error}
gives
\begin{align}
|p_{\widetilde C}(x)-p_C(x)|
&\leq
\varepsilon
\end{align}
for every input $x$.

Finally, $M\leq\operatorname{poly}(s,1/\varepsilon)$, and every Fan-Out has at most $\operatorname{poly}(s,1/\varepsilon)$ targets. Thus
$\eta^{-1}=O(M/\varepsilon)=\operatorname{poly}(s,1/\varepsilon)$,
so \Cref{fact:real-threshold-fanout} gives polynomial size and
ancillae for every replacement. All replacements have constant-depth
and can be performed in parallel on disjoint gates within a common
layer. The final circuit therefore lies in $\HQTC$, has depth
$O(d+1)$ and size and ancillae
$\operatorname{poly}(s,1/\varepsilon)$, and satisfies the claimed
acceptance bound. Since $\UQTC\subseteq\UQTC_f$, the same simulation also applies to circuits without Fan-Out.
\end{proof}

\noindent
The reduction above is intentionally asymmetric. Its final step uses
general Threshold gates to eliminate Fan-Out, so it does not yield a
simulation of $\UQAC$ by $\HQAC$. Establishing that stronger
gate-restricted simulation requires implementing the encoded
generalized Toffoli directly without appealing to general Threshold
gates, which is the task of the next subsection.

\subsection{Encoded Generalized Toffoli Without Fan-Out}
\label{sec:encoded-toffoli-real}

The final ingredient in the real simulation is generalized Toffoli.
For the Clifford gates treated above, the Pauli-transfer action is a
signed permutation and can therefore be implemented directly.
Generalized Toffoli, however,  is qualitatively different. Namely, its Pauli-transfer
matrix mixes Pauli labels, and its arity may grow with the circuit
width, so fixed-dimensional synthesis does not apply. Our approach is
instead to identify its encoded action as a structured reflection and
implement that reflection by testing membership in its negative
eigenspace in constant depth. The construction proceeds via the following six steps:
\begin{enumerate}
    \item \textbf{Characterize the reflection.}
    In \Cref{sec:real-cz-reflection}, we analyze the generalized CZ gate,
    which is equivalent to generalized Toffoli up to Hadamard
    conjugation. Its Pauli-transfer action reduces to two conditions:
    local $X$-eigenvalue checks on inactive positions and agreement of
    the $Y$-eigenvalues on active positions.

    \item \textbf{Isolate candidate pairs.}
In \Cref{sec:real-bucket-primitives}, we use $W$ states to hash active
positions into buckets and exact-one tests to identify buckets
containing a unique active position. Pairing such singleton buckets
lets us search for disagreeing $Y$-labels without performing all
pairwise comparisons.

    \item \textbf{Access isolated positions.}
In \Cref{sec:real-onehot-swap}, we use the one-hot mask of a singleton
bucket to identify its unique active position, and show how to SWAP the
corresponding data qubit into an auxiliary register in constant depth.

    \item \textbf{Parallelize the agreement tests.}
In \Cref{sec:real-phase-label-sharing}, we address the fact that the
same $Y$-eigenvalue may need to participate in many hashing
experiments. Since $\ket{+i}$ and $\ket{-i}$ are orthogonal, we can
coherently encode this eigenvalue into several registers, allowing all
required comparisons to be performed in parallel on disjoint qubits.

\item \textbf{Assemble the encoded gate.}
In \Cref{sec:real-agreement-construction}, we combine the bucket tests,
one-hot extraction, and shared $Y$-labels to detect whether the active
positions disagree. This yields a constant-depth implementation of the
reflection, after which we bound the coherent error and compile away
the auxiliary rotations and small Fan-Out gates.

\item \textbf{Complete the real simulation.}
In \Cref{sec:real-simulation-completion}, we substitute these encoded
gate implementations into the original circuit layer-by-layer and
show that the resulting $\HQAC$ circuit preserves the acceptance
probability up to the desired error.
\end{enumerate}
\noindent For clarity, several intermediate constructions are first described
using exact real rotations and polylogarithmic-width Fan-Out. However, these are
only auxiliary resources and both are eliminated when the full encoded
gate is compiled into $\HQAC$.

\subsubsection{Identifying the Pauli-Transfer Reflection}
\label{sec:real-cz-reflection}

We first determine the operator that the circuit must implement. It is
convenient to work with the generalized CZ gate
\begin{align}
    \operatorname{CZ}_r
    &:= I-2\ket{1^r}\bra{1^r}.
    \label{eqn:generalized-cz}
\end{align}
Generalized Toffoli is obtained by conjugating one qubit of
$\operatorname{CZ}_r$ by Hadamards. Since the real-encoded Hadamard operation is
already available exactly, it suffices to implement
$O_{\operatorname{CZ}_r}$. The following proposition identifies its
action on the two Pauli-label registers. Write
$\ket{\pm i}:=(\ket0\pm i\ket1)/\sqrt2$ for the two
$Y$-eigenstates.

\begin{proposition}[Encoded Generalized CZ as a Reflection] \label{prop:encoded-generalized-cz}
For every
$a,b\in\{0,1\}^r$, the Pauli-transfer matrix of $\operatorname{CZ}_r$ satisfies
\begin{align}
    O_{\operatorname{CZ}_r}\ket{a,b}
    &=\ket a\otimes(I-2Q_a)\ket b,
    \label{eqn:encoded-cz-action}
\end{align}
where $Q_0:=0$. For $a\neq0$, let $S(a):=\{i:a_i=1\}$ and
define the two product states
\begin{align}
    \ket{v_a^\pm}
    &:=
    \bigotimes_{j\notin S(a)}\ket{-}_j
    \otimes
    \bigotimes_{j\in S(a)}\ket{\pm i}_j, \quad \text{such that}\quad  Q_a
    :=\ket{v_a^+}\bra{v_a^+}
       +\ket{v_a^-}\bra{v_a^-}.
    \label{eqn:encoded-cz-negative-projector}
\end{align}
Thus the first label $a$ is unchanged. For each $a\neq0$, the
operation negates the span of $\ket{v_a^+}$ and $\ket{v_a^-}$
and fixes its orthogonal complement.
\end{proposition}

\begin{proof}
For computational-basis strings $x,y\in\{0,1\}^r$, define $\chi(x):=(-1)^{[x=1^r]},$
so that $\operatorname{CZ}_r\ket x=\chi(x)\ket x$. Hence
conjugation acts diagonally on the computational-basis matrix units:
\begin{align}
    \operatorname{CZ}_r
    \ket x\!\bra y
    \operatorname{CZ}_r
    &=
    \chi(x)\chi(y)\ket x\!\bra y.
    \label{eqn:cz-matrix-unit-action}
\end{align}
The sign is negative exactly when one of $x,y$ equals $1^r$ and the
other does not. Thus, for each
$a\in\{0,1\}^r\setminus\{0^r\}$, define
\begin{align}
    E_a
    :=
    \ket{1^r\oplus a}\!\bra{1^r}
    \qquad \text{and} \qquad
    E_a^\dagger
    :=
    \ket{1^r}\!\bra{1^r\oplus a}.
    \label{eqn:cz-negative-matrix-units}
\end{align}
These are precisely the matrix units with eigenvalue $-1$ under
conjugation by $\operatorname{CZ}_r$. Their tensor-product form is determined directly by $a$. Namely,
\begin{align}
    E_a
    &=
    \bigotimes_{j:a_j=0}\ket1\bra1_j
    \otimes
    \bigotimes_{j:a_j=1}\ket0\bra1_j,
    \label{eqn:Ea-local-decomposition}\\
    E_a^\dagger
    &=
    \bigotimes_{j:a_j=0}\ket1\bra1_j
    \otimes
    \bigotimes_{j:a_j=1}\ket1\bra0_j.
    \label{eqn:Eadag-local-decomposition}
\end{align}
Indeed, when $a_j=0$, both strings have a $1$ in position $j$.
When $a_j=1$, the ket of $E_a$ has a $0$ while its bra has a $1$,
with the roles reversed for $E_a^\dagger$.

The operators $E_a,E_a^\dagger$, over all $a\neq0$, are distinct
computational-basis matrix units. They therefore form an orthonormal
basis of the negative eigenspace under the Hilbert--Schmidt inner
product. We now express this basis in Pauli coordinates. Extend the
Pauli encoding linearly to arbitrary operators by
\begin{align}
    \ket{A}_{\mathrm P}
    &:=
    2^{-r/2}
    \sum_{u,v\in\{0,1\}^r}
    \Tr\!\bigl(P(u,v)A\bigr)\ket{u,v}.
    \label{eqn:pauli-vectorization-general}
\end{align}
For one qubit,
\begin{align}
    \ket1\bra1&=\frac{I-Z}{2},&
    \ket0\bra1&=\frac{X+iY}{2},&
    \ket1\bra0&=\frac{X-iY}{2},
\end{align}
and hence their normalized Pauli vectors are
\begin{align}
    \ket{\ket1\bra1}_{\mathrm P}
    &=\ket0\ket{-},&
    \ket{\ket0\bra1}_{\mathrm P}
    &=\ket1\ket{+i},&
    \ket{\ket1\bra0}_{\mathrm P}
    &=\ket1\ket{-i}.
    \label{eqn:local-matrix-unit-pauli-vectors}
\end{align}
Thus a diagonal factor $\ket1\bra1$ contributes first Pauli label
$0$, while either off-diagonal factor contributes first Pauli label
$1$. Applying these identities tensor factor by tensor factor to
\Cref{eqn:Ea-local-decomposition,eqn:Eadag-local-decomposition} gives
\begin{align}
    \ket{E_a}_{\mathrm P}
    &=
    \ket a\otimes
    \left(
        \bigotimes_{j\notin S(a)}\ket{-}_j
        \otimes
        \bigotimes_{j\in S(a)}\ket{+i}_j
    \right)
    =
    \ket a\ket{v_a^+},
    \\
    \ket{E_a^\dagger}_{\mathrm P}
    &=
    \ket a\otimes
    \left(
        \bigotimes_{j\notin S(a)}\ket{-}_j
        \otimes
        \bigotimes_{j\in S(a)}\ket{-i}_j
    \right)
    =
    \ket a\ket{v_a^-}.
\end{align}

Pauli vectorization is an isometry for the Hilbert--Schmidt inner
product. Therefore, for each fixed $a\neq0$, the negative eigenspace
inside the block with first Pauli label $a$ is exactly
\[
    \ket a\otimes
    \operatorname{span}\{\ket{v_a^+},\ket{v_a^-}\}.
\]
Its projector on the second Pauli-label register is precisely
\[
    Q_a
    =
    \ket{v_a^+}\bra{v_a^+}
    +
    \ket{v_a^-}\bra{v_a^-}.
\]
Conjugation has eigenvalue $-1$ on this subspace and $+1$ on its
orthogonal complement, so the action in the $a$ block is
$I-2Q_a$. For $a=0$, no negative matrix unit occurs and the action
is the identity. Hence
\[
    O_{\operatorname{CZ}_r}\ket{a,b}
    =
    \ket a\otimes(I-2Q_a)\ket b,
\]
as claimed.
\end{proof}

For a fixed Pauli label $a\neq0$, let $S(a):=\{i:a_i=1\}$.
We call the positions in $S(a)$ \emph{active} and the remaining
positions \emph{inactive}. By \Cref{prop:encoded-generalized-cz}, the
encoded generalized CZ applies a minus sign exactly on the two product
states
\begin{align}
    \bigotimes_{j\notin S(a)}\ket{-}_j
    \otimes
    \bigotimes_{j\in S(a)}\ket{+i}_j \qquad
    \text{and} \qquad
    \bigotimes_{j\notin S(a)}\ket{-}_j
    \otimes
    \bigotimes_{j\in S(a)}\ket{-i}_j.
\end{align}
Thus membership in the negative eigenspace is characterized by two
conditions on the second Pauli-label register $b$:
\begin{enumerate}
    \item every inactive position is in the $X$-eigenstate $\ket{-}$;
    \item all active positions have the same $Y$-eigenvalue, either
    all $\ket{+i}$ or all $\ket{-i}$.
\end{enumerate}
The first condition is local and can be checked independently at each
inactive position. The second is the main global condition. Importantly,
we do not need to determine whether the common $Y$-eigenvalue is
$+1$ or $-1$. We only need to test whether all active positions agree,
while preserving coherent superpositions of the two valid cases.

To test agreement, it suffices to find one pair of active positions
with opposite $Y$-eigenvalues. We use the probabilistic strategy of the constant-ancillae Dicke state construction of Joshi and Vasconcelos~\cite{joshi2026constant}. Namely, active
positions are assigned to random buckets and the
$\operatorname{EXACT}_1$ function is used to identify buckets containing a unique
active position. We then pair these buckets and compare the
$Y$-eigenvalues of the two isolated positions. The following
elementary observation gives the probabilistic guarantee needed for
this agreement test.

\begin{lemma}[Paired-Bucket Disagreement Test]
\label{lem:paired-bucket-disagreement}
Suppose $m\geq2$ active positions carry labels in $\{+,-\}$, with
both labels present. Assign the positions independently and uniformly
to $2B$ buckets, grouped into the fixed pairs
$
    (1,2),(3,4),\ldots,(2B-1,2B).
$
If $m\leq B<2m$,
then with probability at least $1/32$, some paired buckets each
contain exactly one active position and those two positions have
opposite labels.
\end{lemma}

\begin{proof}
Fix an active position $v$ carrying a minority label, and let
$k\geq m/2$ be the number of active positions carrying the opposite
label. Condition on the bucket containing $v$. Consider the event
that exactly one of these $k$ opposite-label positions lands in the
paired bucket, while every other active position avoids both buckets.
Its probability is
\begin{align}
    \frac{k}{2B}
    \left(1-\frac1B\right)^{m-2}
    &\geq \frac1{32}.
    \label{eqn:paired-bucket-disagreement-probability}
\end{align}
Indeed, $k\geq m/2$ and $B<2m$ make the first factor at least
$1/8$, while $B\geq m$ makes the second at least $1/4$. On this
event, the two buckets are singletons carrying opposite labels.
\end{proof}

Thus, when the number of buckets is within a constant factor of the
number of active positions, a single assignment detects disagreement
with constant probability. Since the number of active positions is
unknown, we run in parallel over all bucket counts $B \in \{1,2,4,\ldots,2^{\lceil\log_2 r\rceil}\}$.
One of these choices is guaranteed to satisfy
$m\leq B<2m$, where $m$ is the number of active positions. We then
repeat the corresponding random assignment independently to make the
probability of missing a disagreement arbitrarily small. 

We next
recall from~\cite{joshi2026constant} the two ingredients used to carry
out this procedure: 1) computation of the $\operatorname{EXACT}_1$ function for identifying
buckets containing a unique active position  and 2) $W$ state synthesis for
assigning positions uniformly to buckets in superposition. We then
develop the two additional operations needed for the coherent
agreement test.

\subsubsection{Exact-One Tests and \texorpdfstring{$W$}{W}-State Preparation}
\label{sec:real-bucket-primitives}

The bucket construction uses two ingredients from
Joshi and Vasconcelos~\cite{joshi2026constant}. An
$\operatorname{EXACT}_1$ test identifies whether a bucket contains
exactly one active position, while a $W$ state assigns a position
uniformly to one of the available buckets in superposition. Their
$W$-state preparation uses continuously parameterized real rotations, so we compile these rotations into the $\HQAC$ gate set using the
real-rotation synthesis developed earlier.

\begin{fact}[Exact-One Test
{\cite[Corollary~2]{joshi2026constant}}]
\label{fact:real-exact-one}
For every $q\geq1$, the function $\operatorname{EXACT}_1(c)=[|c|=1]$
has an exact, constant-depth, polynomial-size clean reversible
implementation using $\HQAC$ gates together with $O(\log(q))$-Fan-Out.
\end{fact}

\begin{lemma}[One-Hot State Preparation]
\label{lem:real-w-state}
For every $q\geq1$ and $\delta\in(0,1/2)$, there is a
constant-depth circuit $P_{q,\delta}$, using $\HQAC$ gates and
$O(\log(q))$-Fan-Out, such that
\begin{align}
    \left\|
        P_{q,\delta}\ket{0^{q+A}}
        -\ket{W_q}\ket{0^A}
    \right\|_2
    &\leq\delta,
\end{align}
for the $q$-qubit $W$ state, $\ket{W_q}:=\frac1{\sqrt q}\sum_{j=1}^q\ket{e_j}$.
The gate and ancillae count $A$ are
polynomial in $q$ and $1/\delta$.
\end{lemma}

\begin{proof}
Joshi and Vasconcelos
\cite[Lemma~3 and Theorem~1]{joshi2026constant} give an exact
constant-depth preparation of $\ket{W_q}$ using
$\operatorname{EXACT}_1$ together with continuously parameterized
single-qubit rotations. These rotations may be chosen to be real
$R_y$ rotations.

Let $L_q=\operatorname{poly}(q)$ bound the number of real
single-qubit rotation occurrences, including those appearing during
uncomputation. By \Cref{lem:h-only-real-rotations}, each such
$R_y(\theta)$ can be approximated by a constant-depth $\HQAC$ circuit
with arbitrary prescribed clean-input error and polynomial gate and
ancilla cost. We apply this construction to every rotation with error
at most $\delta/\max\{1,L_q\},$
using fresh clean ancillae for each occurrence. A telescoping hybrid
over the $L_q$ replacements then gives total error at most $\delta$
on the complete output state, including the preparation workspace,
while preserving constant depth and polynomial resources. For $q=1$,
the state is prepared exactly by $X$.
\end{proof}

When constructing the complete agreement test below, it is convenient
first to reason with the ideal version of this preparation using exact
real rotations and the stated small Fan-Out gates. All such auxiliary
operations, including those used during uncomputation, are compiled
into the $\HQAC$ gate set only after the full circuit has been
assembled.

\subsubsection{One-Hot Controlled SWAP}
\label{sec:real-onehot-swap}

Suppose a bucket contains exactly one active position. The corresponding
mask then has a single $1$, identifying which data qubit belongs to that
bucket. We need to swap that qubit into an auxiliary register so that
it can be compared with the qubit selected from the paired bucket.

Let $c=(c_1,\ldots,c_q)$ denote the one-hot mask,
$d=(d_1,\ldots,d_q)$ the data qubits, and $t$ the auxiliary qubit.
If $c=e_j$, the desired operation is $\operatorname{SWAP}_{d_j,t}$.
Equivalently, on the one-hot subspace it is
\begin{align}
    \prod_{j=1}^q
    \operatorname{CSWAP}_{c_j}(t,d_j).
    \label{eqn:onehot-cswap-product}
\end{align}
These candidate controlled-SWAPs cannot simply be executed in
parallel because they all share the qubit $t$.

This overlap differs from the parallel controlled-SWAP construction
of Joshi and Vasconcelos
\cite[Corollary~3]{joshi2026constant}. There, a common control must
be distributed across many swaps, and unbounded Fan-Out is used to
supply those controls in parallel. Here the controls $c_j$ are
distinct, but the candidate swaps share a common target. We instead
use the one-hot promise to implement their combined action directly.

\begin{lemma}[Enabled One-Hot Controlled SWAP]
\label{lem:real-onehot-routing}
Let $e$ be an enable qubit, let
$c=(c_1,\ldots,c_q)$ be a selector register,
$d=(d_1,\ldots,d_q)$ a data register, and $t$ an auxiliary qubit.
There is an exact clean polynomial-size $\HQAC^0$
circuit satisfying, for every state $\ket\psi_{dt}$, $c\in\{0,1\}^q$, and $j\in[q]$,
\begin{align}
    \ket0_e\ket c\ket\psi_{dt}
    &\longmapsto
    \ket0_e\ket c\ket\psi_{dt},
    \label{eqn:onehot-swap-disabled}\\
    \ket1_e\ket{0^q}_c\ket\psi_{dt}
    &\longmapsto
    \ket1_e\ket{0^q}_c\ket\psi_{dt},
    \\
    \ket1_e\ket{e_j}_c\ket\psi_{dt}
    &\longmapsto
    \ket1_e\ket{e_j}_c
    \operatorname{SWAP}_{d_j,t}\ket\psi_{dt}.
    \label{eqn:onehot-swap-enabled}
\end{align}
Thus, when $e=0$ the circuit is exactly the identity for every
selector $c$. When $e=1$, for
every selector with $|c|=1$, swaps $t$ with the uniquely selected
data qubit.
\end{lemma}

\begin{proof}
We construct one selected CNOT, obtain its reverse by Hadamard
conjugation, and compose them using the usual three-CNOT SWAP
identity.

To begin, define the Boolean function
\begin{align}
    F(c,d)
    &:=
    \bigvee_{j=1}^q(c_j\wedge d_j).
    \label{eqn:onehot-selected-function}
\end{align}
When $c$ is one-hot, $F(c,d)$ simply returns the data bit selected by
$c$. Specifically, if $c=e_j$, then $F(c,d)=d_j$, while $F(0^q,d)=0$. Using this selected bit, define the update
\begin{align}
    A_e\ket{c,d,t}
    &:=
    \ket{c,d,t\oplus eF(c,d)}.
    \label{eqn:onehot-forward-update}
\end{align}
Thus $A_e$ flips the auxiliary qubit $t$ exactly when the enable bit
$e$ is one and the selected data bit is one. In particular, when
$e=1$ and $c=e_j$, $A_e$ is precisely
$\operatorname{CNOT}_{d_j\to t}$.
To implement $A_e$, compute the conjunctions $c_jd_j$ in parallel,
compute their OR into a flag, update $t$ controlled on this flag
and $e$, and uncompute. Generalized Toffoli and $X$ gates give an
exact clean implementation in constant depth with $O(q)$ gates
and ancillae. When $e=1$ and $c=e_j$, we have $F(c,d)=d_j$,
so $A_e$ acts as $\operatorname{CNOT}_{d_j\to t}$.

To complete the SWAP, we also need the CNOT in the opposite direction,
from the auxiliary qubit $t$ to the selected data qubit $d_j$. This
does not require a new selection procedure. Let $K:= H_d^{\otimes q}\otimes H_t$
and define
$B_e:=K A_e K$.
When $e=1$ and $c=e_j$, we have $A_e=\operatorname{CNOT}_{d_j\to t}$.
Conjugating a CNOT by Hadamards on its control and target reverses its
direction:
\begin{align}
    (H_{d_j}\otimes H_t)
    \operatorname{CNOT}_{d_j\to t}
    (H_{d_j}\otimes H_t)
    &=
    \operatorname{CNOT}_{t\to d_j}.
\end{align}
The Hadamards on the remaining data qubits cancel, so $B_e=\operatorname{CNOT}_{t\to d_j}$
whenever $e=1$ and $c=e_j$. Therefore, on this branch,
\begin{align}
    A_e B_e A_e
    &=
    \operatorname{CNOT}_{d_j\to t}
    \operatorname{CNOT}_{t\to d_j}
    \operatorname{CNOT}_{d_j\to t}=
    \operatorname{SWAP}_{d_j,t}.
\end{align}
\end{proof}

We will also use this construction to combine a one-hot bucket choice
with an activation bit. Suppose $h=e_j$ specifies the bucket chosen
by a position, and $a\in\{0,1\}$ records whether that position is
active. We need to produce a mask that equals $e_j$ when $a=1$ and
$0^q$ when $a=0$. For $|h|=1$, the opposite-direction selected CNOT above gives
exactly this transformation:
\begin{align}
    \ket a\ket h\ket{0^q}
    &\longmapsto
    \ket a\ket h\ket{ah},
    \label{eqn:onehot-mask-writing}
\end{align}
where $ah:=(ah_1,\ldots,ah_q)$. Indeed, taking $e=1$, $t=a$,
and $c=h$ in $B_e$ applies a CNOT from $a$ to the unique output
qubit selected by $h$. Thus, if $h=e_j$, the output is $e_j$ when
$a=1$ and $0^q$ when $a=0$. Importantly, this produces the entire
mask in constant depth without first making $q$ copies of $a$.

\subsubsection{Parallel Encoding and Comparison of \texorpdfstring{$Y$}{Y}-Labels}
\label{sec:real-phase-label-sharing}

The same active position may need to participate in several bucket
comparisons in parallel. If its $Y$-eigenvalue were stored in the
computational basis, we could simply copy that classical bit with
CNOT gates. Here the label is instead encoded by the orthogonal states
$\ket{+i}$ and $\ket{-i}$. Conceptually, we therefore want to copy
this binary label in the $Y$ basis, i.e. $\ket{\pm i}\longmapsto
    \ket{\pm i}^{\otimes M}$.
Because the two states are orthogonal, such a coherent repetition
encoding is allowed. The only subtlety is implementing it using real
gates, without first rotating through a complex $Y$-basis change.
The following lemma gives such an implementation.

\begin{lemma}[Repetition Encoding and Comparison of $Y$-Labels]
\label{lem:real-y-sharing}
\label{lem:real-small-width-sharing}
For $s\in\{0,1\}$, define
\begin{align}
    \ket{y_s}
    &:=
    \begin{cases}
        \ket{+i}, & s=0,\\
        \ket{-i}, & s=1.
    \end{cases}
\end{align}
For every $M\geq1$, the real isometry $V_M\ket{y_s}=\ket{y_s}^{\otimes M}$
has an exact polynomial-size implementation in $\HQNC^0$ with
Fan-Out of width $O(M)$. Moreover, there is an exact constant-depth $\HQNC$ circuit that,
given an enable bit $e$, two $Y$-eigenstates, and a comparison bit
$c$ (initialized to zero), implements
\begin{align}
    \ket e\ket{y_s}\ket{y_t}\ket0_c
    &\longmapsto
    \ket e\ket{y_s}\ket{y_t}
    \ket{e(s\oplus t)}_c,
    \label{eqn:enabled-y-comparison}
\end{align}
for every $e,s,t\in\{0,1\}$. Thus, when $e=1$, the comparison bit
records whether the two $Y$-eigenvalues disagree. When $e=0$
the circuit acts trivially.
\end{lemma}

\begin{proof}
We first construct the repetition encoding. Append $M-1$ zero
qubits, apply Hadamards to them, XOR their Parity into the original
qubit, and then apply
\begin{align}
    K_M
    &:=
    \prod_{i<j}\operatorname{CZ}_{i,j}.
\end{align}
On computational-basis input $b\in\{0,1\}$, this produces
\begin{align}
    \frac{1}{2^{(M-1)/2}}
    \sum_{\substack{z\in\{0,1\}^M\\|z|\equiv b\pmod2}}
        (-1)^{\binom{|z|}{2}}\ket z
    &=
    V_M\ket b.
    \label{eqn:small-sharing-explicit-isometry}
\end{align}
The equality follows by expanding
$\ket{+i}^{\otimes M}$ and $\ket{-i}^{\otimes M}$ in the
computational basis. Recall that Parity is simply Hadamard-conjugated
Fan-Out~\cite{moore1999qac0}. To implement $K_M$ in constant
depth, use Fan-Out to make enough computational-basis copies of each
qubit for the pairwise CZ gates, apply those CZ gates on disjoint
registers in parallel, and uncompute the copies. Each qubit requires
only $O(M)$ copies, so Fan-Out of width $O(M)$ suffices.

We now implement the enabled comparison. Recall that
$Y\ket{y_s}=(-1)^s\ket{y_s}$. Hence the two-qubit state
$\ket{y_s}\ket{y_t}$ is an eigenstate of $Y\otimes Y$ with
eigenvalue $(-1)^{s+t}=(-1)^{s\oplus t}$. Thus, determining whether
the two $Y$-eigenvalues agree amounts to recording this sign in a qubit. To do so, introduce a comparison qubit $c$ initialized to $\ket0$. First
consider the case in which the comparison is enabled. Apply $H_c$,
apply $Y\otimes Y$ controlled by $c$, and then apply $H_c$ again.
Since the data register is an eigenstate of $Y\otimes Y$, the action
is
\begin{align}
    \ket0_c\ket{y_s}\ket{y_t}
    &\longmapsto
    \frac{\ket0_c+\ket1_c}{\sqrt2}
    \ket{y_s}\ket{y_t}
    \\
    &\longmapsto
    \frac{
        \ket0_c+(-1)^{s\oplus t}\ket1_c
    }{\sqrt2}
    \ket{y_s}\ket{y_t}
    \\
    &\longmapsto
    \ket{s\oplus t}_c\ket{y_s}\ket{y_t}.
    \label{eqn:y-label-comparison}
\end{align}
Thus the comparison qubit is $0$ when the two $Y$-eigenvalues agree
and $1$ when they disagree, while the two data qubits are unchanged. To make this comparison conditional on the enable bit $e$, we apply
the middle controlled operation only when both $e=1$ and $c=1$.
When $e=0$, the middle operation is therefore the identity, and the
two Hadamards on $c$ cancel. When $e=1$,
\Cref{eqn:y-label-comparison} applies. Consequently, the complete
transformation is, as claimed,
\begin{align}
    \ket e\ket{y_s}\ket{y_t}\ket0_c
    &\longmapsto
    \ket e\ket{y_s}\ket{y_t}
    \ket{e(s\oplus t)}_c.
\end{align}

It remains to verify that this enabled comparison uses only real,
bounded-arity gates. Define $J:=XZ$. Since $J=-iY$, we have $Y\otimes Y=-J\otimes J$.
Therefore, conditioned on $e=c=1$, the required $Y\otimes Y$
operation can be implemented by applying $J$ to each data qubit
together with the additional minus sign. The two controlled $J$
operations decompose into constant-arity controlled-$X$ and
controlled-$Z$ gates, while the minus sign is supplied by
$\operatorname{CZ}_{e,c}$. All of these gates are real and of
bounded arity. Hence the enabled comparison has an exact
constant-depth implementation in $\HQNC$.
\end{proof}

\subsubsection{Constructing the Encoded Reflection}
\label{sec:real-agreement-construction}

We now combine the preceding constructions to implement the reflection
in \Cref{prop:encoded-generalized-cz}. Local checks enforce the
inactive conditions, while the paired-bucket test checks agreement
among the active $Y$-labels. We first analyze the circuit with ideal
auxiliary gates, then compile those gates into $\HQAC$.

\begin{theorem}[Encoded Generalized Toffoli]
\label{thm:real-encoded-toffoli}
\label{thm:constant-depth-real-encoded-toffoli}
For every $r\geq1$ and $\eta\in(0,1/2)$, there is a
constant-depth size-$\operatorname{poly}(r,1/\eta)$ $\HQAC$ circuit $\widetilde O$, using $A=\operatorname{poly}(r,1/\eta)$ clean ancillae, such that
\begin{align}
    \sup_{\|\phi\|_2=1}
    \left\|
        \widetilde O(\ket\phi\ket{0^A})
        -(O_{\mathrm{Toffoli}_r}\ket\phi)\ket{0^A}
    \right\|_2
    &\leq\eta.
    \label{eqn:real-encoded-toffoli-error}
\end{align}
Here $\mathrm{Toffoli}_r$ acts on $r$ qubits, with $r-1$ controls. For $r=1$ it is $X$.
\end{theorem}

\begin{proof}
It suffices to implement $O_{\operatorname{CZ}_r}$, since exact
encoded Hadamard conjugation then gives generalized Toffoli.
The case $r=1$ is immediate, so assume $r\geq2$.

\paragraph{Construction.}
We first describe the circuit assuming access to the exact real
rotations and small Fan-Out gates used in the preceding constructions.
These auxiliary operations will be compiled into $\HQAC$ at the end.

We begin with the local condition from
\Cref{prop:encoded-generalized-cz}. For each position $i$, apply $H$
to $b_i$, record whether $a_i=0$ and $b_i=0$, and then undo the
Hadamard. Since $\ket0$ in this basis is the $+1$ $X$-eigenstate,
these flags detect exactly the forbidden inactive positions. We also
compute a flag indicating whether $a\neq0$. It remains to test whether all active positions have the same
$Y$-eigenvalue. Let $L:=\lceil\log_2 r\rceil$, and for each
$\ell\in[L]$ set $B_\ell:=2^\ell$. We use
$R:=\lceil64\ln(4/\eta)\rceil$ independent repetitions at each bucket
count and set $M:=RL$, the total number of bucket assignments in which
each position participates. For every repetition $u\in[R]$, choice $\ell\in[L]$, and position
$i\in[r]$, prepare an ideal $\ket{W_{2B_\ell}}$. Its unique occupied
location assigns position $i$ uniformly to one of $2B_\ell$ buckets,
which we pair as $(1,2),(3,4),\ldots,(2B_\ell-1,2B_\ell).$
All bucket assignments remain coherent. No measurement is performed.

For each position $i$, the activation bit $a_i$ is used in every bucket
assignment $(u,\ell)$. We therefore use Fan-Out to create one private
copy of $a_i$ for each such assignment, allowing all corresponding
bucket-membership registers $c_{u,\ell,i,*}$ to be computed in
parallel on disjoint controls. Whenever an intermediate
bucket-membership or enable bit is needed by several tests in the same
layer, we similarly create the required private copies beforehand.
All of these copies are uncomputed when the construction is reversed.

We next ignore the bucket choices of inactive positions. Using
\Cref{eqn:onehot-mask-writing}, define
\begin{align}
    c_{u,\ell,i,j}
    &:=
    a_i[h_{u,\ell}(i)=j].
    \label{eqn:real-active-bucket-mask}
\end{align}
Thus, if position $i$ is active, the row $c_{u,\ell,i,*}$ is one-hot
and records its chosen bucket. If $i$ is inactive, the row is zero. The same $Y$-label may be needed in several bucket assignments, so
apply the repetition encoding $V_M$ from
\Cref{lem:real-y-sharing} to every $b_i$, reserving one encoded qubit
for each pair $(u,\ell)$. For a fixed $(u,\ell)$ and position $i$,
allocate one zero-initialized slot for each bucket and use
\Cref{lem:real-onehot-routing} to SWAP the corresponding encoded qubit
into the slot selected by $c_{u,\ell,i,*}$. After this step, different
buckets occupy disjoint registers.

For each bucket $j$, apply $\operatorname{EXACT}_1$ to the column
$c_{u,\ell,*,j}$ to determine whether that bucket contains exactly one
active position. For each paired pair of buckets, enable the next step
only if both exact-one tests succeed. In that case, use
\Cref{lem:real-onehot-routing} to extract the unique data qubit from
each bucket into a designated auxiliary qubit, and compare their
$Y$-eigenvalues using \Cref{eqn:enabled-y-comparison}. If either
bucket is not a singleton, the enabled SWAPs and comparison act
trivially. A comparison bit equal to one therefore witnesses two active positions
with opposite $Y$-eigenvalues. Different bucket pairs use disjoint
registers, and different $(u,\ell)$ assignments use different
repetition-encoded qubits, so all comparisons can be carried out in
parallel.

Finally, apply a minus sign exactly when $a\neq0$, every inactive
check passes, and every comparison bit is zero. One generalized CZ,
with the appropriate $X$ conjugations, implements this condition. We
then reverse the entire computation. The test ancilla workspace is erased exactly for each fixed bucket assignment. The coherent bucket registers are restored only approximately, with the cleanup error bounded below.
\paragraph{Error bound.}
Fix $a$ and use the $X$ basis at inactive positions and the $Y$
basis at active positions. For each fixed bucket assignment, the
circuit applies only a sign in this basis and uncomputes its test
ancillae. The sign can be wrong only when the inactive checks pass
but every repetition misses an existing disagreement. In this case,
$m:=|S(a)|\geq2$, and some chosen $B_\ell$ satisfies
$m\leq B_\ell<2m$. By \Cref{lem:paired-bucket-disagreement},
each repetition detects disagreement with probability at least
$1/32$, so missing it in all repetitions has probability at
most
\begin{align}
    p_{\mathrm{miss}}
    &:=(31/32)^R.
    \label{eqn:real-disagreement-miss-probability}
\end{align}

To extend this bound to arbitrary inputs, write
$\ket\phi=\sum_z\alpha_z\ket z$, where $z$ indexes these basis
states over all choices of $a$. Let $d_h(z)$ and $d_*(z)$
denote the implemented and desired signs for bucket assignment $h$.
After uncomputing the test but before unpreparing the bucket
registers, orthogonality gives squared error
\begin{align}
    \sum_z|\alpha_z|^2
    \sum_h\Pr[h]\,
        |d_h(z)-d_*(z)|^2
    &\leq4p_{\mathrm{miss}}.
\end{align}
The signs differ by magnitude 2 only on missed assignments.
Unpreparing the bucket registers preserves the norm, so the choice
$R=\lceil64\ln(4/\eta)\rceil$ gives
\begin{align}
    \sup_{\|\phi\|_2=1}
    \left\|
        C_{\mathrm{id}}(\ket\phi\ket{0^A})
        -(O_{\operatorname{CZ}_r}\ket\phi)\ket{0^A}
    \right\|_2
    &\leq2(31/32)^{R/2}
    \leq\frac{\eta}{2}.
    \label{eqn:real-coherent-agreement-bound}
\end{align}

\paragraph{Resources and compilation.}
Set $N:=\lceil(r+2)/\eta\rceil$. Since $R,L=O(\log N)$,
the repetition encoding and activation-bit copies require Fan-Out
of width $M=RL=O(\log^2N)$. The exact-one tests, mask copies, and
$W$-state preparations require only logarithmic-width Fan-Out.
The data slots occupy $rR\sum_{\ell=1}^L2B_\ell=O(r^2R)$
qubits, and all gate and ancilla counts are polynomial in $N$.
All repetitions, bucket counts, and bucket pairs run in parallel,
giving constant-depth with the ideal auxiliary gates.

Let $P\leq\operatorname{poly}(N)$ count the real-rotation and
Fan-Out occurrences, including those in uncomputation. By
\Cref{lem:h-only-real-rotations} and \Cref{cor:inverse-poly-polylog-fanout},
each can be replaced by a constant-depth $\HQAC$ circuit with
clean-input error $\rho:=\eta/(2\max\{1,P\})$ and polynomial
cost. Give every occurrence fresh clean ancillae, including inverse
occurrences. A telescoping hybrid bounds the total compilation error
by $P\rho\leq\eta/2$. It compares replacements on ideal
intermediate states and therefore does not require approximate
states to satisfy the one-hot promises.

Together with \Cref{eqn:real-coherent-agreement-bound}, this gives
total clean-input error at most $\eta$. Exact encoded Hadamard
conjugation converts the resulting implementation of
$O_{\operatorname{CZ}_r}$ into the required encoded generalized
Toffoli.
\end{proof}

\subsection{Completing the Real Circuit Simulation}
\label{sec:real-simulation-completion}

All required encoded gates are now available. Replacing them
layer-by-layer completes the real simulation, since the Pauli encoding
preserves disjoint gate supports.

\begin{proof}[Proof of \Cref{thm:h-only-depth-preserving-simulation}]
Let $C$ be a depth-$d$ $\UQAC$ circuit with gate count and width
at most $s$. By \Cref{lem:real-pauli-encoding}, its
computational-basis input and original ancillae can be Pauli-encoded
in constant depth. Exact replacement of every gate $U$ by $O_U$
would propagate the encoding and preserve the acceptance probability
under the local readout of that lemma.

Implement encoded Hadamard and Pauli gates exactly. For encoded $T$
and generalized Toffoli, apply
\Cref{cor:h-only-encoded-t}  and \Cref{thm:real-encoded-toffoli} with
clean-input error $\eta:=\varepsilon/(2s)$, as well as fresh clean
ancillae per occurrence. These guarantees remain valid on registers
entangled with the rest of the computation. Thus, if $\ket{\Psi_x}$
denotes the ideal encoded output, including zero-initialized
implementation ancillae, and $\ket*{\widetilde\Psi_x}$ the actual
output, a telescoping hybrid gives
\begin{align}
    \|\ket*{\widetilde\Psi_x}-\ket{\Psi_x}\|_2
    &\leq s\eta=\frac{\varepsilon}{2}.
\end{align}
Applying the encoded output effect $M_\Pi$, tensored with identity
on the implementation ancillae, yields
\begin{align}
    |p_{\widetilde C}(x)-p_C(x)|
    &\leq
    2\|\ket*{\widetilde\Psi_x}-\ket{\Psi_x}\|_2
    \leq\varepsilon.
\end{align}
Replacements of disjoint gates act on disjoint label registers and
use separate ancillae, so every original layer incurs only
constant-factor depth overhead. Including input preparation and
readout, the resulting $\HQAC$ circuit has depth $O(d+1)$ and
gate and ancilla counts $\operatorname{poly}(s,1/\varepsilon)$.
\end{proof}
\noindent The same simulation applies immediately to circuits with primitive
Fan-Out. In the target $\HQAC_f$ model, the auxiliary Fan-Out gates
used above may be retained as primitives, while each Fan-Out gate of
the original circuit is replaced exactly by its encoded implementation
from \Cref{lem:real-encoded-fanout}. Hence the depth and error bounds
of \Cref{thm:h-only-depth-preserving-simulation} are unchanged for
$\UQAC_f\to\HQAC_f$.
These two cases were proved independently of Threshold simulation.
The $\UQAC_f$ case can therefore be used in
\Cref{cor:real-threshold-simulation}, which gives the remaining
$\UQTC$ and $\UQTC_f$ cases of
\Cref{thm:h-only-depth-preserving-simulation}, with target
$\HQTC\subseteq\HQTC_f$.

\subsection{Bounded-Arity Circuits and Reusable Catalysts}
\label{sec:real-bounded-fanout}

The preceding $\UQAC\to\HQAC$ simulation does not preserve bounded-arity. Specifically, its implementation of the encoded generalized Toffoli gate itself requires the use of
generalized Toffoli gates. Thus, even when the source circuit lies in
$\UQNC$, that construction only yields an $\HQAC$ circuit. To obtain
a $\UQNC\to\HQNC$ simulation, we instead exploit the fact that every
multiqubit gate now has constant-arity.

In this setting, encoded Hadamard, Pauli, and Toffoli gates all admit
exact constant-depth implementations using bounded-arity real gates.
The only remaining difficulty is encoded $T$. Synthesizing its real
rotation separately in every circuit layer would introduce a
nonconstant overhead per layer. Instead, we prepare a bank of
Hadamard-eigenstate catalysts once and reuse the same states throughout
the computation. The resulting $O(\log\log((s+2)/\varepsilon))$ preparation
cost is therefore additive rather than multiplicative in the circuit
depth.

\begin{theorem}[Bounded-Arity Real Decision Simulation]
\label{thm:bounded-arity-real-decision-simulation}
Let $C$ be a depth-$d$ $\UQNC$ circuit, and let $s\geq2$ bound its gate count and width. For every $\varepsilon\in(0,1/2)$, there is an
$\HQNC$ circuit $\widetilde C$ satisfying, for every input $x$,
\begin{align}
    |p_{\widetilde C}(x)-p_C(x)|
    &\leq \varepsilon.
\end{align}
The circuit is depth-$O\!\left(d+\log\log(s/\varepsilon)\right)$
and size-$\operatorname{poly}(s,\log(1/\varepsilon))$.
The same statement holds for $\UQNC_f$ and $\HQNC_f$.
\end{theorem}

\begin{proof}
We use the Pauli encoding and readout of
\Cref{lem:real-pauli-encoding}. We first observe that all encoded
gates other than $T$ can be implemented exactly in constant depth
with bounded-arity real gates. Encoded Hadamard and Pauli gates are already exact. Consider next a
Toffoli gate $F$ of fixed arity $r$. Its Pauli-transfer matrix
$O_F$ has constant dimension, since $r=O(1)$, and its entries are
dyadic rationals. Exact synthesis of dyadic orthogonal matrices
therefore gives a constant-size circuit over
$\{H,X,\operatorname{CNOT},\operatorname{Toffoli}\}$ with clean
ancillae~\cite[Section~2 and Corollary~1]{amy2023improved}.
Hence every encoded bounded-arity Toffoli also has an exact
constant-depth $\HQNC$ implementation.

It remains to implement encoded $T$. Recall that
$O_T=\Lambda(HZ)=\Lambda(H)\operatorname{CZ}.$
The CZ factor is already available exactly, so it suffices to implement
controlled-$H$. For this, define the $+1$ Hadamard eigenstate
\begin{align}
    \ket{h_+}
    &:=
    \cos(\pi/8)\ket0+\sin(\pi/8)\ket1=R_y(\pi/4)\ket0,
    \label{eqn:real-hadamard-catalyst}
\end{align}
such that $ H\ket{h_+}=\ket{h_+}$.
Let $S_{c;t,a}$ denote a SWAP between a target qubit $t$ and a
catalyst qubit $a$, controlled by $c$. Then
\begin{align}
    S_{c;t,a}H_aS_{c;t,a}
    \bigl(\ket\varphi_{ct}\ket{h_+}_a\bigr)
    &=
    \bigl(\Lambda(H)\ket\varphi_{ct}\bigr)
    \ket{h_+}_a.
    \label{eqn:real-controlled-h-catalysis}
\end{align}
Indeed, when $c=0$, the Hadamard acts on its $+1$ eigenstate and
does nothing. When $c=1$, the controlled SWAPs move the target state
onto the catalyst wire, apply $H$, and swap it back. In either case,
the catalyst is returned exactly to $\ket{h_+}$. Since
controlled-SWAP has a constant-size bounded-Toffoli decomposition,
this gives an exact constant-depth implementation of $O_T$ once the
catalyst is available.

The key point is that the same catalyst can be reused. Let $w$ be the
maximum number of $T$ gates appearing in any single layer of $C$.
Since $\ket{h_+}=R_y(\pi/4)\ket0,$
we prepare $w$ approximate copies of $\ket{h_+}$ in parallel using
the bounded-arity real-rotation synthesis of
\Cref{lem:h-only-real-rotations}. We assign one catalyst to each $T$
gate in a layer. Since the ideal controlled-$H$ construction returns
each catalyst unchanged, the same bank can be reused in every
subsequent layer. Thus the nonconstant catalyst-preparation depth is
incurred only once, rather than once per circuit layer.

Using \Cref{lem:h-only-real-rotations}, prepare each catalyst with
joint state error $\eta := \varepsilon/(2w)$,
with no catalyst required when $w=0$. The entire catalyst bank is then
within $w\eta\leq\varepsilon/2$ of the ideal bank. The remainder of
the simulation applies the same unitary to the ideal and approximate
banks. Since unitary evolution preserves Euclidean distance, this
initial error does not grow when the catalysts are reused across
layers. Consequently, the final acceptance probability differs from
the ideal encoded computation by at most $\varepsilon$.

The $w$ catalysts are prepared in parallel, so their preparation depth
is $O(\log\log(w/\varepsilon))$. After preparation, each original gate has an exact constant-depth
encoded implementation, and gates from the same original layer can be
simulated in parallel using separate workspace. Hence each layer incurs
only constant-factor depth overhead. Since $w\leq s$, the total depth
is $O\!\left(d+\log\log(s/\varepsilon)\right)$, with size and ancillae count
$\operatorname{poly}(s,\log(1/\varepsilon))$.

For $\UQNC_f$, the same argument applies. Primitive Fan-Out remains
available in the target $\HQNC_f$ model, and each Fan-Out gate of the
original circuit has the exact encoded implementation from
\Cref{lem:real-encoded-fanout}. Thus the depth and error bounds are
unchanged.
\end{proof}

\subsection{Complexity Consequences}
\label{sec:real-hierarchy-consequences}

We now translate the preceding simulations into relations among the
standard, universal, and real shallow-depth hierarchies. Throughout
this subsection, the classes are \emph{bounded-error decision classes}:
$\mathcal K[D]$ denotes the languages decided with bounded error by
polynomial-size circuit families in model $\mathcal K$ of depth
$O(D(n))$, with one designated output qubit measured in the
computational basis. Some of the standard-to-universal simulations
proved earlier hold at the unitary level, but the realification
theorems need only preserve the acceptance probability. Accordingly,
all equalities below that involve a real class are statements about
decision power.

The wide-gate simulations incur only constant-factor depth overhead,
so they preserve every depth bound $D(n)\geq1$. The bounded-arity
simulation additionally requires a one-time
$O(\log\log n)$ catalyst-preparation depth, which is absorbed whenever
$D(n)=\Omega(\log\log n)$. When needed, we first amplify the original
bounded-error computation by a constant number of parallel
repetitions, leaving constant slack for the subsequent approximation.

We begin with the models containing generalized Toffoli, Fan-Out, or
Threshold gates, where the real simulation is depth preserving.

\begin{corollary}[Real Fan-Out and Threshold Hierarchies]
\label{cor:real-wide-hierarchies}
\label{cor:hqacf-constant-depth-collapse}
For every depth bound $D(n)\geq1$,
\begin{align}
    \QAC_f[D]
    &=
    \UQAC_f[D]
    =
    \HQAC_f[D]
    =
    \QTC[D]
    =
    \UQTC[D]
    \notag\\
    &=
    \HQTC[D]
    =
    \QTC_f[D]
    =
    \UQTC_f[D]
    =
    \HQTC_f[D].
    \label{eqn:real-wide-hierarchies}
\end{align}
\end{corollary}

\begin{proof}
The fixed-gate compilation of
\Cref{thm:qac-fixed-gate-simulation} extends to $\QAC_f$ by leaving
primitive Fan-Out gates unchanged. Applying the real simulation of
\Cref{thm:h-only-depth-preserving-simulation} then gives
\begin{align}
    \QAC_f[D]
    &=
    \UQAC_f[D]
    =
    \HQAC_f[D].
\end{align}
It remains to compare these classes with the Threshold models.
Generalized Toffoli is itself a Threshold gate, while
\Cref{fact:real-threshold-fanout} implements Fan-Out in $\HQTC$
with sufficiently small inverse-polynomial error. Hence
\begin{align}
    \HQAC_f[D]
    &\subseteq
    \HQTC[D].
\end{align}
Conversely, the Takahashi--Tani simulations of
\Cref{fact:takahashi-tani-simulation} give
$\QTC_f[D]=\QAC_f[D]$. Therefore
\begin{align}
    \HQAC_f[D]
    \subseteq
    \HQTC[D]
    \subseteq
    \UQTC[D]
    \subseteq
    \QTC[D]
    \subseteq
    \QTC_f[D]
    =
    \HQAC_f[D],
\end{align}
so every inclusion is an equality. The Fan-Out variants
$\UQTC_f[D]$ and $\HQTC_f[D]$ lie between the same endpoints and
are equal as well.
\end{proof}
\noindent Without primitive Fan-Out or Threshold gates, the generalized-Toffoli
simulation gives a separate collapse.

\begin{corollary}[Real Generalized-Toffoli Hierarchy]
\label{thm:uqac-hqac-collapse}
\label{cor:qac-uqac-hqac-collapse}
For every depth bound $D(n)\geq1$,
\begin{align}
    \QAC[D]
    &=
    \UQAC[D]
    =
    \HQAC[D].
    \label{eqn:uqac-hqac-collapse}
\end{align}
In particular,
$\QAC^k=\UQAC^k=\HQAC^k$ for every fixed integer $k\geq0$.
\end{corollary}

\begin{proof}
The equality $\QAC[D]=\UQAC[D]$ follows from
\Cref{cor:qac-uqac-collapse}. After constant-error amplification,
the $\UQAC\to\HQAC$ case of
\Cref{thm:h-only-depth-preserving-simulation} gives the second
equality with only constant-factor depth overhead. The reverse
inclusions follow from the gate sets.
\end{proof}

\noindent In particular, restricting the single-qubit gate set does not change
the constant-depth Parity problem:
\begin{align}
    \mathrm{Parity}\notin\QAC^0
    \quad\Longleftrightarrow\quad
    \mathrm{Parity}\notin\UQAC^0
    \quad\Longleftrightarrow\quad
    \mathrm{Parity}\notin\HQAC^0.
    \label{eq:parity-lower-bound-transfer}
\end{align}
Thus, for bounded-error decision lower bounds against $\QAC^0$, it
suffices to consider constant-depth circuits containing only $H$, $X$,
and generalized Toffoli gates.

We next turn to bounded-arity circuits. Here the real simulation has
an additive $O(\log\log n)$ catalyst-preparation cost, so it preserves
the asymptotic depth bound once $D(n)=\Omega(\log\log n)$.

\begin{corollary}[Real Bounded-Arity Hierarchies]
\label{cor:real-qnc-hierarchies}
For every depth bound $D(n)=\Omega(\log\log n)$,
\begin{align}
    \UQNC[D]
    =
    \HQNC[D]
    \quad
    \text{and} \quad
    \UQNC_f[D]
    =
    \HQNC_f[D].
    \label{eqn:real-qnc-hierarchies}
\end{align}
\end{corollary}

\begin{proof}
Apply \Cref{thm:bounded-arity-real-decision-simulation} after
constant-error amplification. For polynomial-size circuits and
constant target error, the catalyst bank requires
$O(\log\log n)$ preparation depth. Hence
\begin{align}
    O\!\left(D(n)+\log\log n\right)
    &=
    O(D(n))
\end{align}
whenever $D(n)=\Omega(\log\log n)$. The reverse inclusions follow
from the gate sets.
\end{proof}

There is also a separate constant-depth statement for bounded-arity
circuits \emph{without} Fan-Out. It does not follow from
\Cref{thm:bounded-arity-real-decision-simulation}, whose catalyst
preparation has nonconstant depth. Instead, it is a consequence of
the single-output decision convention.

\begin{proposition}[Constant-Depth Bounded-Arity Decision Classes]
\label{prop:real-qnc-zero-decision}
For single-output bounded-error decision computation,
\begin{align}
    \QNC^0
    &=
    \UQNC^0
    =
    \HQNC^0.
    \label{eqn:real-qnc-zero-decision}
\end{align}
\end{proposition}

\begin{proof}
Fix a constant-depth $\QNC^0$ circuit with bounded gate arity. The
backward light cone of its designated output qubit contains only
$O(1)$ input qubits. Consequently, at each input length, the
acceptance probability depends on only constantly many input bits. Because the circuit decides a language with bounded error, these
$O(1)$ bits determine a Boolean function
$f:\{0,1\}^{O(1)}\to\{0,1\}$. Since $f$ has constant-size truth table,
it can be computed exactly by a constant-size reversible circuit over
$X$, CNOT, and bounded Toffoli gates. Such a circuit lies in
$\HQNC^0$. Hence
$\QNC^0\subseteq\HQNC^0\subseteq\UQNC^0\subseteq\QNC^0$, and therefore
$\QNC^0=\UQNC^0=\HQNC^0$.
\end{proof}

We emphasize that
\Cref{eqn:real-qnc-zero-decision} is only an equality of
\emph{single-output decision classes}. It does not give a
constant-depth simulation of the full quantum state, unitary, or
joint output distribution of an arbitrary $\QNC^0$ circuit. By
contrast, \Cref{eqn:real-qnc-hierarchies} is obtained from the
explicit real simulation, but applies only once the available depth
can absorb the catalyst-preparation cost.

With Fan-Out, we may additionally combine the earlier universalization
of $\QNC_f$ with the bounded-arity real simulation. Both introduce
only an additive $O(\log\log n)$ depth cost. Thus, at doubly
logarithmic depth and above, the bounded-arity Fan-Out model joins the
wide-gate collapse.

\begin{corollary}[Full Real Fan-Out Hierarchy]
\label{cor:full-real-fanout-hierarchy}
For every depth bound $D(n)=\Omega(\log\log n)$,
\begin{align}
    \QNC_f[D]
    &=
    \UQNC_f[D]
    =
    \HQNC_f[D]=
    \QAC_f[D]
    =
    \UQAC_f[D]
    \\
    &=
    \HQAC_f[D]=
    \QTC_f[D]
    =
    \UQTC_f[D]
    =
    \HQTC_f[D].
    \label{eqn:full-real-fanout-hierarchy}
\end{align}
These classes are also equal to
$\QTC[D]$, $\UQTC[D]$, and $\HQTC[D]$.
\end{corollary}

\begin{proof}
By \Cref{cor:fanout-hierarchy-universalization},
$\QNC_f[D]=\UQNC_f[D]$. By
\Cref{cor:real-qnc-hierarchies}, this class also equals
$\HQNC_f[D]$. Finally,
\Cref{cor:real-wide-hierarchies} identifies the remaining Fan-Out and
Threshold classes. The additive $O(\log\log n)$ preparation costs
are absorbed by the assumed depth bound.
\end{proof}
\noindent Note that the argument above does not extend the $\QNC_f$ collapse to
constant depth. The simulation of
\Cref{thm:bounded-arity-real-decision-simulation} retains a
nonconstant catalyst-preparation cost. Moreover, the separate
constant-depth argument of \Cref{prop:real-qnc-zero-decision} no
longer applies, since primitive Fan-Out allows the backward light cone of a
single output qubit to contain polynomially many input bits. Thus, the
decision function need not reduce to a constant-size truth table. We therefore obtain only
\begin{align}
    \HQNC_f^0
    \subseteq
    \UQNC_f^0
    \subseteq
    \QNC_f^0
    =
    \HQAC_f^0
    =
    \HQTC^0,
    \label{eqn:real-constant-depth-boundary}
\end{align}
where the equalities on the right follow from
\Cref{fact:takahashi-tani-simulation} and \Cref{cor:real-wide-hierarchies}.
We leave open whether the first two inclusions are equalities, or
equivalently whether
\begin{align}
    \QNC_f^0
    &=
    \UQNC_f^0
    =
    \HQNC_f^0.
\end{align}
Unlike \Cref{eqn:real-qnc-zero-decision}, the single-output light-cone
argument does not apply here, since primitive Fan-Out allows the
backward light cone of one output qubit to contain polynomially many
input qubits even at constant depth.

\section{Structured Fourier Circuits and Forrelation Normal Forms}
\label{sec:forrelation-normal-forms}

Following Aaronson and Ambainis~\cite{aaronson2015forrelation},
Forrelation describes circuits that alternate global Hadamard
transforms with diagonal sign operations. For a Boolean function
$h:\{0,1\}^m\to\{0,1\}$, define
\begin{align}
    D_h\ket z&:=(-1)^{h(z)}\ket z.
    \label{eq:forrelation-phase-operation}
\end{align}
Given $k\geq1$ phase functions $h_1,\ldots,h_k$, the corresponding
$k$-fold Forrelation circuit is
\begin{align}
    W(h_1,\ldots,h_k)
    &:=H^{\otimes m}D_{h_k}H^{\otimes m}\cdots
       D_{h_1}H^{\otimes m},
    \label{eq:structured-forrelation-circuit}
\end{align}
with associated Forrelation quantity
\begin{align}
    \Phi_m(h_1,\ldots,h_k):=\bra{0^m}W(h_1,\ldots,h_k)\ket{0^m}=\frac{1}{2^{(k+1)m/2}}
      \sum_{z_1,\ldots,z_k\in\{0,1\}^m}
      (-1)^{\sum_{j=1}^k h_j(z_j)
                  +\sum_{j=1}^{k-1}z_j\cdot z_{j+1}}.
    \label{eq:structured-forrelation-scalar}
\end{align}

We derive restricted versions of this representation for shallow-depth
computation. First, we define a structured Fourier hierarchy and
identify the phase families corresponding to the real circuit models.
The results of the preceding sections then give decision-class
characterizations for standard and universal circuits at shallow depth. Second, we convert these structured Fourier circuits
to the full-Hadamard form above and express their acceptance bias
as a scalar Forrelation quantity. We show that these last two steps preserve the
phase restrictions and increase the number of layers by only a
constant factor.

\subsection{The Structured Fourier Hierarchy}
\label{sec:structured-forrelation-definitions}
\label{sec:forrelation-gate-phases}
\label{sec:forrelation-standard-hierarchy}

To relate phase functions to the gates in the shallow-depth hierarchy,
let $V_f$ compute a Boolean function of its controls into a target,
\begin{align}
    V_f\ket{z_S,z_t}
    &:=\ket{z_S,z_t\oplus f(z_S)}.
\end{align}
Since $HXH=Z$, conjugating the target by Hadamards gives
\begin{align}
    H_tV_fH_t\ket{z_S,z_t}
    &=(-1)^{z_t f(z_S)}\ket{z_S,z_t}.
    \label{eq:fourier-gate-phase-correspondence}
\end{align}
The associated phase function is therefore $z_t f(z_S)$, including
the target bit. Toffoli gives an AND phase on all participating bits,
while a Threshold gate gives its target bit times the Threshold
predicate. When several diagonal gates are applied in parallel, their
signs multiply, so their Boolean phase functions combine by XOR.

\begin{definition}[Structured Phase Function and Layer]
\label{def:structured-phase-layer}
Let $\mathcal G$ be a collection of permitted Boolean block
functions. A \textbf{structured phase function} on $m$ bits has
the form
\begin{align}
    h(z)
    &:=\alpha_0\oplus\bigoplus_{i=1}^m\alpha_i z_i
       \oplus\bigoplus_{j=1}^p g_j(z_{B_j}),
    \label{eq:structured-phase-function}
\end{align}
where $\alpha_0,\ldots,\alpha_m\in\{0,1\}$, $p\geq0$,
the blocks $B_1,\ldots,B_p\subseteq[m]$ are pairwise disjoint,
and each $g_j\in\mathcal G$ acts on $B_j$. The corresponding
operation $D_h$ from \Cref{eq:forrelation-phase-operation} is
called a \textbf{structured phase layer}.
\end{definition}
\noindent \noindent
The affine part
$\alpha_0\oplus\bigoplus_{i=1}^m\alpha_i z_i$
gives the overall sign $(-1)^{\alpha_0}$ and the single-qubit
phases $\prod_{i=1}^m Z_i^{\alpha_i}$, while the block part
$\bigoplus_{j=1}^p g_j(z_{B_j})$ gives the multiqubit phase
operations $\prod_{j=1}^p D_{g_j}$. Thus
\begin{align}
    D_h
    &=
    (-1)^{\alpha_0}
    \left(\prod_{i=1}^m Z_i^{\alpha_i}\right)
    \left(\prod_{j=1}^p D_{g_j}\right).
    \label{eq:structured-phase-layer-factorization}
\end{align}
Because the blocks $B_j$ are pairwise disjoint, the multiqubit phase
gates $D_{g_j}$ can be applied in parallel. The XOR in
\Cref{eq:structured-phase-function} simply records multiplication of
the corresponding signs, since
$(-1)^{a\oplus b}=(-1)^a(-1)^b$.

We group together all structured phase functions whose nonlinear
block functions are drawn from the same collection $\mathcal G$, as follows.

\begin{definition}[Structured Phase Family]
\label{def:structured-phase-families}
For a collection $\mathcal G$ of permitted block functions, let
$\operatorname{Phase}(\mathcal G)$ denote the family of all
structured phase functions over $\mathcal G$.
\end{definition}
\noindent In particular, we now specialize to the multi-qubit gates in the shallow-depth
hierarchy. Let $\mathcal G_3$, $\mathcal G_{\mathrm{AND}}$,
and $\mathcal G_{\mathrm{THR}}$ denote the block-function
classes arising from bounded Toffoli, generalized Toffoli, and
Threshold, respectively, as specified in
\Cref{tab:structured-phase-families}. We write
\begin{align}
    \mathcal A_3&:=\operatorname{Phase}(\mathcal G_3),
    \notag\\
    \mathcal A&:=\operatorname{Phase}(\mathcal G_{\mathrm{AND}}),
    \notag\\
    \mathcal T&:=\operatorname{Phase}(\mathcal G_{\mathrm{THR}}).
    \label{eq:base-structured-phase-families}
\end{align}
For bounded-arity, we take CNOT and ordinary three-qubit Toffoli as
the multiqubit primitives. Any larger fixed-arity Toffoli can be
reduced to these in constant depth using constantly many clean
ancillae per gate.

Fan-Out requires conjugating all targets rather than a single target.
Write $H_S:=\prod_{i\in S}H_i$ for Hadamards applied in parallel
on $S$, with identity elsewhere. For Fan-Out $F_{c;S}$ with
control $c$ and target set $S$,
\begin{align}
    H_SF_{c;S}H_S
    &=\prod_{i\in S}\operatorname{CZ}_{c,i}
      =D_{\,z_c(\bigoplus_{i\in S}z_i)}.
    \label{eq:fourier-fanout-phase-correspondence}
\end{align}
Thus, $B=\{c\}\mathbin{\dot\cup}S$ Fan-Out contributes the controlled-Parity block function
\begin{align}
    g_{\mathrm{FO}}(z_B)
    :=
    z_c\left(\bigoplus_{i\in S}z_i\right).
\end{align}
Although this phase factors into CZ gates
$\prod_{i\in S}\operatorname{CZ}_{c,i}$ that all share the control
qubit $c$, we regard their product as a single block supported on
$B$. The disjointness requirement therefore applies between distinct
Fan-Out blocks, not between the individual CZ gates within one such
block. Let $\mathcal G_{\mathrm{FO}}$ collect these
block functions. Adjoining them gives the Fan-Out variants
\begin{align}
    \mathcal A_{3,f}
    &:=\operatorname{Phase}(\mathcal G_3\cup\mathcal G_{\mathrm{FO}}),
    \notag\\
    \mathcal A_f
    &:=\operatorname{Phase}(\mathcal G_{\mathrm{AND}}\cup\mathcal G_{\mathrm{FO}}),
    \notag\\
    \mathcal T_f
    &:=\operatorname{Phase}(\mathcal G_{\mathrm{THR}}\cup\mathcal G_{\mathrm{FO}}).
    \label{eq:fanout-structured-phase-families}
\end{align}

Mirroring the inclusions among the gate sets, the phase families
satisfy
\begin{align}
    \mathcal A_3&\subseteq\mathcal A\subseteq\mathcal T,
    \notag\\
    \mathcal A_{3,f}&\subseteq\mathcal A_f\subseteq\mathcal T_f.
    \label{eq:structured-phase-family-inclusions}
\end{align}
Bounded-width AND is a special case of arbitrary-width AND, and an
AND on $B=S\mathbin{\dot\cup}\{t\}$ equals
$z_t[\sum_{i\in S}z_i\geq|S|]$. Moreover,
$\mathcal A_3\subseteq\mathcal A_{3,f}$,
$\mathcal A\subseteq\mathcal A_f$, and
$\mathcal T\subseteq\mathcal T_f$.

We now interleave structured phase layers with parallel Hadamard
layers.

\begin{definition}[Structured Fourier Circuit]
\label{def:structured-fourier-circuit}
A \textbf{structured Fourier circuit} on $m$ qubits has the form
\begin{align}
    U&:=H_{S_\ell}D_{h_\ell}\cdots
         H_{S_1}D_{h_1}H_{S_0},
    \label{eq:structured-fourier-circuit}
\end{align}
where $S_j\subseteq[m]$ and each $h_j$ is a structured phase
function. Each $D_{h_j}$ represents all the block phases in that
layer, not one individual gate. Empty Hadamard subsets and the
identity phase $D_0=I$ are allowed.
\end{definition}

To define decision classes, we use the same bounded-error,
single-output convention as in the standard shallow-depth circuit
hierarchy. On input $x\in\{0,1\}^n$, a circuit $U_n$ acts on
$\ket{u_x}:=\ket{x}\ket{0^{m-n}}$, and a designated output
qubit $\textnormal{\textsf{out}}$ is measured in the computational
basis. Writing
$\Pi_{\textnormal{\textsf{out}}}
:=(\ket1\bra1)_{\textnormal{\textsf{out}}}\otimes I$,
the acceptance probability is
\begin{align}
    p_{U_n}(x)
    &:=\left\|\Pi_{\textnormal{\textsf{out}}}
                  U_n\ket{u_x}\right\|_2^2.
    \label{eq:forrelation-one-bit-acceptance}
\end{align}
A circuit family decides a language $L$ if $p_{U_n}(x)\geq2/3$
for $x\in L$ and $p_{U_n}(x)\leq1/3$ for $x\notin L$.

\begin{definition}[Structured Fourier Decision Hierarchy]
\label{def:structured-fourier-hierarchy}
Let $D(n)\geq1$ specify an asymptotic layer bound as a function of
the input length, and let $\mathcal F$ be a family of structured
phase functions. The class $\mathsf{FH}_{\mathcal F}[D]$
consists of the languages decided by families $\{U_n\}$ of
structured Fourier circuits with polynomial width and total
description length, $\ell_n=O(D(n))$ phase layers, and
$h_{n,j}\in\mathcal F$ for every layer $j$.
\end{definition}
\noindent For fixed $k\geq1$, let $\mathsf{FH}_{\mathcal F}^{(k)}$
impose at most $k$ phase layers. Thus
\begin{align}
    \mathsf{FH}_{\mathcal F}[1]
    &=\bigcup_{k\geq1}\mathsf{FH}_{\mathcal F}^{(k)}
\end{align}
is the constant-layer class. The circuits and their designated
outputs depend on $n$, not on the particular input $x$. More generally, the parameter $D(n)$ bounds the Hadamard--phase alternations,
while $\mathcal F$ restricts the intervening phase functions. We
focus on $\mathcal A_3$, $\mathcal A$, $\mathcal T$, and
their Fan-Out variants because they retain the multi-qubit gate
restrictions of the standard shallow-depth hierarchy.

\paragraph{Correspondence with the real circuit hierarchy.}
For these six phase families, the correspondence holds at the level
of the full implemented unitary. The model--family pairs appear in
\Cref{tab:real-structured-fourier-correspondence}. Use $\operatorname{Unitary}(\mathcal C[D])$ with the approximation convention of \Cref{sec:hierarchy-gate-sets}, and use the analogous convention for structured Fourier circuits. The circuit translations below are exact, so they also preserve these approximation-based classes.

\begin{table}[t]
\centering
\small
\setlength{\tabcolsep}{5pt}
\renewcommand{\arraystretch}{1.35}
\begin{tabular}{@{}lccc@{}}
\toprule
& \textbf{Bounded Toffoli}
& \textbf{Generalized Toffoli}
& \textbf{Threshold}\\
\midrule
\textbf{Without Fan-Out}
& $\HQNC\leftrightarrow\mathcal A_3$
& $\HQAC\leftrightarrow\mathcal A$
& $\HQTC\leftrightarrow\mathcal T$\\
\textbf{With Fan-Out}
& $\HQNC_f\leftrightarrow\mathcal A_{3,f}$
& $\HQAC_f\leftrightarrow\mathcal A_f$
& $\HQTC_f\leftrightarrow\mathcal T_f$\\
\bottomrule
\end{tabular}
\caption{Real shallow-depth circuit models and their structured
phase families. Fan-Out contributes the controlled-Parity block
functions in $\mathcal G_{\mathrm{FO}}$.}
\label{tab:real-structured-fourier-correspondence}
\end{table}

\begin{lemma}[Exact Real--Structured Fourier Correspondence]
\label{lem:forrelation-diagonal-primitives}
For every model--family pair $(\mathcal R,\mathcal F)$ in
\Cref{tab:real-structured-fourier-correspondence} and every
$D(n)\geq1$,
\begin{align}
    \operatorname{Unitary}\bigl(\mathcal R[D]\bigr)
    &=\operatorname{Unitary}\bigl(\mathsf{FH}_{\mathcal F}[D]\bigr).
    \label{eq:real-structured-fourier-unitaries}
\end{align}
\end{lemma}
\noindent More precisely, every depth-$d$ circuit over the real gate set has
an exact structured Fourier representation on the same wires with
$O(d+1)$ phase layers. Conversely, every structured Fourier circuit
over $\mathcal F$ with $\ell$ phase layers has an exact
implementation over that gate set of depth $O(\ell+1)$. Both
conversions have polynomial cost in the original width and description
length.

\begin{proof}
\textbf{$\operatorname{Unitary}(\mathcal R[D])
\subseteq\operatorname{Unitary}(\mathsf{FH}_{\mathcal F}[D])$.}
Apply
\Cref{eq:fourier-gate-phase-correspondence,eq:fourier-fanout-phase-correspondence}
to the multiqubit gates in each original layer. Their supports are
disjoint, so their diagonal forms combine into one allowed structured
phase layer. Write $X=HZH$ for single-qubit bit flips and retain
the original Hadamards. Each original layer becomes constantly many
partial Hadamard and phase layers, giving $O(d+1)$ phase layers
overall. Insert identity layers as needed to maintain the alternating
form without merging phase layers whose blocks overlap.

\noindent\textbf{$\operatorname{Unitary}(\mathsf{FH}_{\mathcal F}[D])
\subseteq\operatorname{Unitary}(\mathcal R[D])$.}
Factor each structured phase layer as in
\Cref{eq:structured-phase-layer-factorization} and reverse the
same conjugation identities. The block phases are implemented by
the associated real gates in parallel on disjoint supports. Apply
the affine $Z$ terms separately using $Z=HXH$, and realize an
overall minus sign by $(XZ)^2=-I$ on one wire. Every structured
phase layer therefore has an exact constant-depth implementation,
as does every partial Hadamard layer. This gives depth $O(\ell+1)$.
Both directions preserve the unitary exactly and have polynomial cost.
\end{proof}

The decision-class equality follows by retaining the same input and
output qubit.

\begin{corollary}[Real and Structured Fourier Decision Classes]
\label{cor:real-structured-forrelation-classes}
For every $D(n)\geq1$ and every pair $(\mathcal R,\mathcal F)$
in \Cref{tab:real-structured-fourier-correspondence},
\begin{align}
    \mathcal R[D]&=\mathsf{FH}_{\mathcal F}[D].
    \label{eq:real-structured-forrelation-classes}
\end{align}
\end{corollary}

\begin{proof}
By \Cref{lem:forrelation-diagonal-primitives}, every depth-$O(D(n))$
circuit in $\mathcal R$ has an exact structured Fourier representation
over $\mathcal F$ with $O(D(n))$ phase layers and polynomial overhead.
Conversely, every structured Fourier circuit over $\mathcal F$ with
$O(D(n))$ phase layers has an exact depth-$O(D(n))$ implementation in
$\mathcal R$. In both directions, we retain the same computational-basis input and
designated output qubit. Since the conversion preserves the implemented
unitary exactly, it preserves the acceptance probability on every input
$x$. Hence a language is decided with bounded error by one model if and
only if it is decided with bounded error by the other. Therefore
$\mathcal R[D]=\mathsf{FH}_{\mathcal F}[D]$.
\end{proof}

\noindent Explicitly, without Fan-Out we have
$\HQNC[D]=\mathsf{FH}_{\mathcal A_3}[D]$,
$\HQAC[D]=\mathsf{FH}_{\mathcal A}[D]$, and
$\HQTC[D]=\mathsf{FH}_{\mathcal T}[D]$. With Fan-Out,
$\HQNC_f[D]=\mathsf{FH}_{\mathcal A_{3,f}}[D]$,
$\HQAC_f[D]=\mathsf{FH}_{\mathcal A_f}[D]$, and
$\HQTC_f[D]=\mathsf{FH}_{\mathcal T_f}[D]$.

\paragraph{Standard and universal circuits.}
For standard and universal circuits, we first apply the gate-set
simulations from the preceding sections to obtain real circuits with
nearly the same acceptance probabilities. Their structured Fourier
representations are then exact. Thus the following result preserves
the designated decision probability, not necessarily the original
complex unitary.

\begin{theorem}[Structured Fourier Representations of Standard Circuits]
\label{thm:standard-forrelation-normal-forms}
Let $C$ be a depth-$d$ circuit on $n$ input bits in $\QAC$,
$\QAC_f$, $\QTC$, or $\QTC_f$, with corresponding phase family
$\mathcal A$, $\mathcal A_f$, $\mathcal T$, or $\mathcal T_f$.
Let $s\geq2$ bound its gate count and width. For every
$\varepsilon\in(0,1/2)$, there is a structured Fourier circuit
$U$ over that family with $O(d+1)$ phase layers and width and
description length $\operatorname{poly}(s,1/\varepsilon)$ such that, for every $x\in\{0,1\}^n$,
\begin{align}
    |p_U(x)-p_C(x)|&\leq\varepsilon.
    \label{eq:standard-forrelation-acceptance}
\end{align}
The circuit $U$ and its designated output are independent of $x$. For circuits in $\UQNC$ and $\UQNC_f$, the corresponding phase
families are $\mathcal A_3$ and $\mathcal A_{3,f}$, respectively.
In both cases, the resulting structured Fourier circuit has $O\left(d+\log\log(s/\varepsilon)\right)$
phase layers and $\operatorname{poly}(s,1/\varepsilon)$ size and
width. The same bounds hold for circuits in $\QNC_f$, with phase
family $\mathcal A_{3,f}$.
\end{theorem}

\begin{proof}
For circuits in $\QAC$ and $\QAC_f$, first apply the fixed-gate
compilation and then
\Cref{thm:h-only-depth-preserving-simulation}, choosing the
approximation parameters so that the total change in acceptance
probability is at most $\varepsilon$. In the Fan-Out case, the
Fan-Out gates are retained exactly, as in the proof of
\Cref{cor:real-wide-hierarchies}. For circuits in $\QTC$ and
$\QTC_f$, use the corresponding Threshold and Fan-Out simulations
from the proof of \Cref{cor:real-wide-hierarchies} to obtain the
associated real circuit model, again with total acceptance error at
most $\varepsilon$. In each case, the resulting real circuit has
depth $O(d+1)$ and polynomial size. Applying
\Cref{lem:forrelation-diagonal-primitives} then gives the desired
structured Fourier representation exactly, with no additional error.

For circuits in $\UQNC$ and $\UQNC_f$, apply
\Cref{thm:bounded-arity-real-decision-simulation} with acceptance
error at most $\varepsilon$, followed by the exact structured Fourier
conversion. For circuits in $\QNC_f$, first apply
\Cref{thm:catalytic-universalization} and then
\Cref{thm:bounded-arity-real-decision-simulation}, allocating error
$\varepsilon/2$ to each step. Since the intermediate circuit size
remains polynomial in $s$ and $1/\varepsilon$, the resulting
structured Fourier circuit has
$O\!\left(d+\log\log(s/\varepsilon)\right)$ phase layers and
polynomial size.
\end{proof}

\noindent Together with the hierarchy equalities established earlier, this
gives the following characterization.

\begin{corollary}[Structured Fourier Characterizations]
\label{cor:standard-forrelation-hierarchy}
For every $D(n)\geq1$,
\begin{align}
    \mathsf{FH}_{\mathcal A}[D]
    &=\QAC[D]=\UQAC[D]=\HQAC[D].
    \label{eq:forrelation-qac-equality}
\end{align}
The Fan-Out and Threshold families satisfy
\begin{align}
    \mathsf{FH}_{\mathcal A_f}[D]
    &=\mathsf{FH}_{\mathcal T}[D]
      =\mathsf{FH}_{\mathcal T_f}[D]=\QAC_f[D]=\QTC[D]=\QTC_f[D].
    \label{eq:forrelation-wide-equality}
\end{align}
For $D(n)=\Omega(\log\log n)$,
\begin{align}
    \mathsf{FH}_{\mathcal A_3}[D]
    &=\UQNC[D]=\HQNC[D],
    \label{eq:forrelation-bounded-equality}\\
    \mathsf{FH}_{\mathcal A_{3,f}}[D]=\QNC_f[D]=\UQNC_f[D]&=\HQNC_f[D]=\mathsf{FH}_{\mathcal A_f}[D]
      =\mathsf{FH}_{\mathcal T}[D]
      =\mathsf{FH}_{\mathcal T_f}[D].
    \label{eq:forrelation-full-fanout-equality}
\end{align}
\end{corollary}

\begin{proof}
Combine the exact correspondence in
\Cref{cor:real-structured-forrelation-classes} with the hierarchy
equalities in
\Cref{cor:qac-uqac-hqac-collapse,cor:real-wide-hierarchies,cor:real-qnc-hierarchies,cor:full-real-fanout-hierarchy}.
The first two apply for every $D(n)\geq1$, and the latter two apply
when $D(n)=\Omega(\log\log n)$, giving the respective claims.
\end{proof}

Without Fan-Out, \Cref{prop:real-qnc-zero-decision} separately gives
$\QNC^0=\mathsf{FH}_{\mathcal A_3}[1]$ under single-output
readout. This uses the constant-size backward light cone and does
not simulate the original unitary. We do not establish
$\QNC[D]=\mathsf{FH}_{\mathcal A_3}[D]$ at arbitrary positive
depths, or $\QNC_f^0=\mathsf{FH}_{\mathcal A_{3,f}}[1]$.
The latter remains the constant-depth gate-set question discussed in
\Cref{sec:real-hierarchy-consequences}.

\subsection{From Structured Fourier Circuits to Forrelation}
\label{sec:forrelation-full-hadamards}
\label{sec:forrelation-decision-bias}

The previously described structured Fourier hierarchy characterizes shallow-depth
decision computation using the usual computational-basis input and
single-output measurement convention. We now connect this
characterization to the scalar Forrelation quantity of Aaronson and
Ambainis~\cite{aaronson2015forrelation}. Specifically, for each
classical input $x$, we construct a Forrelation instance whose scalar
value encodes the acceptance bias of the corresponding structured
Fourier circuit and, therefore, shallow-depth circuit.

For $k\geq1$ Boolean phase functions $h_1,\ldots,h_k$, define\begin{align}
    W(h_1,\ldots,h_k)
    &:=
    H^{\otimes m}D_{h_k}H^{\otimes m}\cdots
    D_{h_1}H^{\otimes m},
    \label{eq:structured-forrelation-circuit-recall}
\end{align}
with associated $k$-fold Forrelation quantity\begin{align}
    \Phi_m(h_1,\ldots,h_k)
    &:=
    \bra{0^m}W(h_1,\ldots,h_k)\ket{0^m}=
    \frac{1}{2^{(k+1)m/2}}
    \sum_{z_1,\ldots,z_k\in\{0,1\}^m}
    (-1)^{
        \sum_{j=1}^k h_j(z_j)
        +
        \sum_{j=1}^{k-1}z_j\cdot z_{j+1}
    }.
    \label{eq:structured-forrelation-scalar-recall}
\end{align}
Here $k$ is the number of phase functions and the dot products are
taken modulo $2$. The only additional ingredient is the Aaronson--Ambainis
partial-to-full Hadamard conversion, which lets us replace the
partial Hadamard layers of a structured Fourier circuit by full
transforms without changing the allowed phase family. After this
conversion, we absorb the classical input into the boundary phases
and obtain a scalar Forrelation value equal to the original
single-output acceptance bias.

\begin{proposition}[Structured Fourier Acceptance as Scalar Forrelation]
\label{prop:forrelation-acceptance-bias}
Let $U$ be a structured Fourier circuit on $q\geq n\geq1$ qubits with
$d$ phase layers in one of the six families $\mathcal F$ from
\Cref{tab:real-structured-fourier-correspondence}. Let $p_U(x)$
be its designated-output acceptance probability on input
$\ket{x}\ket{0^{q-n}}$. There are $q\leq m\leq q+1$ and,
for each $x\in\{0,1\}^n$, a list $\boldsymbol g_x$ of
$O(d+1)$ phase functions in $\mathcal F$ such that
\begin{align}
    \Phi_m(\boldsymbol g_x)&=1-2p_U(x).
    \label{eq:forrelation-acceptance-bias}
\end{align}
The lists have size polynomial in the size and width of $U$.
Only their first and last phase functions depend on $x$, through
addition of the linear function $\bigoplus_{i=1}^n x_i z_i$.
\end{proposition}

\begin{proof}
First apply the partial-to-full Hadamard conversion of Aaronson and
Ambainis~\cite[Section~6, proof of Theorem~25]{aaronson2015forrelation}
to $U$. After padding to odd width, their two-qubit gadget can be
applied in parallel to pairs in each even-sized Hadamard subset.
For an odd-sized subset, use its complement and one additional full
transform. The added phase layers consist of disjoint CZ gates and
belong to every family under consideration. Relabel the wires to
account for the swaps, and keep adjacent phase layers separate when their blocks overlap by inserting $H^{\otimes m}D_0H^{\otimes m}=I$ between them. This gives a full-Hadamard circuit
$W=W(h_1,\ldots,h_r)$ with $r\geq1$, $r=O(d+1)$, and
$h_j\in\mathcal F$,
whose relabeled output has the same acceptance probability as $U$.

For $u_x=(x,0^{m-n})$, define
$h_1^{(u_x)}(z):=h_1(z)\oplus\bigoplus_{i=1}^n x_i z_i$ and
$X^{u_x}:=\bigotimes_{i=1}^m X_i^{(u_x)_i}$. The identity
$H^{\otimes m}X^{u_x}=D_{u_x\cdot z}H^{\otimes m}$ gives
\begin{align}
    \widetilde W_x
    :=W(h_1^{(u_x)},h_2,\ldots,h_r)&=WX^{u_x}.
    \label{eq:forrelation-input-absorption}
\end{align}
Thus $\widetilde W_x\ket{0^m}=W\ket{u_x}$. Since
$Z_{\textnormal{\textsf{out}}}=I-2\Pi_{\textnormal{\textsf{out}}}$,
\begin{align}
    1-2p_U(x)
    &=\bra{0^m}\widetilde W_x^\dagger
       Z_{\textnormal{\textsf{out}}}\widetilde W_x\ket{0^m}.
    \label{eq:forrelation-bias-expectation}
\end{align}
Writing $h_{\mathrm{out}}(z):=z_{\textnormal{\textsf{out}}}$,
the operator in this matrix element has the phase list
\begin{align}
    \boldsymbol g_x
    &:=\bigl(h_1^{(u_x)},h_2,\ldots,h_r,h_{\mathrm{out}},
             h_r,\ldots,h_2,h_1^{(u_x)}\bigr).
    \label{eq:forrelation-echo-list}
\end{align}
Indeed, reversing $\widetilde W_x$ reverses its phase order, so
$W(\boldsymbol g_x)=\widetilde W_x^\dagger
Z_{\textnormal{\textsf{out}}}\widetilde W_x$.
For $r=1$, the list is
$(h_1^{(u_x)},h_{\mathrm{out}},h_1^{(u_x)})$.
The list contains $2r+1=O(d+1)$ phases. Its new terms are affine,
so every phase remains in $\mathcal F$. Taking the matrix element
between $\ket{0^m}$ states proves the claim.
\end{proof}

For each input $x$, this construction produces a scalar Forrelation
instance by incorporating $x$ into the first phase of a fixed circuit.
That phase appears twice in the final list because the construction
uses both the circuit and its inverse. All nonlinear blocks remain
independent of $x$.

\begin{corollary}[Constant-Gap Scalar Representation]
\label{cor:forrelation-scalar-normal-form}
Let $C$ be a depth-$d$ circuit on $n$ input bits in $\QAC$,
$\QAC_f$, $\QTC$, or $\QTC_f$, with associated phase family
$\mathcal A$, $\mathcal A_f$, $\mathcal T$, or $\mathcal T_f$,
respectively. Let $s\geq2$ bound its size and width. For every
$\varepsilon\in(0,1/2)$, each input $x$ determines a list
$\boldsymbol g_x$ of $O(d+1)$ phase functions in that family
such that
\begin{align}
    \left|\Phi_m(\boldsymbol g_x)-(1-2p_C(x))\right|
    &\leq2\varepsilon.
    \label{eq:approximate-forrelation-bias}
\end{align}
The width and phase-list size are
$\operatorname{poly}(s,1/\varepsilon)$. All lists are obtained
from one fixed list by adding $\bigoplus_{i=1}^n x_i z_i$ to
its first and last phases.
\end{corollary}

\begin{proof}
By \Cref{thm:standard-forrelation-normal-forms}, there is a fixed
structured Fourier circuit $U$ with
$|p_U(x)-p_C(x)|\leq\varepsilon$ for every $x$. Apply
\Cref{prop:forrelation-acceptance-bias} to $U$. Then
$\Phi_m(\boldsymbol g_x)=1-2p_U(x)$, so the error is
$2|p_U(x)-p_C(x)|\leq2\varepsilon$.
\end{proof}

For a bounded-error computation, $\varepsilon\leq1/12$ gives
$\Phi_m(\boldsymbol g_x)\leq-1/6$ on yes-inputs and
$\Phi_m(\boldsymbol g_x)\geq1/6$ on no-inputs. In particular,
a structured Fourier circuit $U$ computes Parity with error at most
$1/3$ exactly when its lists from
\Cref{prop:forrelation-acceptance-bias} satisfy, for every input $x \in \{0,1\}^n$,
\begin{align}
    (-1)^{\operatorname{Parity}(x)}\Phi_m(\boldsymbol g_x)
    &\geq\frac13.
    \label{eq:forrelation-parity-sign-condition}
\end{align}
Together with \Cref{eq:forrelation-qac-equality}, this expresses
the $\QAC^0$ Parity question as a constant-gap sign condition on
circuit-generated Forrelation quantities with disjoint AND blocks.
The circuit and its nonlinear blocks are fixed at each input length,
while the input enters only through the two boundary phase layers.

\bibliographystyle{alphaurl}
\bibliography{references}

\end{document}